\documentclass[12pt,reqno]{amsart}
\usepackage[left=1.25in,right=1.25in,top=1.25in,bottom=1.25in]{geometry}
\usepackage{color}
\usepackage[dvipsnames]{xcolor}
\usepackage{amsmath, amsthm, amssymb, amsfonts}
\usepackage{natbib}
\usepackage{har2nat}
\usepackage{tabularx}
\usepackage{parskip}
\usepackage[british]{babel}
\usepackage{mathtools}
\usepackage{xfrac}
\usepackage[colorlinks=true,allcolors=Blue,linkcolor=Blue]{hyperref}
\usepackage{dsfont}
\usepackage[onehalfspacing]{setspace}

\usepackage{tikz}
\usepackage{caption}
\usepackage{pgfplots}
\usepackage{float}
\usepackage{verbatim}
\restylefloat{table}
\usepgfplotslibrary{fillbetween}
\usetikzlibrary{patterns}
\pgfplotsset{compat=newest}
\usetikzlibrary{decorations.pathreplacing,angles,quotes}
\usepackage{graphicx}
\usepackage{subcaption}
\usepackage{threeparttable}
\allowdisplaybreaks
\newtheorem{theorem}{Theorem}
\newtheorem{corollary}{Corollary}
\newtheorem{manualtheoreminner}{Corollary}
\newenvironment{manualcorollary}[1]{%
  \renewcommand\themanualtheoreminner{#1}%
  \manualtheoreminner
}{\endmanualtheoreminner}

\newtheorem{lemma}{Lemma}

\newtheorem{proposition}{Proposition}

\newtheorem{remark}{Remark}

\theoremstyle{definition}

\newtheorem{assumption}{Assumption}
\newtheorem{example}{Example}

\newcommand{\ul}{\underline}
\newcommand{\ol}{\overline}
\newcommand{\df}{\mathrm{d}}

\newcommand{\bdis}{\begin{displaymath}}
\newcommand{\edis}{\end{displaymath}}
\newcommand{\beq}{\begin{equation}}
\newcommand{\eeq}{\end{equation}}
\newcommand{\bea}{\begin{eqnarray*}}
\newcommand{\eea}{\end{eqnarray*}}
\newcommand{\bean}{\begin{eqnarray}}
\newcommand{\eean}{\end{eqnarray}}

\newcommand{\R}{\mathbb{R}}
\newcommand{\N}{\mathbb{N}}
\newcommand{\E}{\mathbb{E}}

\DeclareMathOperator\supp{supp}

\DeclareMathOperator*{\argmax}{arg\,max}
\DeclareMathOperator*{\argmin}{arg\,min}
\newcommand{\1}{\mbox{\bf 1}}

\def \a {\alpha }
\def \b {\beta }
\def \c {\gamma }

\begin{document}
\title[Persuasion and Matching]{Persuasion and Matching: \linebreak Optimal Productive Transport}
\author[\uppercase{Kolotilin, Corrao, and Wolitzky}]{\larger \textsc{Anton Kolotilin, Roberto Corrao, and Alexander Wolitzky}}
\date{\today}
\thanks{\\
\textit{Kolotilin}: School of Economics, UNSW Business School. \\
\textit{Corrao} and \textit{Wolitzky}: Department of Economics, MIT.
\ \\
This paper supercedes earlier papers titled ``Persuasion with Non-Linear Preferences,'' ``Persuasion as Matching,'' and ``Assortative Information Disclosure.'' For helpful comments and suggestions, we thank Jak{\v{s}}a Cvitani{\'c}, Laura Doval, Piotr Dworczak, Jeffrey Ely, Drew Fudenberg, Emir Kamenica, Elliot Lipnowski, Stephen Morris, Paula Onuchic, Eran Shmaya, and Andriy Zapechelnyuk, as well as many seminar participants. We thank Daniel Clark and Yucheng Shang for excellent research assistance. Anton Kolotilin gratefully acknowledges support from the
Australian Research Council Discovery Early Career Research Award
DE160100964 and from MIT Sloan's Program on Innovation in Markets and
Organizations. Alexander Wolitzky gratefully acknowledges support from NSF CAREER Award 1555071 and Sloan Foundation Fellowship 2017-9633.}

\begin{abstract}
We consider general Bayesian persuasion problems where the receiver's utility is single-peaked in a one-dimensional action. We show that a signal that pools at most two states in each realization is always optimal, and that such \emph{pairwise} signals are the only solutions under a non-singularity condition (the \emph{twist condition}). Our core results provide conditions under which riskier prospects induce higher or lower actions, so that the induced action is \emph{single-dipped} or \emph{single-peaked} on each set of nested prospects. We also provide conditions for the optimality of either full disclosure or \emph{negative assortative disclosure}, where all prospects are nested. Methodologically, our results rely on novel duality and complementary slackness theorems. Our analysis extends to a general problem of assigning one-dimensional inputs to productive units, which we call \emph{optimal productive transport}. This problem covers additional applications including club economies (assigning workers to firms, or students to schools), robust option pricing (assigning future asset prices to price distributions), and partisan gerrymandering (assigning voters to districts).\\

\noindent\emph{JEL\ Classification:}\ C78, D82, D83 \\ 


\end{abstract}

\maketitle
\thispagestyle{empty}

\let\MakeUppercase\relax %

\newpage
\clearpage
\pagenumbering{arabic} 
\section{Introduction}

\label{s:intro}

Following the seminal papers of \citet{RS} and \citet{KG}, the past decade has witnessed an explosion of interest in the design of optimal information disclosure policies, or Bayesian persuasion. While significant progress has been made in the special case where the sender's and receiver's utilities are linear in the unknown state (e.g., \citealt{GK-RS}, \citealt{KMZL}, \citealt{Kolotilin2017}, \citealt{DM}, \citealt{KMS})---so that a distribution over states can be summarized by its mean---general results beyond this simple case remain scarce. The literature to date thus has little to say about the qualitative implications of economically natural curvature properties of utilities, or about the robustness of optimal disclosure patterns uncovered in the linear case when utilities are non-linear.

This paper studies persuasion with non-linear preferences as an instance of a general class of economic models that we call \emph{optimal productive transport}. In the persuasion context, we consider a standard setting with one sender and one receiver, where the receiver's action and the state of the world are both one-dimensional. We assume that the sender always prefers higher actions, the receiver prefers higher actions at higher states, and the receiver's expected utility is single-peaked in his action for any belief about the state. In this model, the receiver's action is optimal iff his expected marginal utility from increasing his action equals zero: that is, iff the receiver's first-order condition holds. This \emph{first-order approach} is key for tractability. We provide three types of results, all of which have general analogues beyond persuasion.

First, we show that it is always without loss to focus on \emph{pairwise signals}, where each induced posterior belief has at most binary support. Moreover, under a non-singularity condition on the sender's and receiver's utilities---which we call the \emph{twist condition}---every optimal signal is pairwise. 

Second, we ask when it is optimal for the sender to induce higher actions with riskier or safer prospects. In particular, when the sender pools two extreme states $x_1 < x_4$ and separately pools two moderate states $x_2 \leq x_3$ such that $x_1 < x_2 \leq x_3 < x_4$, do the extreme states induce a higher action---in which case we say that disclosure is \emph{single-dipped}, as the receiver's action is single-dipped on the set $\{x_1 ,x_2 ,x_3 ,x_4 \}$---or a lower action---in which case we say that disclosure is \emph{single-peaked}? This question turns out to be key for understanding optimal disclosure patterns with non-linear preferences. Our core results provide general conditions for the optimality of single-dipped disclosure (and, similarly, single-peaked disclosure). The conditions are based on the following simple idea. If disclosure is not single-dipped, then there must exist a \emph{single-peaked triple}: a pair of pooled states $x_1 <x_4$ and an intervening state $ x_2 \in ( x_1,  x_4)$ such that the induced action at $ x_2$ (say, action $y_2$) is greater than the induced action at $\{ x_1,  x_4\}$ (say, action $y_1$). Our conditions ensure that any single-peaked triple can be profitably perturbed in the direction of single-dippedness by shifting weight on $ x_1$ and $ x_4$ from $y_1$ to $y_2$, while shifting weight on $ x_2$ in the opposite direction. 

Third, we provide conditions for the optimality of either \emph{full disclosure}, where the state is always disclosed, or (more interestingly) \emph{negative assortative disclosure}, where all states are paired in a negatively assortative manner, so that all prospects can be ordered from safests to riskiest, and only a single state ``in the middle'' is disclosed. Intuitively, full disclosure and negative assortative disclosure represent the extremes of maximum disclosure (disclosing all states) and minimal pairwise disclosure (disclosing only one state). There is a unique full disclosure outcome, but there are many negative assortative disclosure outcomes, depending on the weights on the states in each pair. We further characterize the optimal negative assortative disclosure pattern as the solution of a pair of ordinary differential equations, and provide examples where these equations admit an explicit solution.

While this paper is mainly motivated by Bayesian persuasion, the theory we develop applies equally to several other applications. We consider three: club economies (e.g., assigning workers to firms to maximize output, or assigning students to schools to maximize welfare), robust option pricing (assigning future asset prices to price distributions to bound the price of a derivative), and partisan gerrymandering (assigning voters to districts to maximize expected seat share). To facilitate the analysis of these applications, in Section \ref{s:applications} we recast our model in terms of assigning general ``inputs'' to ``productive units.'' Table 1 explains how our general model maps to each of our applications.

\begin{table} 
\centering
\small
\begin{tabular}{|p{3.1cm}|p{2.4cm}|p{2.9cm}|p{2.7cm}|p{3.3cm}|}
\hline
\textbf{Application} & Input ($x$) & Productive Unit ($\mu$) & Output ($y$) & Meaning of Single-Dippedness \\
\hline
\textbf{Persuasion:} & state & posterior & receiver action & riskier prospects induce higher actions \\
\hline
\textbf{Worker-Firm Matching:} & worker with ability $x$ & firm & intra-firm spillover & diverse firms are more productive \\
\hline
\textbf{Student-School Matching:} & student with ability $x$ & school & peer effect & diverse schools are more desirable \\
\hline
\textbf{Option Pricing:} & period-2 asset price & period-2 asset price distribution & period-1 asset price & riskier assets are more expensive \\
\hline
\textbf{Gerrymandering:} & voter with partisanship $x$ & district & probability that designer's party wins the district & polarized districts are stronger \\
\hline
\end{tabular}
\caption{Atlas of Our Applications}
\end{table}

Mathematically, our model combines aspects of a production problem (combining inputs to produce output) and a transportation problem (matching inputs and outputs to generate utility)---hence the name ``optimal productive transport.'' The model is a new kind of optimal transport problem. Our key technical results are duality and complementary slackness theorems for this problem. The closest strand of the optimal transport literature is that on \emph{martingale optimal transport} (e.g., \citealt{BHP}, \citealt{GHT}, \citealt{BJ}), which we discuss in Section \ref{s:duality}.\footnote{A few recent papers apply optimal transport to persuasion, but these works are not very related to ours either methodologically or substantively. \citet{PRS} and \citet{LinLiu} consider limited sender commitment; \citet{ABS22} consider persuasion with multiple receivers; \citet{MCS} focus on the question of when optimal signals partition a multidimensional state space.} 

The optimal productive transport framework nests a great deal of prior work, both in persuasion and in the contexts of the other applications we cover. Some key prior works are \citet{RS}, \citet{GL}, and \citet{GS2018} (on persuasion); \citet{ArnottRowse} and \citet{SaintPaul} (on matching); \citet{BJ} (on option pricing); and \citet{FH} (on gerrymandering). In light of our analysis, some of the main results in these papers can be viewed as showing that single-dipped or single-peaked disclosure is optimal---that is, that riskier or safer prospects induce higher actions---in some special settings. For instance, \citeauthor{FH}'s \citeyear{FH}\textquotedblleft matching slices\textquotedblright\ gerrymandering
solution, where a gerrymanderer creates electoral districts that pool extreme
supporters with similarly extreme opponents, and wins those districts with
the most extreme supporters and opponents with the highest probability, is an example of single-dipped disclosure. \citeauthor{GL}'s \citeyear{GL} non-monotone stress tests, where a regulator designs a test that pools the weakest banks that it wants to receive funding with the strongest banks (and subsequently pools less weak banks with less strong ones), such that the weakest and strongest banks receive the highest funding, is another such example. On the other hand, \citeauthor{GS2018}'s \citeyear{GS2018} \textquotedblleft nested intervals\textquotedblright\ disclosure rule, where a designer pools favorable states with similarly
unfavorable states, and persuades the receiver to take her preferred action
with higher probability at more moderate states, is an example of single-peaked disclosure.

\textbf{Organization.} The paper is organized as follows. Section \ref{s:model} presents our model in the context of persuasion. Section \ref{s:optimality} formulates primal and dual versions of our problem and establishes strong duality and complementary slackness. Section \ref{s:pairwise} shows that pairwise signals are without loss. Sections \ref{s:assortative} and \ref{s:MNAD} present our main substantive results: Section \ref{s:assortative} provides conditions for single-dipped or single-peaked disclosure to be optimal, and Section \ref{s:MNAD} provides conditions for full disclosure or negative assortative disclosure to be optimal. Section \ref{s:applications} reframes our model as ``optimal productive transport,'' and applies it to matching, option pricing, and gerrymandering, as well as some specific persuasion problems. Section \ref{s:conclusion} concludes. Additional results, as well as all proofs, are deferred to the Appendix or Online Appendix

\section{Persuasion with Non-Linear Preferences}

For concreteness, we exposit our model and main results in the context of Bayesian persuasion. In Section \ref{s:general}, we rephrase the model as a general problem of assigning inputs to productive units, which we call \emph{optimal productive transport}. This more general framing covers our matching, option pricing, and gerrymandering applications.

\label{s:model}

\subsection{Model} 

We consider a standard persuasion problem, where a sender chooses a signal to reveal information to a receiver, who then takes an action. The sender's utility $V(y,x)$ and the receiver's utility $U(y,x)$ depend on the receiver's action $y\in Y:=[0,1]$ and the state of the world $x\in [0,1]$. The sender and receiver share a common prior $\phi \in \Delta([0,1])$, with support $X:=\supp (\phi)$.\footnote{Throughout, for any compact metric space $X$, $\Delta(X)$ denotes the set of Borel probability measures on $X$, endowed with the weak* topology. For any $\mu\in \Delta (X)$, its support $\supp (\mu)$ is the smallest compact set of measure one.} A \textit{signal} $\tau\in \Delta (\Delta (X))$ is a distribution over posterior beliefs $\mu \in \Delta (X)$ such that the average posterior equals the prior: $\int\mu \mathrm{d} \tau=\phi$ (\citealt{AM}, \citealt{KG}). An \emph{outcome} $\pi \in \Delta( Y \times X)$ is a joint distribution over actions and states. As we will see, it is equivalent to view the sender as choosing a signal $\tau$ (the \emph{signal-based problem}) or as directly choosing an outcome $\pi$ subject to an obedience constraint (the \emph{outcome-based problem}).

We impose four standard assumptions on preferences, which are similar to those in canonical unidimensional models of communication such as signaling (\citealt{Spence}), cheap talk (\citealt{CS82}), and hard information disclosure (\citealt{SeidmannWinter}). First, utilities are smooth.

\begin{assumption}
\label{a:smooth} $V(y,x)$ and $u(y,x):=\partial U(y,x)/\partial y$ are three times differentiable.
\end{assumption}
Apart from the receiver's marginal utility $u$, we denote partial derivatives with subscripts: e.g., $V_{y}(y,x)=\partial V(y,x)/\partial y$.

Second, the receiver's expected utility is single-peaked in his action for any posterior belief. 

\begin{assumption}
\label{a:qc} $u(y,x)$ satisfies \emph{strict aggregate single-crossing} in $y$: for all posteriors $\mu\in \Delta( X)$, 
\begin{equation*}
\int_{X} u(y,x)\df\mu(x)=0 \implies \int_{X} u_{y}(y,x)\df\mu (x)<0.
\end{equation*}
\end{assumption}

\citeauthor{Quah2012} (\citeyear*{Quah2012}) and \citeauthor{CS} (\citeyear*{CS}) characterized a weak version of aggregate single-crossing. We provide an analogous characterization of strict aggregate single-crossing in Appendix \ref{a:ad}.  A sufficient condition is strict monotonicity of $u$ (or equivalently strict concavity of $U$): i.e., $u_{y}(y,x)<0$ for all $(y,x)$. In fact, Appendix \ref{a:ad} shows that strict aggregate single-crossing is equivalent to strict monotonicity up to a normalization.

Third, the receiver's optimal action satisfies an interiority condition.\footnote{The substance of Assumption \ref{a:int} is that for each $x$, there exists $y$ such that $u(y,x)=0$. Note that it can never be optimal for the receiver to take any $y$ such that $u(y, x)$ has a constant sign for all $ x$. We can then remove all such $y$ from $Y$ and renormalize $Y$ to $[0,1]$, so that Assumption \ref{a:int} holds.}

\begin{assumption}
\label{a:int} $\min_{x \in [0,1]
}u(0,x )=\max_{x \in [0,1]}u(1,x)=0$.
\end{assumption}

The key implication of Assumptions \ref{a:smooth}--\ref{a:int} is that for any posterior $\mu$, the receiver's optimal action $\gamma(\mu ):= \argmax_{y\in [0,1]} \int U(y,x)\df \mu$ is unique and is characterized by the first-order condition 
\begin{equation}
\int_X u(\gamma(\mu ),x )\df\mu(x)=0.    \label{e:FOCy}
\end{equation} 
Our assumptions thus allow a ``first-order approach'' to the persuasion problem, similar to the approach of \citet{Mirrlees75} and \citet{Holmstrom79} to the classical moral hazard problem.\footnote{The first-order approach to persuasion was introduced in \citet{Kolotilin2017}.} %

Uniqueness of the receiver's optimal action implies that any signal $\tau$ induces a unique outcome $\pi _{\tau }$, and that we can define the sender's indirect utility from inducing posterior $\mu$ as 
\[W(\mu)=\int_X V(\gamma(\mu),x)\df \mu(x).\]

Fourth, the sender prefers higher actions, and the receiver's utility is supermodular.
\begin{assumption}\label{a:ord}
$V_y(y,x)>0$ and $u_{x}(y,x)>0$.
\end{assumption}

Together with Assumptions \ref{a:smooth}--\ref{a:int}, Assumption \ref{a:ord} ensures that for each action $y$ there is a unique state $\chi(y)$ such that $u(y,\chi(y))=0$ (i.e., the receiver's optimal action at $\chi(y)$ is $y$), and that $\chi(y)$ is a strictly increasing, continuous function with range $[0,1]$. 

A common interpretation of the receiver's action $y \in [0,1]$ is that the receiver has a private type and makes a binary choice---say, whether to accept or reject a proposal---and $y$ is the receiver's choice of a cutoff type below which he accepts. This interpretation is especially useful for some special cases of the model, as we see next.\footnote{To spell out this interpretation, let $g(t| x )$ be the conditional density of the receiver's type $t\in\lbrack 0,1]$ given the state $ x \in \lbrack 0,1]$. The sender's and receiver's utilities from rejection are normalized to zero. The sender's and receiver's utilities from acceptance are functions $\tilde{v}(t, x )$ and $\tilde{u}(t, x)$, with $\tilde{u}(t, x )g(t| x )$ satisfying Assumption \ref{a:qc}. For $y \in [0,1]$ (interpreted as the cutoff such that the receiver accepts iff $t \leq y$), we recover our model with $V(y , x )=\int_{0}^{y }\tilde{v}(t, x )g(t| x )\df t$ and $U(y , x )=\int_{0}^{y }\tilde{u}(t, x)g(t| x )\df t$.}

\subsection{Special Cases} \label{s:special}

We list some leading special cases of the model, which we return to periodically to illustrate our results.

(1) The \emph{linear case} (\citealt{GK-RS}): $u(y, x)= x-y$ and $V(y, x)=V(y)$. That is, $\gamma (\mu)=\mathbb{E}_{\mu}[ x]$ and $V$ is state-independent. This is the well-studied case where the sender's indirect utility $W(\mu)$ depends only on $\E_\mu[ x]$.

(2) The \emph{linear receiver case}: %
$u(y, x)= x-y$ but $V$ is arbitrary (e.g., possibly state-dependent). Here the receiver's preferences are as in the linear case, while the sender's preferences are general. 

(2a) The \emph{separable subcase} (\citealt{RS}): $V(y, x)=w( x)G(y)$ with $w> 0$ and $G> 0$. An interpretation is that the receiver has a private type $t$ with distribution $G$ and accepts a proposal iff $\mathbb{E}_{\mu}[ x]\geq t$, and the sender's utility when the proposal is accepted is $w( x)$.\footnote{\citeauthor{RS} focused on the sub-subcase with the uniform distribution $G(y)=y$. They assume that the state $(x,z )$ is two-dimensional, that the sender's and receiver's marginal utilities are $V_y(y, x ,z )=z $ and $u(y, x )= x -y$, and that there are finitely many states $( x ,z )$, so generically the sender's utility can be written as $V_y(y, x )=w( x)$. \citet{Rayo}, \citet{NP}, and \citet{OR} consider the separable subcase where $ x$ is continuous and $(x,z )$ is supported on the graph of $ x\rightarrow w( x)$. \citet{Rochetvila}, \citet{T18}, \citet{KX}, and \citet{DK} allow more general distributions of $(x,z)\in\R^2$.}

(2b) The \emph{translation-invariant subcase}: %
$V(y, x)=P(y- x)$. An interpretation is that the receiver ``values'' the proposal at $\mathbb{E}_{\mu}[ x]$, and the sender's utility depends on the amount by which the proposal is ``over-valued,'' $\mathbb{E}_{\mu}[ x]- x$. For example, a school may care about the extent to which its students are over- or under-placed. These preferences are similar to those in \citeauthor{GL}'s \citeyear{GL} model of stress tests (see Appendix \ref{s:GL}).

(3) The \emph{state-independent sender case}: %
$V(y, x)=V(y)$ but $u$ is arbitrary. Here the sender's preferences are as in the linear case, while the receiver's preferences are general.

(3a) The \emph{separable subcase}: $u(y, x)=I( x)( x-y)$, with $I>0$. This subcase extends the linear case by letting the receiver put more weight on some states than others.

(3b) The \emph{translation-invariant subcase}: $u(y, x)=T( x-y)$, with $T(0)=0$.
 An example that fits this subcase is that the sender's utility when the proposal is accepted is $1$, and accepting the proposal corresponds to the receiver undertaking a project that can either succeed or fail, where the receiver's payoff is $1-\kappa$ when the project succeeds and $-\kappa$ when it fails (and $0$ when it is not undertaken), with $\kappa\in (0,1)$. The difficulty of the project is $1- x$, the receiver's ability is $1-t$, the receiver's ``bad luck'' $\varepsilon$ has distribution $J$, and the project succeeds iff $1- x \leq 1-t -\varepsilon$, or equivalently $\varepsilon \leq  x-t$. This example fits the current subcase with $V$ equal to the distribution of $t$ and $T( x -y)=J( x -y)-\kappa$.

(3c) The \emph{quantile sub-subcase}: $u(y, x)=\1 \{ x\geq y\}-\kappa$, with $\kappa\in (0,1)$. This subcase corresponds to the previous example with $J( x -y)=\1 \{ x\geq y\}$, so the project succeeds iff the receiver's ability exceeds the project's difficulty. While $u$ is now discontinuous, we admit this subcase as a limit of the translation-invariant case. %
\citet{YangZentefis} also study the quantile sub-subcase.

\section{Optimality Conditions} \label{s:optimality}

This section establishes optimality conditions that form the basis for our analysis. Section \ref{s:duality} formulates signal-based and outcome-based primal and dual problems, and shows that they are equivalent. We will make use of both formulations. Section \ref{s:contact} establishes our key complementary slackness theorem. %

\subsection{Primal and Dual Programs}

\label{s:duality}

The sender's signal-based primal problem is to find a signal $\tau \in \Delta(\Delta(X))$ to
\begin{gather}
\text{maximize} \quad \int_{\Delta (X)} W(\mu) \df \tau (\mu)\tag{P}\label{PS}\\
\text{subject to}\quad \int_{\Delta(X)} \mu \df \tau (\mu) = \phi. \tag{BP}\label{PS1}
\end{gather}
Here, the primal constraint \eqref{PS1} is the usual \emph{Bayes plausibility} constraint \citep{KG}.

Next, let $L(X)$ denote the set of Lipschitz continuous functions on $X$. The signal-based dual problem is to find a \emph{price function} $p\in L(X)$ to
\begin{gather}
\text{minimize}\quad  \int_X p(x)\df \phi (x) \tag{D}\label{DS}\\
\text{subject to}\quad \int_X p(x)\df \mu(x)\geq W(\mu), \quad \text{for all $\mu\in \Delta(X)$.}  \tag{ZP}\label{DS1}	
\end{gather}
The interpretation is that $p(x)$ is the shadow price of state $x$, and the dual constraint \eqref{DS1} is the \emph{zero profit} condition that the sender's indirect utility from inducing any posterior $\mu$ cannot exceed the expectation of $p(x)$ under $\mu$. This interpretation will become clearer in the general framework of Section \ref{s:general}.

A preliminary result is that strong duality holds: (optimal) solutions to \eqref{PS} and \eqref{DS} exist and give the same value.
\begin{lemma} \label{l:dual} There exists $\tau\in \Delta(\Delta(X))$ that solves \eqref{PS}; there exists $p\in L(X)$ that solves \eqref{DS}; and the values of \eqref{PS} and \eqref{DS} are equal: for any solutions $\tau$ of \eqref{PS} and $p$ of \eqref{DS}, we have 
\begin{equation*}
\int_{\Delta(X) }W(\mu)\df\tau (\mu )=\int_{X}p(x)\df\phi (x).
\end{equation*}
\end{lemma}

Lemma \ref{l:dual} follows by showing that $W(\mu)$ is Lipschitz continuous and applying Corollary 2 of \citet{DK}, which in turn generalizes Theorem 2 of \citet{DM} from linear persuasion problems to non-linear ones.\footnote{Corollary 2 of \citet{DK} is proved using strong duality in an optimal transport problem, as in \citet{villani}. Further duality results for persuasion problems include those of \citet{DizdarKovac}, \citet{KX}, \citet{GLP}, and \citet{SY}.}

Next, the outcome-based primal problem is to find an outcome $\pi \in \Delta(Y \times X)$ to
\begin{gather}
\text{maximize}\quad  \int_{Y\times X }V(y,x)\df\pi (y,x) 
\tag{P'}\label{PO} \\
\text{subject to}\quad  \int_{Y\times \widetilde{X }}\df\pi (y,x)
=\int_{\widetilde{X}}\df\phi (x),\quad  \text{for all measurable }\widetilde{X }\subset X ,  \tag{BP'}\label{PO1} \\
 \int_{\widetilde{Y}\times X }u(y,x)\df\pi (y,x)=0,\quad 
\text{for all measurable } \widetilde{Y}\subset Y.  \tag{OB}\label{PO2}
\end{gather}

Here, \eqref{PO1} is an outcome-based version of Bayes plausibility, which says that the marginal of $\pi$ on $X$ equals the prior, $\phi$; and \eqref{PO2} is the \textit{obedience} constraint that the receiver's action at each posterior $\mu$ is $\gamma(\mu)$. A joint distribution $\pi$ that violates \eqref{PO2} is inconsistent with optimal play by the receiver, as there exists $\widetilde{Y}\subset Y$ such that the receiver's play is suboptimal conditional on the event $\{y \in \widetilde{Y}\}$. Conversely, for any joint distribution $\pi$ that satisfies \eqref{PO1} and \eqref{PO2}, if the sender designs a mechanism that draws $(y,x)$ according to $\pi$ and recommends action $y$ to the receiver, it is optimal for the receiver to obey the recommendation. We therefore say that an outcome $\pi$ is \emph{implementable} iff it satisfies \eqref{PO1} and \eqref{PO2}, and \emph{optimal} iff it solves \eqref{PO}.

Finally, letting $B(Y)$ denote the set of bounded, measurable functions on $Y$, the outcome-based dual problem is to find $p\in L(X)$ and $q\in B(Y)$ to
\begin{gather}
\text{minimize}\quad \int_{X}p(x)\df\phi (x)  \tag{D'}\label{DO} \\
\text{subject to}\quad p(x) \geq V(y,x)+q(y)u(y,x),\quad
\text{for all }(y,x)\in Y\times X.  \tag{ZP'}\label{DO1}
\end{gather}

The interpretation is that $p(x)$ is the shadow price of state $x$; $q(y)$ is the value of relaxing the obedience constraint at action $y$; and \eqref{DO1} says that $p(x)$ is no less than the sender's value from assigning state $x$ to any action $y$, where this value is the sum of the sender's utility, $V(y,x)$, and the product of $q(y)$ and the amount by which obedience at $y$ is relaxed when state $x$ is assigned to action $y$, $u(y,x)$.

We now establish that the feasible (and optimal) price functions in the signal-based and outcome-based formulations are the same. In particular, by Lemma \ref{l:dual}, strong duality holds in the outcome-based formulation, as well as the signal-based one.\footnote{Strong duality in the outcome-based formulation is established under slightly different assumptions in Lemmas 1 and 2 of \citet{Kolotilin2017}. However, a key step in the proof---that $q$ can be taken to be bounded---is incomplete in \citet{Kolotilin2017}.}

\begin{lemma}\label{l:Deq}
A price function $p\in L(X)$ is feasible (optimal) for \eqref{DS} iff there exists $q\in B(Y)$ such that $(p,q)$ is feasible (optimal) for \eqref{DO}.
\end{lemma}

\subsection{Complementary Slackness}
\label{s:contact}
Letting $p$ be the optimal price function (which we will see in Remark \ref{r:unique} is unique), define the set
\begin{gather}\label{e:L}
\Lambda=\left \{\mu\in \Delta(X):\int_X p(x)\df \mu(x)=W(\mu)\right\}.	
\end{gather}
Note that $\Lambda$ is compact, because $\int p(x)\df \mu$ and $W(\mu)$ are continuous in $\mu$.

By Lemma \ref{l:dual}, together with \eqref{PS1}, a signal $\tau$ is optimal iff
\[
\int_{\Delta(X)}\left(\int_X p(x)\df \mu(x)-W(\mu) \right)\df \tau (\mu)=0.
\]
Hence, since the integrand is non-negative by \eqref{DS1} and $\Lambda$ is compact, $\tau$ is optimal iff $\supp (\tau)\subset \Lambda$. Any posterior $\mu \notin \Lambda$ is thus excluded from the support of any optimal signal. In analogy with the optimal transport literature (e.g., Section 3 in \citealt{abs}), we refer to the set $\Lambda$ as the \emph{contact set}.

The following is our main technical result.

\begin{theorem}\label{t:contact}
There exists $q\in B(Y)$ such that

(1) $(p,q)$ is optimal for \eqref{DO}; 

(2) for all $\mu$ in $\Lambda$ (and, thus, in the support of any optimal signal $\tau$), we have
\begin{gather}\label{e:q}
q(\gamma(\mu))=-\frac{\int_X V_y(\gamma(\mu),x)\df\mu(x) }{\int_X u_y(\gamma(\mu),x)\df\mu(x)};
\end{gather} %

(3) for all non-degenerate $\mu$ in $\Lambda$ (and, thus, in the support of any optimal signal $\tau$), the function $q$ has derivative $q'(\gamma(\mu))$ at $\gamma(\mu)$ satisfying, for all $x\in \supp (\mu)$,
\begin{gather}\label{e:FOC}
V_y(\gamma(\mu),x)+q(\gamma(\mu))u_y(\gamma(\mu),x)+q'(\gamma(\mu))u(\gamma(\mu),x)=0.
\end{gather}
\end{theorem} 
 
Theorem \ref{t:contact} is our key tool for characterizing optimal signals. %
Intuitively, by complementary slackness, the support of any optimal outcome $\pi$ is contained in the set of points $(y,x)$ that satisfy \eqref{DO1}. Thus, if it is ever optimal to induce action $y$ at state $x$---i.e., if $y$ maximizes $V(y,x)+q(y)u(y,x)$---then $y$ must satisfy the first-order condition
\[
V_y(y,x)+ q(y)u_y(y,x)+q'(y)u(y,x)=0,
\]
which is just \eqref{e:FOC} with $\gamma(\mu)=y$. Moreover, taking the expectation of this equation with respect to posterior $\mu$ yields \eqref{e:q}. This equation simply says that $q(y)$ equals the product of the sender's expected marginal utility at $y$ and the rate at which $y$ increases as obedience is relaxed, where the latter term equals $-1/\E_{\mu}[u_y(y,x)]$ by the implicit function theorem applied to obedience.

As shown in Appendix \ref{a:r:unique}, another implication of Theorem \ref{t:contact} is:

\begin{remark}\label{r:unique}
There is a unique solution $p$ to (D).
\end{remark}

Lemmas \ref{l:dual}--\ref{l:Deq} and Theorem \ref{t:contact} can be compared to results in optimal transport. In standard optimal transport, two marginal distributions are given (e.g., of men and women, or workers and firms), and the problem is to find an optimal joint distribution with the given marginals. In our problem, the marginal distribution over states is given (by the prior $\phi$), and the problem is to find an optimal joint distribution with this marginal, where for each action the conditional distribution over states satisfies obedience. Strong duality and complementary slackness theorems are likewise key tools in optimal transport (e.g., \citealt{villani}, Theorem 5.10), but the relevent versions of these results differ from ours.\footnote{For example, in standard optimal transport, both dual variables appear in the dual objective function, and they are both uniquely determined.}

The most relevant strand of the optimal transport literature is that on \emph{martingale optimal transport (MOT)}. The MOT problem is to find an optimal joint distribution of two variables (say, $y$ and $x$) with given marginals, subject to the martingale constraint that the expectation of $x$ given $y$ is $y$. This problem coincides with our linear receiver case, but with an exogenously fixed distribution of the receiver's action. Motivated by problems in mathematical finance, \citet{BHP} (see also \citealt{BNT}) introduce MOT and prove that the primal and dual problems have the same value; however, they also show that their dual problem may not have a solution, unlike in our model with endogenous actions (or in standard optimal transport). Results in MOT also do not establish compactness of the contact set, which holds in our model as well as in standard optimal transport. Thus, MOT is related to our linear receiver case, but the endogenous action distribution apparently makes our model more tractable.\footnote{The MOT literature uses the contact set of an outcome-based dual problem. See \citet{KCW} for an alternative development of the results in the current paper that relies on the contact set of our outcome-based dual, \eqref{DO}. The approach in the current version, which is based on the contact set $\Lambda$ of the signal-based dual, \eqref{DS}, turns out to be simpler.}

\section{Pairwise Disclosure and the Twist Condition}

\label{s:pairwise}

Our first substantive result is that there is always an optimal signal that pools at most two states in every realized posterior, and that under an additional condition every optimal signal has this property. This result simplifies the persuasion problem to a generalized matching problem, where the sender chooses what pairs of states to match together, and with what weights.

Formally, a set of posteriors $M \subset \Delta(X)$ is \emph{pairwise} if $|\supp (\mu) |\leq 2$ for all $\mu \in M$. A signal $\tau$ is \emph{pairwise} if $\supp (\tau)$ is pairwise: that is, a pairwise signal induces posterior beliefs with at most binary support. For example, with a uniform prior, for any cutoff $\hat x \in [0,1]$ the signal that reveals states below the cutoff and pools each pair of states $x$ and $1+\hat x-x$ for $x \in [\hat x,(1+\hat x)/2]$ to induce posterior $\mu =\delta _{x}/2+\delta _{1+\hat x-x}/2$ is pairwise. The special case where $\hat x =1$ is full disclosure, which is also pairwise. In contrast, no disclosure, where $\tau(\phi)=1$, is not pairwise.\footnote{See Figure \ref{f:DP}. The ``disclose-pair'' pattern in Panel d.\ is reminiscent of this example, but with different weights on the states in each pair.}

If the receiver's utility is not quasi-concave, pairwise signals may be suboptimal. For example, suppose the sender rules three castles, one of which is undefended. The state $x$---the identity of the undefended castle---is uniformly distributed. Suppose the receiver can attack any two castles, and payoffs are $(-1,+1)$ for the sender and receiver, respectively, if the receiver attacks the undefended castle, and are $(+1,-1)$ otherwise. Then any pairwise signal narrows the set of possibly undefended castles to at most two, so the receiver always wins. But if the sender discloses nothing, the receiver wins only with probability $2/3$.\footnote{Pairwise signals are also suboptimal in the price-discrimination problem of \citeasnoun{BBM}, as well as in \citeasnoun{Brzutowski}, where $U(y,x)=\mathbf{1}\{ y \geq x \} -y$. In these models, the receiver's utility is not quasi-concave.}

In contrast, pairwise signals are without loss under Assumptions \ref{a:smooth}--\ref{a:int}. Moreover, equation \eqref{e:FOC} implies that if it is optimal to induce the same action $y$ at three states $x_1$, $x_2$, and $x_3$, then the vector $(V_y(y,x_1),V_y(y,x_2),V_y(y,x_3))$ must be a linear combination of the vectors $(u(y,x_1),u(y,x_2),u(y,x_3))$ and $(u_y(y,x_1),u_y(y,x_2),u_y(y,x_3))$. This observation gives a condition---which we call the \emph{twist condition}---under which pooling more than two states is suboptimal, so that every optimal signal is pairwise.\footnote{The term ``twist condition'' is in analogy to optimal transport, where the twist condition is an analogous non-singularity condition (e.g., Definition 1.16 in \citealt{santambrogio}).}

\textbf{Twist Condition} \emph{For any action $y$ and any triple of states $x_1<x_2<x_3$ such that $x_1<\chi(y)<x_3$, we have $|S| \neq 0$,\footnote{Here $|\cdot |$ denotes the determinant of a matrix; we also use the same notation for the cardinality of a set.} where}
\begin{equation}\label{e:sing}
S:=
\begin{pmatrix}
V_y(y,x_1) & V_y(y,x_2) & V_y(y,x_3)\\
u(y,x_1) &u(y,x_2) &u(y,x_3)\\
u_y(y,x_1) &u_y(y,x_2) &u_y(y,x_3)
\end{pmatrix}
.%
\end{equation}

We will apply this condition extensively in Section \ref{s:assortative}.

\begin{theorem}\label{p:S}
\label{t:pairwise} For any signal $\tau$ (whether optimal or not), there exists a pairwise signal $\hat{\tau}$ that induces the same outcome. Moreover, if the twist condition holds, then the contact set is pairwise, and hence so is any optimal signal.
\end{theorem}

The intuition for the first part of the theorem is that for any posterior $\mu$, there exists a hyperplane passing through it such that all posteriors on the hyperplane induce the same action, and the extreme points of the hyperplane in the simplex have at most binary support. Thus, any posterior that puts weight on more than two states can be split into posteriors with at most binary support without affecting the induced outcome. Figure 1 illustrates this argument for a posterior with weight on three states.

The intuition for the second part is that this splitting leaves an extra degree of freedom, which can be profitably exploited under the twist condition. Consider a posterior $\mu $ with $\supp (\mu)=\{x_1,x_2,x_3\}$. We can split $\mu$ into posteriors $\mu'$ and $\mu''$ with at most binary support that both induce action $\gamma (\mu)$. For example, suppose that $\supp (\mu')=\{x_1 ,x_2\}$ and $\supp(\mu'')=\{x_1 ,x_3\}$. Now consider a perturbation that moves probability mass $\varepsilon$ on $x_1 $ from $ \mu'$ to $\mu''$. This perturbation induces non-zero marginal changes in the action at $\mu'$ and $ \mu''$. Under the twist condition, these changes have a non-zero marginal effect on the planner's expected utility, by the implicit function theorem. Therefore, either this perturbation or the reverse perturbation, where $\varepsilon$ is replaced with $-\varepsilon$, is strictly profitable.\footnote{Formally, the second part of Theorem \ref{t:pairwise} directly follows from Theorem \ref{t:contact}.}

\begin{figure}[t]
\centering
\begin{tikzpicture}[scale=0.7]
\node[circle,fill=black,inner sep=0pt,minimum size=3pt,label=below:{$x_1$}] (a) at (0,0) {};
\node[circle,fill=black,inner sep=0pt,minimum size=3pt,label=above:{$x_2$}] (b) at (4,6.928) {};
\node[circle,fill=black,inner sep=0pt,minimum size=3pt,label=below:{$x_3$}] (c) at (8,0) {};
\draw (a) -- (b);
\draw (b) -- (c);
\draw (c) -- (a);
\node[circle,fill=black,inner sep=0pt,minimum size=3pt,label=left:{$\mu'$}] (d) at (2.664,4.616) {};
\node[circle,fill=black,inner sep=0pt,minimum size=3pt,label=below:{$\mu''$}] (e) at (3.2,0) {};
\draw (d) -- (e);
\node[circle,fill=black,inner sep=0pt,minimum size=3pt] (f) at (2.936,2.312) {};
\node (g) at (4,3) {$\leftarrow \gamma(\mu)$};
\node(h) at (2.6,2) {$\mu$};
\end{tikzpicture}
\caption{Pairwise Signals are Without Loss}
\caption*{\emph{Notes:} The optimal action at any posterior on the line between $\mu'$ and $\mu''$ equals $\gamma(\mu)$, so splitting $\mu$ into $\mu'$ and $\mu''$ eliminates a non-binary-support posterior without changing the outcome.}
\label{f:simplex}
\end{figure}

Prior results by \citet{RS}, \citet{AC}, and \citet{ZZ} also give conditions under which all optimal signals are pairwise. Theorem \ref{p:S} easily implies these earlier results.\footnote{Proposition 4 in \citet{AC} states that if $u(y,x)=x - y$ and there do not exist $\zeta\leq 0$ and $\iota\in \R $ such that $V_y(y,x_i)=\zeta x_i+\iota$ for $i=1,2,3$, then it is not optimal to induce action $y$ at states $x_1$, $x_2$, and $x_3$. This result is too strong as stated, and it is not correct unless $\zeta$ is also allowed to be positive. Theorem \ref{p:S} implies this corrected version of \citeauthor{AC}'s result.} Note that the twist condition always fails in the linear case, where $|S|=0$. Hence, in the linear case, Theorem \ref{p:S} never rules out pooling multiple states, and indeed pooling multiple states is often optimal (e.g., \citealt{KMZL}).\footnote{Of course, Theorem \ref{t:pairwise} shows that even when pooling multiple states is optimal, there also exists an optimal pairwise signal, where the ``multi-state pool'' is split into pairs. Conversely, if multiple posteriors all induce the same action, they can be pooled without affecting the outcome.} 

An immediate corollary of Theorem \ref{t:pairwise} is that no disclosure is generically suboptimal when there are at least three states, because for a fixed action $y$, a generic vector $(V_y(y,x))_{x \in X}$ with $|X|\geq 3$ coordinates cannot be expressed as a linear combination of two fixed vectors $(u(y,x))_{x \in X}$ and $(u_y(y,x))_{x \in X}$, as is required by \eqref{e:FOC}.

\begin{corollary}\label{c:nodisc}
For any prior $\phi$ with $|\supp(\phi)|\geq 3$ and any $u$, no disclosure is suboptimal for generic $V_y$.
\end{corollary}

Given \citeauthor{KG}'s concavification result, Corollary \ref{c:nodisc} implies that, for generic utilities, the sender's indirect utility is not concave in the posterior when there are more than two states. Note that Corollary \ref{c:nodisc} allows the case where $u$ and $V_y$ always have the opposite sign, so the sender's and receiver's ordinal preferences over actions are diametrically opposed. Hence, even in this case no disclosure is generically suboptimal.

\section{Single-Dipped and Single-Peaked Disclosure} \label{s:assortative}

The next two sections present our main results, which characterize optimal disclosure patterns. The current section asks when it is optimal for riskier or safer prospects to induce higher actions: that is, when optimal signals are ``single-dipped'' or ``single-peaked.'' As we will see, this is a key question, which unifies and generalizes much of what is known about special cases of the persuasion problem with non-linear preferences, as well as other models that fit our optimal productive transport framework.\footnote{In the MOT context, \citet{BJ} argue that single-dippedness/-peakedness are canonical properties analogous to positive/negative assortativity in standard matching models. Mathematically, positive/negative assortativity corresponds to monotonicity in the FOSD order, while single-dippedness/-peakedness corresponds to monotonicity in a variability order that depends on $u$; when $u(y,x)=x-y$, this variability order is the usual convex order.}

\subsection{Single-Dippedness/-Peakedness} A signal $\tau$ is \emph{single-dipped} (\emph{-peaked}) if for any $\mu_1,\mu_2\in \supp(\tau)$ such that $\supp(\mu_1)$ contains $x_1<x_3$ and $\supp (\mu_2)$ contains $x_2\in (x_1,x_3)$, we have $\gamma(\mu_1)\geq (\leq ) \gamma (\mu_2)$. Similarly, $\tau$ is \emph{strictly single-dipped} (\emph{-peaked}) if for any $\mu_1,\mu_2\in \supp(\tau)$ such that $\supp(\mu_1)$ contains $x_1<x_3$ and $\supp (\mu_2)$ contains $x_2\in (x_1,x_3)$ we have $\gamma(\mu_1)> (<)\gamma (\mu_2)$. We also apply these definitions to an arbitrary set of posteriors $M\subset \Delta(X)$ by replacing $\supp(\tau)$ with $M$ in the definitions. In particular, a pairwise signal is single-dipped if the induced receiver action is single-dipped on each set of nested pairs of states.

An equivalent definition is that a signal $\tau$ is single-dipped if it never induces a \emph{strictly single-peaked triple} $(y_1,x_1)$, $(y_2,x_2)$, $(y_1,x_3)$, with $x_1<x_2<x_3$ and $y_1<y_2$, in that there exist $\mu_1,\mu_2\in \supp (\tau)$ such that $x_1,x_3\in \supp (\mu_1)$ and $y_1=\gamma(\mu_1)$, and $x_2\in \supp (\mu_2)$ and $y_2=\gamma(\mu_2)$. (Otherwise, such a triple would witness a violation of single-dippedness.)

\begin{figure}
\centering
\begin{tabular}{cc}
 \begin{tikzpicture}[scale = 0.9]
	\small
	\begin{axis}		
		[axis x line = middle,
		axis y line = middle,
		xmin = 0, xmax = 1.1,
		ymin = 0, ymax = 1.1,
		xlabel=$x$,
		ylabel=$y$,
		xtick={0.001, 0.5, 1},
		xticklabels={$0$, $\frac{1}{2}$, $1$},
		ytick={0.001, 0.5, 1},
		yticklabels={$0$, $\frac{1}{2}$, $1$},
		clip=false]
		\draw [thick, solid, black]		(0, 0) -- (1,1);
		\draw [thin, dotted, black]			(0, 0) -- (1, 1);
	\end{axis}
\end{tikzpicture}
&\qquad \quad
  \begin{tikzpicture}[scale = 0.9]
	\small
	\begin{axis}		
		[axis x line = middle,
		axis y line = middle,
		xmin = 0, xmax = 1.1,
		ymin = 0, ymax = 1.1,
		xlabel=$x$,
		ylabel=$y$,
		xtick={0.001, 0.5, 1},
		xticklabels={$0$, $\frac{1}{2}$, $1$},
		ytick={0.001, 0.5, 1},
		yticklabels={$0$, $\frac{1}{2}$, $1$},
		clip=false]
		\draw [thick, solid, black]		(0, 0.5) -- (1,0.5);
		\draw [thin, dotted, black]			(0, 0) -- (1, 1);
		\draw [dashed, black]			(1/2, 0) -- (1/2, 1/2);
	\end{axis}
\end{tikzpicture}
\\
\small\textbf{a. Full Disclosure} &\qquad \small\textbf{b. No Disclosure} 
\bigskip
\\
\begin{tikzpicture}[scale = 0.9]
	\small
	\begin{axis}		
		[axis x line = middle,
		axis y line = middle,
		xmin = 0, xmax = 1.1,
		ymin = 0, ymax = 1.1,
		xlabel=$x$,
		ylabel=$y$,
		xtick={0.001, 1/3, 1},
		xticklabels={$0$, $\frac{1}{3}$, $1$},
		ytick={0.001, 1/3, 2/3,1},
		yticklabels={$0$, $\frac{1}{3}$,$\frac{2}{3}$, $1$},
		clip=false]
		
		\draw [thin, dotted, black]			(0, 0) -- (1, 1);
			\draw [thin, dotted, black]			(0, 2/3) -- (1, 2/3);

		\draw [thick, solid, black]		(0, 2/3) -- (1/3, 1/3);
		\draw [thick, solid, black]		(1/3,1/3) -- (1, 2/3);
		\draw [dashed, black]			(1, 0) -- (1, 2/3);
		\draw [dashed, black]			(1/3, 0) -- (1/3, 1/3);
		\draw [dashed, black]			(0,1/3) -- (1/3, 1/3);
	\end{axis}
\end{tikzpicture}
&\qquad \quad
 \begin{tikzpicture}[scale = 0.9]
	\small
	\begin{axis}		
		[axis x line = middle,
		axis y line = middle,
		xmin = 0, xmax = 1.1,
		ymin = 0, ymax = 1.1,
		xlabel=$x$,
		ylabel=$y$,
		xtick={0.001,1/3, 0.5, 1},
		xticklabels={$0$,$\frac{1}{3}$, $\frac{1}{2}$, $1$},
		ytick={0.001,1/3, 0.5, 5/6,1},
		yticklabels={$0$, $\frac{1}{3}$,$\frac{1}{2}$,$\frac{5}{6}$, $1$},
		clip=false]
		
		\draw [thin, dotted, black]			(0, 0) -- (1, 1);
		\draw [thin, dotted, black]			(0,5/6) -- (1, 5/6);

		\draw [thick, solid, black]		(0, 0) -- (1/3, 1/3);
		\draw [thick, solid, black]		(1/3, 5/6) -- (1/2, 1/2);
		\draw [thick, solid, black]		(1/2,1/2) -- (1, 5/6);
		\draw [thin, dotted, black]		(1/9, 1/9) -- (5/9, 5/9);
    	\draw [dashed, black]			(1/3, 0) -- (1/3, 5/6);
		\draw [dashed, black]			(1, 0) -- (1, 5/6);
		\draw [dashed, black]			(1/3, 0) -- (1/3, 1/3);
		\draw [dashed, black]			(0,1/3) -- (1/3, 1/3);
		\draw [dashed, black]			(1/2, 0) -- (1/2, 1/2);
		\draw [dashed, black]			(0,1/2) -- (1/2, 1/2);
	\end{axis}
\end{tikzpicture}\\
\small\textbf{c. Negative Assortative Disclosure } &\qquad \quad \quad \small\textbf{d. Disclose-Pair} 
\bigskip
\\

\begin{tikzpicture}[scale = 0.9]
	\small
		\begin{axis}		
		[axis x line = middle,
		axis y line = middle,
		xmin = 0, xmax = 1.1,
		ymin = 0, ymax = 1.1,
		xlabel=$x$,
		ylabel=$y$,
		xtick={0.001,2/5,1},
		xticklabels={$0$,$\frac{2}{5}$,$1$},
		ytick={ 0.001,1/4, 7/10,9/10,1},
		yticklabels={$0$,,,,$1$},
		clip=false]
		\draw [thin, dotted, black]			(0, 0) -- (1, 1);

		\draw [black, thick, smooth, tension=3] (1/8, 1/8) to [out=10,in=240] (1/3, 1/4);
        \draw [black,thick, smooth, tension=3] (0, 1/4) to [out=300,in=170] (1/8, 1/8);
		
		\draw [thick, solid, black]		(1/2.5, 1/2.5) -- (3/5, 3/5);
		
		\draw [black, thick, smooth, tension=3] (3/5, 3/5) to [out=0,in=240] (4/5, 3.5/5);
		
		\draw [black, thick, smooth, tension=3] (3/5, 3/5) to [out=180,in=310] (1/2.5, 3.5/5);
		
 		\draw [black, thick, smooth, tension=3] (1/2.5, 4.5/5) to [out=90,in=360] (1/3, 4.8/5);

		\draw [thick, solid, black]		(4/5, 4/5) -- (4.5/5, 4.5/5);
		
		\draw [black, thick, smooth, tension=3] (4.5/5, 4.5/5) to [out=20,in=250] (1,4.8/5);

		\draw [thin, dotted, black]			(1/3, 0) -- (1/3, 1);
		\draw [thin, dotted, black]			(0, 1/4) -- (1, 1/4);
			\draw [thin, dotted, black]	(1/2.5, 0) -- (1/2.5, 1);
				\draw [thin, dotted, black]	(4/5, 0) -- (4/5, 1);
			\draw [thin, dotted, black]	(0, 3.5/5) -- (1, 3.5/5);
			
			\draw [thin, dotted, black]	(0, 4.5/5) -- (1, 4.5/5);
			
			\node [circle, fill = red, scale=0.5]		at (1/2.5, 1/2.5)	{};
	\node [circle, fill = red, scale=0.5]		at (1/2.5, 3.5/5)	{};
	\node [circle, fill = red, scale=0.5]		at (1/2.5, 4.5/5)	{};
	\end{axis}
\end{tikzpicture}
&\qquad \quad
\begin{tikzpicture}[scale = 0.9]
	\small
	\begin{axis}		
		[axis x line = middle,
		axis y line = middle,
		xmin = 0, xmax = 1.1,
		ymin = 0, ymax = 1.1,
		xlabel=$x$,
		ylabel=$y$,
		xtick={0.001,1/4,1/2,3/4,1},
		xticklabels={$0$,$\frac{1}{4}$,$\frac{1}{2}$,$\frac{3}{4}$,$1$},
		ytick={0.001, 1/4,1/2, 3/4,1},
		yticklabels={$0$,$\frac{1}{4}$,$\frac{1}{2}$, $\frac{3}{4}$,$1$},
		clip=false]
		
		\draw [thin, dotted, black]			(0, 0) -- (1, 1);

		\draw [thick, solid, black]		(0,1/4) -- (1/2, 3/4);
		\draw [thick, solid, black]		(1/2,1/4) -- (1, 3/4);
		
    	\draw [thin, dotted, black]			(0, 1/2) -- (1, 1/2);
    		\draw [thin, dotted, black]			(0, 3/4) -- (1, 3/4);
		\draw [thin, dotted, black]			(1/2,0) -- (1/2, 1);
		\draw [thin, dotted, black]			(1/4,0) -- (1/4, 1);
		\draw [thin, dotted, black]			(3/4,0) -- (3/4, 1);
		\draw [thin, dotted, black]			(0,1/4) -- (1,1/4);
	
	\node [circle, fill = red, scale=0.5]		at (1/4, 0.5)	{};
	\node [circle, fill = red, scale=0.5]		at (1/2, 3/4)	{};
	\node [circle, fill = red, scale=0.5]		at (3/4, 0.5)	{};
	\draw [dashed, red]	(1/4, 0.5) -- (1/2, 3/4);
	\draw [dashed, red]	(1/2, 3/4) -- (3/4, 0.5);
	\end{axis}
\end{tikzpicture}\\
\small\textbf{e. A Complicated Single-Dipped Set} &\qquad \quad \small\textbf{f. Median Matching is not Single-Dipped} \\
\end{tabular}
\caption{Some Single-Dipped Disclosure Patterns}
\label{f:DP}
\end{figure}

Each panel in Figure \ref{f:DP} illustrates a signal in the linear receiver case ($u(y,x)=x-y$). Panel a.\ is full disclosure, which is trivially strictly single-dipped, as no states are paired. Panel b.\ is no disclosure, which is single-dipped but not strictly single-dipped. Panels c., d., and e.\ are all strictly single-dipped. Panel c.\ is an example of negative assortative disclosure, where state $x=1/3$ is disclosed and the other states are paired with weight $2/3$ on the higher state in each pair. Panel d.\ shows a signal where all states below $1/3$ (as well as state $1/2$) are disclosed, and the other states are paired with weight $3/4$ on the higher state in each pair. This ``disclose-pair'' pattern is a strictly single-dipped analogue of \emph{upper-censorship}, where all states below a cutoff are disclosed, and all states above the cutoff are pooled (e.g., \citealt{KMZ}). Upper-censorship is only weakly single-dipped, whereas disclose-pair splits up the pooling region in upper-censorship to obtain strict single-dippedness. Panel e.\ shows a more complicated strictly single-dipped signal. While strict single-dippedness implies that each action is induced at at most two states, Panel e.\ shows that more than two actions can be induced at a single state (here, state $2/5$).\footnote{Also, while the function $\chi_2$ is always monotone under strict single-dippedness, Panel e.\ shows that the function $\chi_1$ can be non-monotone.} Finally, Panel f.\ shows ``matching across the median'' (e.g., \citealt{KM}), which is not single-dipped, for example because it contains the strictly single-peaked triple $\left\{(1/4,1/2),(1/2,3/4),(3/4,1/2)\right\}$.

All of our results (and all proofs, except for the proof of Theorem \ref{t:SDPD}) are symmetric between the single-dipped and single-peaked cases. We thus present our results and proofs only for the single-dipped case (except for the proof of Theorem \ref{t:SDPD}), omitting the analogous results for the single-peaked case.

\begin{remark}\label{r:SDD}
A strictly single-dipped set can be described by two functions $\chi_1$ and $\chi_2$ that specify the states $\chi_1(y)$ and $\chi_2(y)$ which are pooled together to induce each action $y$. Specifically, for any strictly single-dipped set $\Lambda$, there exist unique functions $\chi_1$ and $\chi_2$ from $Y_\Lambda=\{\gamma(\mu):\mu\in \Lambda\}$ to $X$ such that $\supp(\mu)=\{\chi_1(\gamma(\mu)),\chi_2(\gamma(\mu))\}$ for all $\mu\in \Lambda$, $\chi_{1}(y)= \chi(y)= \chi_{2}(y)$ or $\chi_{1}(y)< \chi(y)< \chi_{2}(y)$ for all $y$, and $\chi_{2}(y)\leq \chi_{2}(y^{\prime })$ and $\chi_{1}(y^{\prime })\notin (\chi_{1}(y),\chi_{2}(y))$ for all $y<y^{\prime }$.\footnote{This remark follows from Corollary 1.6 and Lemma A.9 of \citet{BJ}. For completeness, we provide a simple self-contained proof in the appendix.}
\end{remark}

\subsection{Variational Theorem}

The next result captures the core economic logic behind single-dippedness. It is also our key tool for determining when optimal plans are single-dipped: we use it to establish our main sufficient condition for single-dipped disclosure to be optimal (Theorem \ref{t:SDPD}, in the next subsection), and also use it directly to study some applications in Appendix \ref{s:persuasion}.\footnote{Theorem \ref{l:ssdd} provides conditions under which \emph{every} optimal signal is single-dipped. In addition, Lemma \ref{l:stab} in Appendix \ref{a:SDPD} provides weaker conditions under which \emph{some} optimal signal has this property.}

\begin{theorem}\label{l:ssdd} 
Suppose that for any pair of actions $y_1<y_2$ and any triple of states $x_1<x_2<x_3$ such that $x_1< \chi(y_1)< x_3$, there exists a vector $\b \geq 0$ such that $R\b\geq 0$ and $R\b\neq 0$, where
\begin{equation*}\label{e:A}
R:=
\begin{pmatrix}
V(y_2,x_1)-V(y_1,x_1) &-(V(y_2,x_2)-V(y_1,x_2)) &V(y_2,x_3)-V(y_1,x_3)\\
-u(y_1,x_1) &u(y_1,x_2) &-u(y_1,x_3) \\
u(y_2,x_1) &-u(y_2,x_2) &u(y_2,x_3)
\end{pmatrix}.
\end{equation*}
Then the contact set is single-dipped, and hence so is every optimal signal.
\end{theorem}

\begin{figure}[t]
\centering
\begin{tikzpicture}
\node (a1) at (0,0) {};
\node[label=right:{$y_1$}] (a2) at (10,0) {};
\node[circle,fill=black,inner sep=0pt,minimum size=3pt,label=below:{$x_1$}] (a3) at (1,0) {};
\node[circle,fill=black,inner sep=0pt,minimum size=3pt,label=below:{$x_2$}] (a4) at (5,0) {};
\node[circle,fill=black,inner sep=0pt,minimum size=3pt,label=below:{$x_3$}] (a5) at (9,0) {};
\draw (a1) -- (a2);
\node (b1) at (0,2) {};
\node[label=right:{$y_2$}] (b2) at (10,2) {};
\draw (b1) -- (b2);
\node[circle,fill=black,inner sep=0pt,minimum size=3pt] (b3) at (1,2) {};
\node[circle,fill=black,inner sep=0pt,minimum size=3pt] (b4) at (5,2) {};
\node[circle,fill=black,inner sep=0pt,minimum size=3pt] (b5) at (9,2) {};
\draw[->] (a3) -- (b3);
\node (c3) at (1.3,1) {$\b_1$};
\draw[->] (b4) -- (a4);
\node (c3) at (5.3,1) {$\b_2$};
\draw[->] (a5) -- (b5);
\node (c3) at (9.3,1) {$\b_3$};
\end{tikzpicture}
\caption{A Profitable Perturbation of a Non-Single-Dipped Signal}
\caption*{\emph{Notes:} The figure shows a perturbation that shifts weights $\b_1$ and $\b_3$ on $x_1$ and $x_3$ from a posterior inducing action $y_1 $ to a posterior inducing action $y_2$, and shifts weight $\b_2$ on $x_2$ in the opposite direction. This perturbation is profitable if it increases the receiver's expected marginal utility at $y_1$ and $y_2$ and also increases the sender's expected utility for fixed $y_1$ and $y_2$.}
\label{f:sdd}
\end{figure}

The intuition behind Theorem \ref{l:ssdd} is very simple, and is illustrated in Figure \ref{f:sdd}. The condition in Theorem \ref{l:ssdd} says that a signal that induces a strictly single-peaked triple $(y_1,x_1)$, $(y_2,x_2)$, $(y_1,x_3)$ with positive probability can be improved by re-allocating mass $\b_1$ on $x_1$ and mass $\b_3$ on $x_3$ from $y_1$ to $y_2$, while re-allocating mass $\b_2$ on $x_2$ from $y_2$ to $y_1$. This re-allocation is profitable for the sender, because the sender's expected utility increases when $y_1$ and $y_2$ are held fixed (i.e., the first coordinate of $R\b $ is non-negative); the receiver's marginal utility conditional on being recommended $y_1$ increases (i.e., the second coordinate of $R\b $ is non-negative), which increases the receiver's action, and hence increases the sender's expected utility; and the receiver's marginal utility conditional on being recommended $y_2$ also increases (i.e., the third coordinate of $R\b $ is non-negative), which again increases the sender's expected utility. Moreover, at least one of these improvements is strict (i.e., $R\b  \neq 0$). The same logic also applies for any signal that induces a strictly single-peaked triple, even if this triple occurs with $0$ probability, except now mass must be re-allocated from small intervals around $x_1$, $x_2$, and $x_3$.\footnote{Formally, this step relies on our complementary slackness theorem, Theorem \ref{t:contact}.}

\subsection{Sufficient Conditions for Single-Dipped Disclosure}
We can now give our main sufficient condition on utilities for single-dipped disclosure to be optimal. This is a central result of our paper. As we will see, our condition covers several prior models, as well as some new applications.

\begin{theorem}\label{t:SDPD}
If
$u_{yx}(y,x)/u_{x}(y,x)$ and  $V_{yx} (y_2,x)/u_x(y_1,x)$  are increasing in $x$ for any $y$ and $y_1 \leq y_2$,
then there exists an optimal single-dipped signal.

Moreover, if in addition either $u_{yx}(y,x)/u_{x}(y,x)$ or $V_{yx} (y_2,x)/u_x(y_1,x)$  is strictly increasing in $x$ for any $y$ and $y_1 \leq y_2$, then the contact set is strictly single-dipped, and hence so is every optimal signal.
\end{theorem}

The proof establishes single-dippedness by constructing perturbations that satisfy the conditions in Theorem \ref{l:ssdd}, and further establishes strict single-dippedness by verifying the twist condition from Theorem \ref{t:pairwise}.

\begin{figure}
\begin{flushleft}
\begin{tabular}{cc}
 \pgfplotsset{width = 8 cm, compat = 1.18}
\begin{tikzpicture}
\begin{axis}[
    axis y line = left,
    axis x line = middle,
    ylabel style = {rotate=-90},
    xmin = -1,
    xmax = 10,
    ymin = -10,
    ymax = 10,
    xticklabels = {,,},
    yticklabels = {,,},
    xtick style={draw=none},
    ytick style={draw=none},
    legend style={at={(2.0,0.65)}},
    ]
\addplot[name path=C1,mark=none,solid,thick] [
    domain=2.3:9, 
    samples=100, 
]
    {1.1*(0.7*x^2-8*x)+23};
\addplot[name path=C1,mark=none,solid,thick] [
    domain=2.3:9, 
    samples=100, 
]
    {0.4*(0.7*x^2-8*x)+4};
\draw[dashed,-] (2.3,-2.) to (2.3,6.9);
 \node (a) at (2,-1){$x_1$};
\draw[dashed,-] (5.66,0) to (5.6,-5.29);
\node (b) at (5.3,-1){$x_2$};
\draw[dashed,-] (9,-2.15) to (9,6.3);
\node (c) at (8.65,-1){$x_3$};
\end{axis}
\node (d) at (1.15,4.4){$V\left(y_2,x\right)$};
\node (e) at (1.15,2){$V\left(y_1,x\right)$};
\node (f) at (-.4,5.6){$V$};
 \node (g) at (6.3,3){$x$};
\end{tikzpicture}
&
\pgfplotsset{width = 8 cm, compat = 1.18}
\begin{tikzpicture}
\begin{axis}[
    axis y line = left,
    axis x line = middle,
    ylabel style = {rotate=-90},
    xmin = -1,
    xmax = 10,
    ymin = -10,
    ymax = 10,
    xticklabels = {,,},
    yticklabels = {,,},
    xtick style={draw=none},
    ytick style={draw=none},
    legend style={at={(2.0,0.65)}},
]
\addplot[name path=C3,mark=none,solid,thick] [
    domain=2:8, 
    samples=100, 
]
{exp(0.20*x)*(x-7)+3};
\addplot[name path=C3,mark=none,solid,thick] [
    domain=2:8, 
    samples=100, 
]
{exp(0.05*x)*(1.4*x-5)};
\draw[dashed,-] (2,0) to (2,-4.6);
\node (a) at (1.66,-0.8){$x_1$};
\draw[dashed,-] (5,2.5) to (5,-2.5);
\node (b) at (4.63,-1){$x_2$};
\draw[dashed,-] (8,0) to (8,9.4);
\node (c) at (8,-1){$x_3$};
\end{axis}
\node (d) at (1,2){$u\left(y_1,x\right)$};
\node (e) at (1,1.4){$u\left (y_2,x\right )$};
\node (f) at (-.4,5.6){$u$};
 \node (g) at (6.3,3){$x$};
\end{tikzpicture}
\\
\small\textbf{a. Linear Receiver Case.} & \small\textbf{b. State-Independent Sender Case.}
\end{tabular}
\caption{The Intuition for Theorem \ref{t:SDPD} in Two Special Cases}
\caption*{\emph{Notes:} Panel a.\ In the linear receiver case, when the sender's utility increment $V(y_2,x)-V(y_1,x)$ is convex in the state, more extreme states should induce higher actions.
\\
Panel b.\ In the state-independent sender case, when the receiver's marginal utility $u(y,x)$ is more convex in the state at higher actions, more extreme states should induce higher actions.}
\label{f:lsp}
\end{flushleft}
\end{figure}
The intuition for Theorem \ref{t:SDPD} is relatively straightforward in the linear receiver and state-independent sender cases. (See Figure \ref{f:lsp}.) In the linear receiver case, we have $u_{yx}(y,x)/u_{x}(y,x)=0$ and $V_{yx} (y_2,x)/u_x(y_1,x)=V_{yx} (y_2,x)$, so our sufficient conditions for single-dipped disclosure to be optimal are satisfied iff $V_y$ is convex in $x$.\footnote{In the separable and translation-invariant subcases, convexity of $V_y$ simplifies to convexity of $w$ and $P'$, respectively.} To see why, note that for any strictly single-peaked triple $(y_1,x_1)$, $(y_2,x_2)$, $(y_1,x_3)$, the perturbation that moves mass on $x_1$ and $x_3$ from $y_1$ to $y_2$ and moves mass on $x_2$ in the opposite direction, so as to hold fixed the receiver's marginal utility conditional on being recommended either action, has the effect of also holding fixed the probability of each recommendation, while spreading out the state conditional on action $y_2$ and concentrating the state conditional on action $y_1$. This perturbation is profitable when the difference $V(y_2,x)-V(y_1,x)$ is convex in $x$, which holds whenever $V_y$ is convex in $x$.%

In the state-independent sender case, we have $V_{yx} (y_2,x)/u_x(y_1,x)=0$, so our sufficient conditions for single-dipped disclosure to be optimal are satisfied iff $u_x$ is log-supermodular in $(y,x)$, or equivalently $u$ is more convex in $x$ at higher actions $y$.\footnote{In the separable and translation-invariant subcases, log-supermodularity of $u_x$ simplifies to $2I'(x)^2\geq I(x)I''(x)$ and log-concavity of $T'$, respectively.} To see why, note that for any strictly single-peaked triple $(y_1,x_1)$, $(y_2,x_2)$, $(y_1,x_3)$, the perturbation that moves mass on $x_1$ and $x_3$ from $y_1$ to $y_2$ and moves mass on $x_2$ in the opposite direction, so as to hold fixed the receiver's marginal utility conditional on being recommended $y_1$ as well as the total probability of each recommendation, has the effect of increasing the receiver's marginal utility conditional on being recommended $y_2$. This follows because, by log-supermodularity of $u_x$, for the receiver's expected marginal utility the marginal rate of substitution between ``shifting weight from $x_1 $ to $x_2$'' and ``shifting weight from $x_2$ to $x_3$'' is higher at $y_1$ than $y_2$. Finally, when $V$ is state-independent, this perturbation increases the sender's expected utility.\footnote{In the linear receiver and state-independent sender cases, the sufficient conditions for the optimality of strict single-dipped disclosure in Theorem \ref{t:SDPD} are ``almost necessary,'' because the condition $|S|\neq 0$ on  $Y\times [0,1]$ implies that $|S|$ has a constant sign on $Y\times [0,1]$, which can be shown to be equivalent to strict convexity of $V_y$ in the linear receiver case, and to strict log-supermodularity of $u_x$ in the state-independent sender case. By Theorem \ref{t:pairwise}, a necessary condition for the optimality of strictly single-dipped disclosure is that $|S|\neq 0$ on the restricted domain where $x_1<\chi (y)<x_3$.}

As we explain in Section \ref{s:applications}, there are close antecedents to the conditions in Theorem \ref{t:SDPD} for the linear receiver and state-independent sender cases, in non-persuasion settings that nonetheless fall in our general optimal productive transport framework. In particular, results in the MOT literature (e.g., Theorem 6.1 of \citealt{BJ}) can be translated to our framework to imply the linear receiver case of Theorem \ref{t:SDPD}, while results in the gerrymandering literature (Lemma 1 of \citealt{FH}) can be translated to our framework to imply the state-independent sender case of Theorem \ref{t:SDPD}. Theorem \ref{t:SDPD} thus unifies and generalizes these disparate contributions.

We also establish an additional result in Appendix \ref{s:Brenier}: under our conditions for strictly single-dipped disclosure to be optimal (and a regularity condition), the optimal signal is unique.\footnote{This result is somewhat akin to \emph{Brenier's theorem} in optimal transport, which shows that the optimal transport plan is unique under a suitable complementarity-type condition, called the twist or generalized Spence-Mirrlees condition (\citealt{Brenier}, \citealt{GangboMcCann}; or see Section 1.3 in \citealt{santambrogio}).}

\section{Full Disclosure and Negative Assortative Disclosure}\label{s:MNAD}

While single-dippedness is an important property, it remains important to fully characterize optimal signals, when this is tractable.\footnote{Recall that many different disclosure patterns can be single-dipped, as illustrated in Figure \ref{f:DP}.} The current section does this for the polar cases of ``maximum'' and ``minimal'' pairwise disclosure. The former case corresponds to \emph{full disclosure}, where each state is disclosed; while the latter case corresponds to \emph{negative assortative disclosure}, where all states are paired in a negatively assortative manner, so all posteriors can be ordered from least to most extreme. Here our results on full disclosure extend existing results, whereas our results on negative assortative disclosure are entirely novel.

\subsection{Full Disclosure}
\emph{Full disclosure} is the (unique) signal $\tau$ where every $\mu \in \supp(\tau)$ is degenerate.

If for all states $x_1$ and $x_2$, and all probabilities $\rho $, the sender prefers to split the posterior $\mu=\rho \delta_{x_1}+(1-\rho )\delta_{x_2}$ into degenerate posteriors $\delta_{x_1}$ and $\delta_{x_2}$, then the sender prefers full disclosure to any pairwise signal. Since pairwise plans are without loss by Theorem \ref{t:pairwise}, full disclosure is then optimal. Conversely, if the sender strictly prefers not to split $\mu =\rho \delta_{x_1 }+(1-\rho )\delta _{x_2}$ into $\delta _{x_1 }$ and $\delta _{x _2}$ for some states $x_1$ and $x _2$ and some probability $\rho $, then the sender strictly prefers the pairwise signal that differs from full disclosure only in that it pools states $x_1$ and $x_2$ into $\mu $; so full disclosure is not optimal.\footnote{This argument is valid when $\phi $ has finite support. The general case (Theorem \ref{p:full}) uses duality and is adaptated from part (2) of Proposition 1 in \citet{Kolotilin2017}; we give a slightly simpler proof and also establish uniqueness.} Recalling that belief $\mu=\rho \delta _{x_1 }+(1-\rho )\delta _{x_2}$ induces action $\gamma(\mu )$ satisfying $\rho u(\gamma(\mu ),x_1 )+(1-\rho )u(\gamma(\mu ) ,x_2)= 0$, we obtain the following result.

\begin{theorem}
\label{p:full}
Full disclosure is optimal iff, for all $\mu=\rho \delta _{x_1}+(1-\rho )\delta _{x_2}$ with $x_1 <x_2$ in $X$ and $\rho \in (0,1)$, we have
\begin{equation}
\rho V(\gamma(\mu),x_1) +(1-\rho )V(\gamma(\mu),x_2)\leq \rho V(\gamma(\delta_{x_1}),x_1) +(1-\rho )V(\gamma(\delta_{x_2}),x_2).   \label{e:full}
\end{equation}
Moreover, full disclosure is uniquely optimal if \eqref{e:full} holds with strict inequality for all such $\mu$.
\end{theorem}

In the linear case, condition \eqref{e:full} holds iff $V$ is convex in $y$. In the state-independent sender case, condition \eqref{e:full} simplifies as follows:
\setcounter{corollary}{1}
\begin{corollary}
\label{c:fd} In the state-independent sender case, full disclosure is optimal iff, for all $\mu =\rho \delta _{x_1 }+(1-\rho )\delta _{x_2}$ with $x_1,x_2\in X$ and $\rho\in (0,1) $, we have
\begin{equation}
V\left(\gamma\left( \mu \right) \right) \leq \rho V\left(\gamma\left( \delta _{x_1 }\right) \right) +\left(1-\rho \right) V\left(\gamma\left( \delta _{x_2}\right)
\right).
\label{e:fullsi}
\end{equation}
\end{corollary}

In a classical one-to-one matching model, \citeasnoun{B73} showed that if the utility from matching two types $h\left( x_1 ,x _2\right) $ is supermodular, then it is optimal to match like types. \citeasnoun{LN} refer to this extreme form of positive assortative matching as \emph{segregation}. Their Propositions 4 and 9 show that segregation is optimal iff $h\left( x_1,x_1 \right) +h\left( x _2,x_2\right) \geq 2h\left( x _1,x_2\right) $ for all $x_1 ,x_2$ (which is a strictly weaker property than supermodularity). In the persuasion setting, segregation corresponds to full disclosure. Note that if we fix $\rho=1/2$ and let $h\left( x_1 ,x _2\right) =V\left( \gamma\left( \delta _{x_1 }/2+\delta _{x _2}/2\right) \right) $, then (\ref{e:fullsi}) reduces to \citeauthor{LN}'s condition. Intuitively, full disclosure is ``less likely'' to be optimal in persuasion than segregation is in classical matching, because in persuasion the designer has an extra degree of freedom $\rho $ in designing matches.

In the linear receiver case, there is a simple sufficient condition for \eqref{e:full}:
\begin{manualcorollary}{\ref{c:fd}'}
\label{c:fds} In the linear receiver case, full disclosure is optimal if $V(y ,x )$ is convex in $y$ and satisfies $V(x_1 ,x_2)+V(x_2,x_1 )\leq V(x_1 ,x_1 )+V(x_2,x_2)$ for all $x_1,x_2\in X$.
\end{manualcorollary}

A sufficient condition for $V(x_1 ,x_2)+V(x_2,x_1 )\leq V(x_1 ,x_1 )+V(x_2,x_2)$
is supermodularity of $V$: for all $x_1 <x_2$ and $y_1<y_2$, $V(y_1 ,x_1 )+V(y_2,x_2)\geq V(y_1 ,x_2)+V(y_2,x_1 )$. Thus, in the linear receiver case, full disclosure is optimal whenever the sender's utility is convex in $y$ and supermodular in $\left(y ,x \right) $. This sufficient condition for full disclosure generalizes that given by \citet{RS} for the separable subcase.\footnote{Their condition is that $w$ is increasing in $x$ and $G$ is convex in $y$, where $V(y,x)=w(x)G(y)$. In the sub-subcase with $G(y)=y$, \eqref{e:full} holds iff $w$ is increasing in $x$, because \eqref{e:full} simplifies to $\rho (1-\rho)(w(x_2)-w(x_1))(x_2-x_1)\geq 0$.}

In addition, in Appendix \ref{s:Brenier} we show that when the prior has full support and the twist condition holds, full disclosure is uniquely optimal whenever it is optimal.

\subsection{Negative Assortative Disclosure}
A pairwise signal $\tau$ is \emph{negative assortative} if the supports of any $\mu , \mu' \in \supp(\tau)$ are nested: that is, denoting $\supp(\mu)=\{x_1 , x_2\}$ and $\supp(\mu')=\{x'_1 , x'_2\}$, where without loss $x_1 \leq x_2$ and $x'_1 \leq x'_2$, we have either $x_1 \leq x'_1 \leq x'_2 \leq x_2$ or $x'_1 \leq x_1 \leq x_2 \leq x'_2$. We also apply the same definition to an arbitrary set of posteriors $M \in \Delta(X)$ by replacing $\supp(\tau)$ with $M$. In particular, a strictly single-dipped contact set $\Lambda$ is negative assortative if $\chi_1 :Y_\Lambda \to X$ is decreasing and $\chi_2: Y_\Lambda \to X$ is increasing.\footnote{Recall that, as defined in Remark \ref{r:SDD}, $\chi_1 (y)$ and $\chi_2 (y)$ are the smaller and larger states that are pooled together to induce action $y\in Y_\Lambda$.}

The main result of this section is that if strictly single-dipped disclosure is optimal and the sender strictly prefers to pool any two states, then negative assortative disclosure is optimal. Moreover, if the prior has a density, then the optimal signal is unique (by Theorem \ref{t:brenier} in Appendix \ref{s:Brenier}) and is characterized as the solution to a system of two ordinary differential equations.

To see the intuition, note that if strictly single-dipped disclosure is optimal, then any two pairs of pooled states $\{x_1,x_3\}$ and $\{x'_1,x'_3\}$ with $x_1<x_3$, $x'_1<x'_3$, and $x_1\leq x'_1$, must be either ordered (i.e., $x_1<x_3\leq x'_1< x'_3$) or nested (i.e., $x_1\leq x'_1<x'_3\leq x_3$). This follows because if the pairs overlap (i.e., $x_1<x'_1<x_3<x'_3$), then either $(x_1, x'_1, x_3)$ or $(x'_1, x_3, x'_3)$, together with the corresponding actions, would form a single-peaked triple. Hence, for any pair of pooled states $\{x_1,x_3\}$, there must exist a disclosed state $x_2 \in (x_1,x_3)$: intuitively, there must exist pairs of pooled states in the interval $(x_1,x_3)$ that are closer and closer together, until the pair degenerates into a single disclosed state. Therefore, if any two pairs of pooled states $\{x_1,x_3\}$ and $\{x'_1,x'_3\}$ are ordered, there would exist two distinct disclosed states $x_2 \in (x_1,x_3)$ and $x'_2 \in (x'_1,x'_3)$. But if the sender strictly prefers to pool any two states, this is impossible. Finally, if pairs of pooled states cannot overlap or be ordered, the only remaining possibility is that all pairs of pooled states are nested: that is, disclosure is negative assortative.\footnote{In this argument, the existence of the two disclosed states relies on the assumption that $\supp(\phi)=[0,1]$. The formal proof again relies on complementary slackness.} %

\begin{theorem}\label{t:NAD}
Assume that $X=[0,1]$. If $\Lambda$ is strictly single-dipped and for all $x_1<x_2$ there exists $\rho\in (0,1)$ such that
\begin{equation}
\rho V(\gamma(\mu),x_1) +(1-\rho )V(\gamma(\mu),x_2)> \rho V(\gamma(\delta_{x_1})),x_1) +(1-\rho )V(\gamma(\delta_{x_2}),x_2),   \label{e:nd}
\end{equation}
with $\mu=\rho \delta _{x_1}+(1-\rho )\delta _{x_2}$, then $\Lambda$ is negative assortative. Moreover, if the prior $\phi$ has a density $f$, then the optimal signal is unique, and the functions $\chi_1$ and $\chi_2$ are continuous and solve the system of differential equations,
\begin{gather}
u(y,\chi_1(y))(-\df \phi ([0,\chi_1(y)]))+u(y,\chi_2(y))\df \phi ([0,\chi_2(y)])=0, \label{e:obed} \\
\begin{gathered}
\frac{\df}{\df y} \left(\frac{V_y(y,\chi_1(y))u(y,\chi_2(y))-V_y(y,\chi_2(y))u(y,\chi_1(y))}{u(y,\chi_1(y))u_y(y,\chi_2(y))-u(y,\chi_2(y))u_y(y,\chi_1(y))}\right) \\
=\frac{V_y(y,\chi_1(y))u_y(y,\chi_2(y))-V_y(y,\chi_2(y))u_y(y,\chi_1(y))}{u_y(y,\chi_1(y))u(y,\chi_2(y))-u_y(y,\chi_2(y))u(y,\chi_1(y))},	\label{e:q'} 
\end{gathered}	
\end{gather}
for all $y>\ul y$ in $Y_\Lambda$, with the boundary conditions
\begin{equation}
\begin{gathered}
	(\chi_1(\ol y),\chi_1(\ul y),\chi_2(\ul y),\chi_2(\ol y))=(0,\chi (\ul y),\chi (\ul y),1), \label{e:boundary}
\end{gathered}
\end{equation}
where $\ul y=\min Y_\Lambda$ and $\ol y=\max Y_\Lambda$.
\end{theorem}

Like equation \eqref{e:full} in the previous subsection, equation \eqref{e:nd} simplifies in special cases. In the linear case, \eqref{e:nd} holds iff $V$ is strictly concave in $y$.\footnote{In the linear case, $V$ is strictly concave iff no disclosure is uniquely optimal for all priors, by Corollary 1 in \citet{KMZ}.} In the state-independent sender case, it holds iff $V(\gamma(\mu))> \rho V(\gamma(\delta_{x_1})) +(1-\rho )V(\gamma(\delta_{x_2}))$. In the linear receiver case, it holds if $V(y,x)$ is concave in $y$ and satisfies $V(x_1 ,x_2)+V(x_2,x_1 )> V(x_1 ,x_1 )+V(x_2,x_2)$ for all $x_1<x_2$; a sufficient condition for the latter property is strict submodularity of $V$. These conditions generalize the sufficient condition for pooling given by \citet{RS} for the separable subcase.\footnote{Their condition is that $w$ is strictly decreasing in $x$ and $G$ is concave in $y$, where $V(y,x)=w(x)G(y)$.  In the sub-subcase with $G(y)=y$, \eqref{e:nd} holds iff $w$ is strictly decreasing in $x$.}
 
To understand the differential equations, note that if $\chi_1$ and $\chi_2$ are differentiable then \eqref{e:obed} can be written as 
\[u(y,\chi_1(y))f(\chi_1(y))\chi_1'(y)=u(y,\chi_2(y))f(\chi_2(y))\chi_2'(y).\] 
This follows from \eqref{e:FOCy} and \eqref{PS1}, as each posterior $\mu$ inducing $y=\gamma(\mu)$ is 
\[
\mu=\frac{-f(\chi_1(y))\chi_1'(y)}{-f(\chi_1(y))\chi_1'(y)+f(\chi_2(y))\chi_2'(y)}\delta_{\chi_1(y)}+\frac{f(\chi_2(y))\chi_2'(y)}{-f(\chi_1(y))\chi_1'(y)+f(\chi_2(y))\chi_2'(y)}\delta_{\chi_2(y)}.\footnote{This equation is a version of the Monge-Ampere equation in optimal transport (e.g., Section 1.7.6 in \citealt{santambrogio}).}
\]
In addition, \eqref{e:q'} results from solving the system of equations (from the sender's FOC, \eqref{e:FOC}), 
\begin{gather*}
v(y,\chi_1(y))+ q(y)u_y(y,\chi_1(y))+q'(y)u(y,\chi_1(y))=0, \\
v(y,\chi_2(y))+ q(y)u_y(y,\chi_2(y))+q'(y)u(y,\chi_2(y))=0,
\end{gather*}
for $q(y)$ and $q'(y)$, and recalling that $q'$ is the derivative of $q$. Finally, the boundary condition \eqref{e:boundary} follows because the lowest induced action $\underline{y}$ is induced at the disclosed state $\chi (\ul{y})=\chi_1 (\ul{y})=\chi_2 (\ul{y})$, and the highest induced action $\ol y$ is induced at states $0=\chi_1 (\ol y)$ and $1=\chi_2 (\ol y)$.\footnote{In the linear receiver case, \eqref{e:q'} simplifies to
\[
\frac{\df}{\df y} \left(V_y(y,\chi_1(y))\frac{\chi_2(y)-y}{\chi_2(y)-\chi_1(y)}+V_y(y,\chi_2(y))\frac{y-\chi_1(y)}{\chi_2(y)-\chi_1(y)}\right) =-\frac{V_y(y,\chi_2(y))-V_y(y,\chi_1(y))}{\chi_2(y)-\chi_1(y)}.
\]
Geometrically, this says that the slope of the curve $\gamma(\mu)\mapsto \E_{\mu} [V_y(\gamma(\mu),x)]$ is equal to the negative of the slope of the secant passing through the points $(\chi_1(\gamma(\mu)),V_y(\gamma(\mu),\chi_1(\gamma(\mu))))$ and  $(\chi_2(\gamma(\mu)),V_y(\gamma(\mu),\chi_2(\gamma(\mu))))$. \citet{NP} derive this condition for the separable sub-subcase with $V_y(y,x)=w(x)$.}

We can also give primitive conditions on $V$ and $u$ for \eqref{e:nd} to hold, and hence for the unique optimal signal to be negative assortative.

\begin{corollary}\label{c:ndSDD}
Assume that the condition for strict single-dippedness given in Theorem \ref{t:SDPD} holds. Then for all $x_1<x_2$ there exists $\rho\in (0,1)$ such that \eqref{e:nd} holds iff, for all $y\in Y$,
\begin{equation}\label{e:ndSDD}
\begin{gathered}
V_{yy}(y,\chi (y)) \leq \tfrac{V_y(y,\chi (y)) u_{yy}(y,\chi (y))}{u_y(y,\chi (y))}
+2\tfrac{V_{yx}(y,\chi (y)) u_y(y,\chi (y)) - V_y(y,\chi (y)) u_{yx}(y,\chi (y))}{u_x(y,\chi (y))}.
\end{gathered}
\end{equation}

\end{corollary}

Equation \eqref{e:ndSDD} is a local necessary condition for \eqref{e:nd}: if  \eqref{e:ndSDD} fails, then  \eqref{e:nd} also fails for $x_1<x_2$ sufficiently close to $\chi(y)$. When the condition for strict single-dippednes holds, this local necessary condition turns out to be globally sufficient for \eqref{e:nd}. Equation \eqref{e:ndSDD} simplifies dramatically in some special cases. In the linear receiver case, \eqref{e:ndSDD} simplifies to $V_{yy}(y,y)+2V_{yx}(y,y) \leq 0$; in the translation-invariant subcase of the linear receiver case, this simplifies further to $P''(0)\geq 0$. In the separable (resp., translation-invariant) subcase of the state-independent sender case, \eqref{e:ndSDD} simplifies to ${V_{yy}(y)}/{V_y(y)}\leq 2I'(y)/I(y)$ (resp., ${V_{yy}(y)}/{V_y(y)}\leq {T''(0)}/{T'(0)}$).

In Appendix \ref{s:examples}, we give three examples of optimal single-dipped negative assortative disclosure. Example \ref{e:RS} illustrates how the differential equations \eqref{e:obed}--\eqref{e:q'} can sometimes be explicitly solved to find the optimal signal. Example \ref{e:quantile} characterizes the optimal signal in quantile persuasion (i.e., Case (3c) in Section \ref{s:special}). In quantile persuasion, our sufficient conditions for strictly single-dipped disclosure to be optimal are not satisfied, and there are multiple optimal signals; however, one optimal signal is strictly single-dipped negative assortative. Finally, Example \ref{ex:segpair} illustrates that in some cases the unique optimal signal randomizes conditional on the state, even when the prior is atomless.\footnote{In contrast, \citet{Zheng} shows that there is always a deterministic optimal signal in the separable subcase of the linear receiver or state-independent sender case.}

\section{Optimal Productive Transport: Theory and Applications} \label{s:applications}

We now give a general interpretation of our framework in terms of assigning inputs to productive units, and describe the implications of our results for matching, option pricing, and partisan gerrymandering, as well as some specific persuasion models. We pay particular attention to the implications of single-dippedness summarized in Table 1 in the Introduction.

\subsection{Optimal Productive Transport}

\label{s:general}

Our signal-based primal problem \eqref{PS} may be generalized as follows: given a distribution $\phi \in \Delta([0,1])$ with support $\supp (\phi)=X$ of \emph{inputs} $x$, find a distribution $\tau \in \Delta(\Delta(X))$ of \emph{productive units} $\mu \in \Delta(X)$ to 
\begin{gather*}
\text{maximize} \quad \int_{\Delta (X)} \int_{X} V(\tilde{\gamma}(\mu),x) \df \mu(x) \df \tau (\mu)\\
\text{subject to}\quad \int_{\Delta(X)} \mu \df \tau (\mu) = \phi,
\end{gather*}
where $\tilde{\gamma}:\Delta(X) \to Y$ is an arbitrary \emph{production function} that specifies the \emph{output} $y $ produced by unit $\mu$.\footnote{The interpretation of the assumption that the domain of $\tilde{\gamma}$ is probability measures on $X$ rather than arbitrary measures is that production exhibits constant returns to scale: i.e., nothing is gained by varying the scale of a productive unit.} This \emph{optimal productive transport} problem is the same as \eqref{PS}, except that the value of an arbitrary production function $\tilde{\gamma}(\mu)$ is not necessary given by the first-order condition \eqref{e:FOCy} for some function $u$.

But now suppose that $\tilde{\gamma}$ satisfies the following two properties:
\begin{description}
\item[Betweenness] For any $\mu,\eta \in \Delta(X)$ satisfying $\tilde{\gamma}(\mu)>\tilde{\gamma}(\eta)$, and any $\rho \in (0,1)$, we have
$\tilde{\gamma}(\mu)>\tilde{\gamma}(\rho \eta+ (1-\rho) \mu)>\tilde{\gamma}(\eta).$
\item[Continuity] $\tilde{\gamma}$ is continuous on $\Delta(X)$.
\end{description}
The key property here is the first one, which says that mixing two units produces an output in between those produced by each of them in isolation. The following result---which adapts Proposition A.1 of \citet{dekel}---shows that these properties ensure the existence of a function $u$ such that $\tilde{\gamma}$ is given by \eqref{e:FOCy}. Thus, under Betweenness and Continuity, the optimal productive transport problem is exactly the same as our persuasion problem.

\begin{proposition} \label{p:dekel}
A function $\tilde{\gamma}$ satisfies Betweenness and Continuity iff there exists a continuous function $u:Y \times X \to \mathbb{R}$ such that, for any $\mu \in \Delta(X)$,
\[
\int_X u(y,x)\df \mu (x) = (<) 0 \iff y = (>)\tilde\gamma(\mu).
\]
\end{proposition}

In the persuasion context, the function $u$ in Proposition \ref{p:dekel} is the receiver's marginal utility. More generally, $u(y,x)$ can be viewed as a measure of the ``efficacy'' of input $x$ in the production of output $y$. For the remainder of this section, we assume Betweenness and Continuity, as well as that $u$ satisfies Assumptions \eqref{a:smooth}--\eqref{a:ord}.

The general interpretation of the signal-based primal problem \eqref{PS} is that a planner obtains utility $V(y,x)$ from using input $x$ in the production of output $y$, and assigns inputs to productive units according to a \emph{production plan} $\tau \in \Delta(\Delta(X))$ to maximize the expectation of $V(\gamma(\mu),x)$. The corresponding interpretation of the signal-based dual problem \eqref{DS} is that there is a decentralized economy with price $p(x)$ attached to input $x$, where the zero-profit condition \eqref{DS1} says that an entrepreneur who obtains utility $V(y,x)$ from using input $x$ in the production of output $y$ cannot create a unit $\mu \in \Delta(X)$ that leaves her with a positive utility after paying for its inputs.

More precisely, a pair $(\tau,p)\in \Delta(\Delta(X))\times L(X)$ is a \emph{competitive equilibrium} if\\
(i) all inputs are assigned to productive units:  $\int \mu \df \tau (\mu)=\phi$;\\
(ii) operating units make zero profits: $\int p(x)\df \mu =W(\mu)$ for all $\mu \in \supp (\tau)$;\\
(iii) no entrant can make strictly positive profits: $\int p(x)\df \mu \geq W(\mu)$ for all $\mu\in \Delta(X)$.\\
Then, by strong duality (Lemma \ref{l:dual}), we have
\begin{remark} \label{r:ce}
A competitive equilibrium exists, and a pair $(\tau,p)$ is a competitive equilibrium iff $\tau$ solves \eqref{PS} and $p$ solves \eqref{DS}.
\end{remark}
We will make use of the interpretation of optimal plans and prices as competitive equilibria in the matching context in Section \ref{s:matching}.

Similarly, the general interpretation of the outcome-based primal problem \eqref{PO} is that a planner chooses an outcome $\pi \in \Delta(Y\times X)$ to maximize her expected utility, subject to the constraints that all inputs are utilized and that each output $y$ is produced by a unit $\mu$ satisfying $\gamma(\mu)=y$. The corresponding dual can again be interpreted as a decentralized economy, where now the zero-profit condition \eqref{DO1} says that an entrepreneur who produces any output $y$ cannot profitably employ any input in the production of this output, after accounting for the input's price, $p(x)$, and the marginal effect of its use on the output, $q(y)u(y,x)$. Finally, complementary slackness (Theorem \ref{t:contact}) says that any entrepreneur who breaks even must employ only inputs that satisfy \eqref{DO1} with equality.

\subsection{Matching and Clubs} \label{s:matching}

Assigning workers with heterogeneous abilities to firms with workplace peer effects (i.e., intra-firm spillovers) is an important topic in labor economics (\citealt{KM}, \citealt{SaintPaul}, \citealt{Eeckhout}, \citealt{BTZ}). In \citet{SaintPaul}, there is a continuum of workers indexed by ability $x \in [0,1]$. The population distribution of $x$ is $\phi$, with support $X\subset [0,1]$. Workers sort into firms, which are ex ante homogeneous and face constant returns to scale, so that a firm can be identified with its distribution of workers $\mu \in \Delta(X)$. Workplace peer effects depend on the mean worker ability in a firm, $\gamma(\mu)=\mathbb{E}_{\mu}[x]$, so the output of a worker with ability $x$ in firm $\mu$ can be written as $V(\gamma(\mu),x)$. The planner's problem of assigning workers to firms to maximize total output is thus precisely \eqref{PS}, in the linear receiver case where $\gamma(\mu)=\mathbb{E}_{\mu}[x]$. Moreover, the problem of finding competitive equilibrium wages $p(x)$ for workers with ability $x$ (where, as in Section \ref{s:general}, a competitive equilibrium is an assignment of workers to firms and wages such that operating firms make zero profits and no entrant can make strictly positive profits), is precisely \eqref{DS}. In this context, Lemma \ref{l:dual} says that a competitive equilibrium exists and maximizes total output.\footnote{This generalizes Proposition 3 of \citet{SaintPaul}.}

\citet{SaintPaul} considered the special case of this model where $V(y,x)=xG(y)$ for an increasing function $G$. In this case, the total output of a firm $\mu$ equals $\mathbb{E}_{\mu}[x]G(\mathbb{E}_{\mu}[x])$. Since this is a function only of $\mathbb{E}_{\mu}[x]$, \citet{SaintPaul} coincides with the linear case of our model. We now summarize the implications of our results for the general linear receiver case of our model in the worker assignment context. We pay special attention to the separable case where $V(y,x)=w(x)G(y)$ for increasing functions $w$ and $G$, which may be particularly natural in this context.

In the worker assignment context, an assignment is pairwise if each firm contains at most two worker types, and a pairwise assignment is strictly single-dipped if firms with more heterogeneous workers have higher average worker ability. Since we are in the linear receiver case, Theorem \ref{t:SDPD} implies that a strictly single-dipped assignment is optimal whenever $V_y$ is strictly convex in $x$---or, in the separable case, $w$ is strictly convex. Intuitively, $V_y$ is the marginal benefit of having higher-ability coworkers, so when this is convex in a worker's own ability, it is optimal for firms where the distribution of worker abilities is more spread out to have higher average worker ability. Moreover, when in addition $V$ is strictly increasing in both arguments and convex in $x$, firms with more heterogeneous workers also produce higher output, as if $\mathbb{E}_{\mu}[x]=y<y'=\mathbb{E}_{\mu'}[x]$ and $\supp(\mu')$ is more spread out than $\supp(\mu)$, then
\[\mathbb{E}_{\mu}[V(y,x)]<\mathbb{E}_{\mu}[V(y',x)]<\mathbb{E}_{\mu'}[V(y',x)].\footnote{Here, the first inequality follows because $V_y >0$, and the second inequality follows because $V_{x},V_{xx}>0$ and $\mu'$ can be obtained from $\mu$ by increasing its mean and then taking a mean-preserving spread (i.e., $\mu'$ is larger than $\mu$ in the increasing convex order).}\]
Thus, in this case firms with more heterogeneous workforces are more productive, and hence also pay higher average wages (by zero profit).

Our conditions for full disclosure (i.e., \emph{segregation}, where each firm has a homogeneous workforce) and negative assortative disclosure/matching (where all firms can be ordered from least to most heterogeneous) are also interesting in the worker assignment context. By Corollary \ref{c:fds}, segregation is optimal if $V$ is convex in $y$ and supermodular---or, in the separable case, $G$ is convex. On the other hand, by Corollary \ref{c:ndSDD}, negative assortative matching is optimal if $V_y$ is strictly convex in $x$ and $V_{yy}+2V_{yx}\leq 0$---or, in the separable case, $w$ is strictly convex and $w(x)G''(y)+2w'(x)G'(y)\leq 0$. Intuitively, these results say that segregation is optimal if peer effects are convex, and that negative assortative matching is optimal if single-dipped assignment is optimal and peer effects are sufficiently concave.

We also note that $p(x)$ is convex whenever $V$ is convex in $x$---or, in the separable case, $w$ is convex. This follows because $p(x)=\sup_{y\in Y} V(y,x)+q(y)(x-y)$, so if $V_{xx}>0$ then $p$ is the supremum of a set of convex functions. Recalling that $p(x)$ is the equilibrium wage of a worker with ability $x$, this says that wages rise more than one-for-one with ability. This result reflects the fact that higher-ability workers are not only better workers, but also better coworkers.

A model that is equivalent to worker-firm matching can capture the assignment of students to schools with peer effects, or more generally the assignment of heterogeneous agents to clubs. In \citet{ArnottRowse}, there is a continuum of students indexed by ability $x$, who must be assigned to ex ante identical schools, which can be identified with their student bodies $\mu \in \Delta(X)$. A student with ability $x$ who attends a school $\mu$ attains an education that she values at $V(\gamma(\mu),x)$, where again $\gamma(\mu)=\mathbb{E}_{\mu}[x]$. \citet{ArnottRowse} study the planner's problem of assigning students to schools to maximize total educational attainment: this problem is equivalent to the linear receiver case of \eqref{PS}. The ``decentralized'' version of this problem is considered by \citet{EppleRomano}, who study competitive equilibrium in a market for private schooling, where a school $\mu$ with mean student ability $y=\gamma(\mu)$ charges tuition $t(y,x)$ to students of ability $x$. Here, a competitive equilibrium may be defined precisely as in Section \ref{s:general}, with the interpretation that $p(x)$ is the equilibrium utility of a student with ability $x$, and that the tuition charged to a student with ability $x$ to attend a school with mean student ability $y$ is $t(y,x)=V(y,x)-p(x)$. (The assumption here that schools take student utility levels as given when setting tuition is called \emph{utility taking} in the literature on club economies: e.g., \citealt{EGSZ}.) In this context, Lemma \ref{l:dual} says that a competitive equilibrium exists and maximizes total educational attainment.\footnote{\citet{ArnottRowse} additionally endogenize public spending on schools, while \citet{EppleRomano} let students differ in income as well as ability. \citet{ArnottRowse} focus on the Cobb-Douglas production function $V(y,x)=x^\alpha y^\beta$, and provide conditions for the optimality of segregation (``perfect streaming'') or no disclosure (``complete mixing'').}

The conditions on $V$ under which an optimal assignment of students to schools is single-dipped, segregation, or negative assortative are the same as in the worker assignment context. Indeed, bearing in mind that $p(x)$ is the equilibrium utility of an agent with ability $x$ in either model---so that the wage of a worker with ability $x$ is $p(x)$, and the tuition paid by a student with ability $x$ who attends a school with mean ability $y$ is $V(y,x)-p(x)$---the models are identical. In particular, if $V_y$ is strictly convex in $x$, then a strictly single-dipped assignment is optimal, so that schools with more heterogeneous student bodies are more desirable for all students.

Finally, the outcome-based dual \eqref{DO} has a particularly natural interpretation in the student assignment/club economy setting. In a competitive equilibrium, a school with mean student ability $y$ charges tuition $t(y,x)=q(y)(y-x)$ to students with ability $x$. (Thus, a school subsides its students with above-average ability, and charges students with below-average ability.) A student with ability $x$ attends the school $y$ that gives her the highest utility, $p(x)=\sup_{y \in Y} V(y,x)-q(y)(y-x)$. All operating schools break even, and no entrepreneur can turn a positive profit by starting a new school.

\subsection{Option Pricing} \label{s:option}

In mathematical finance, the literature on martingale optimal transport (e.g., \citealt{BHP}, \citealt{GHT}, \citealt{BJ}) studies the following problem. An underlying asset will be marketed in two future periods, 1 and 2. In period 0, an exotic option is for sale, which will pay $V(y,x)$ if the realized asset price is $y$ in period 1 and $x$ in period 2. An analyst knows the marginal distributions of $y$ and $x$, but her only information regarding their joint distribution is that it satisfies $\mathbb{E}[x|y]=y$ for every $y$. The interpretation of this assumption is that there are liquid markets for European call options on the asset price in each period, from which the analyst can infer the marginal distributions of the asset price (by \citealt{BreedenLitzenberger}); and the analyst believes that the market satisfies no-arbitrarge, which implies that the asset price is a martingale under the risk-neutral measure. The analyst's problem is to find the joint distribution $\pi \in \Delta (Y\times X)$ that maximizes the expected value of the option (and, thus, the maximum option price consistent with no-arbitrage), subject to the two marginal constraints and the martingale constraint.

Now consider the variant of this problem where the marginal distribution of $y$ is also unknown. The interpretation is that there is a liquid market for call options only on the period 2 asset price: for example, perhaps the asset is a share in a firm that is expected to go public after period 1, and there are only liquid options markets for the prices of publicly traded firms. Then the analyst's problem of determining the maximum option price, subject to constraint that the marginal distribution of the period 2 price $x$ is $\phi$, and the martingale constraint $\mathbb{E}[x|y]=y$, is precisely \eqref{PO}, in the linear receiver case where $\gamma(\mu)=\mathbb{E}_{\mu}[x]$. Lemma \ref{l:dual} establishes strong duality for this problem.\footnote{The possibility that the period 1 marginal may be unknown, and the resulting problem \eqref{PO}, are briefly considered in Corollary 1.5 of \citet{Acciaio}. That result establishes weak duality and primal attainment, but not dual attainment, which as we discuss is an important issue.}

In this context, the optimal dual variables $(p,q)$ have an important interpretation. Recall that the option to be priced pays $V(y,x)$ when the asset price is $y$ in period 1 and $x$ in period 2. An alternative to buying this exotic option is to buy a simple option that pays $p(x)$ when the period 2 asset price is $x$, and in addition to plan to sell $q(y)$ units of the asset itself in period 1 when the period 1 asset price is $y$. Since selling $q(y)$ units at price $y$ yields a profit of $q(y)(y-x)$ when the period 2 asset price turns out to be $x$, this alternative strategy is sure to outperform---or \emph{super-replicate}---the exotic option iff
\[p(x)+q(y)(y-x) \geq V(y,x) \quad \text{ for all } (y,x)\in Y \times X.\]
Note that this condition is precisely \eqref{DO1}. Thus, Lemma \ref{l:dual} implies that the maximum option price can be calculated as either $\int_{Y \times X}V(y,x) \df \pi(y,x)$ under the joint distribution of asset prices $\pi$ that solves \eqref{PO} (i.e., the maximum expected value of the exotic option), or as $\int_{X} p(x) \df \phi(x)$, for the simple option payouts $p(x)$ that solve \eqref{DO} (i.e., the price of the cheapest simple option that super-replicates the exotic option).

In the option pricing context, a joint distribution of asset prices is pairwise if it is a \emph{binomial tree}: each period 1 price $y$ can be followed by at most two distinct period 2 prices $x$. Moreover, a pairwise joint distribution is strictly single-dipped if more dispersed period 2 prices follow higher period 1 prices: that is, if riskier assets are more expensive. Since we are in the linear receiver case, Theorem \ref{t:SDPD} implies that the option price is maximized by a strictly single-dipped distribution whenever $V_y$ is strictly convex in $x$. This condition is known as the ``martingale Spence-Mirrlees condition'' in the MOT literature, which \citet{BJ}, \citet{HT}, and \citet{BHT} show implies that a strictly single-dipped distribution (which they call a ``left-curtain coupling'') is optimal in the standard MOT problem (where the period 1 asset price distribution is fixed exogenously).\footnote{Specifically, \citet{BJ} show that the unique optimal outcome is single-dipped in the translation-invariant subcase if $P'$ is strictly convex (Theorem 6.1), and in the separable subcase if $w$ is strictly convex (Theorem 6.3); while Theorem 5.1 in \citet{HT} and Theorem 3.3 in \citet{BHT} extend this conclusion to the general linear receiver case where $V_y$ is strictly convex in $x$. All these papers concern the MOT context, where the distribution of $y$ is fixed exogenously.} Moreover, by Corollary \ref{c:fds}, full disclosure (where $x=y$ with probability 1) is optimal if $V$ is strictly convex in $y$ and supermodular; while Corollary \ref{c:ndSDD} implies that negative assortative matching (where higher period 1 prices are always followed by more dispersed period 2 prices, so more expensive assets are riskier) is optimal if $V_y$ is strictly convex in $x$ and $V_{yy}+2V_{yx}\leq 0$.

The formula for $q(y)$ also has an interesting interpretation in the option pricing context. By equation \eqref{e:q}, for every period 1 price $y$ in the support of the marginal of an optimal joint distribution $\pi$, we have $$q(y)=\mathbb{E}_{\pi}[V_{y}(y,x)|y].$$ This is a version of Shephard's lemma: the amount of the asset sold at period 1 price $y$ under the cheapest super-replicating strategy equals the derivative of the option price with respect to $y$. Moreover, in finance, the derivative of the option price with respect to the underlying asset price is known as the option's ``Delta.'' Thus, in the option pricing context, $q(y)$ is simply Delta.

\subsection{Partisan Gerrymandering} \label{s:gerry}

Partisan gerrymandering---where a partisan designer assigns voters to districts to maximize her party's seat share---is an important feature of American politics. \citet{KW} develop and calibrate a model of partisan gerrymandering, which generalizes the leading earlier models of \citet{OG}, \citet{FH}, and \citet{GP}. %
In this model, there is a continuum of voters indexed by their partisanship $x \in [0,1]$. The population distribution of $x$ is $\phi$, with support $X \subset [0,1]$. The designer chooses a \emph{districting plan} $\tau \in \Delta(\Delta(X))$ that assigns voters to equipopulous districts $\mu \in [0,1]$, prior to the realization of an aggregate shock $y \in \mathbb{R}$ with cdf $V$. The share of type-$x$ voters who vote for the designer's party when the aggregate shock takes value $y$ is deterministic and is denoted by $v(y,x) \in [0,1]$.\footnote{Among other notational differences, the order of the arguments of $v$ is reversed in \citet{KW}.}
The function $v(y,x)$ is assumed to be strictly decreasing in $y$ and strictly increasing in $x$: that is, higher aggregate shocks are less favorable for the designer, while voters with higher partisanship are more favorable. The designer wins a district $\mu$ iff she receives a majority of the district vote. Thus, defining $u(y,x):=v(y,x)-1/2$, note that the designer wins a district $\mu$ iff $y \leq \gamma(\mu)$, where $\gamma(\mu)$ is given by \eqref{e:FOCy}. The designer thus wins a district $\mu$ with probability $V(\gamma(\mu))$. Finally, the designer chooses $\tau$ to maximize her expected seat share, subject to the constraint that all voters are assigned to equipopulous districts: i.e., $\int_{\Delta(X)}\mu \df \tau(\mu)=\phi$.\footnote{As discussed in \citet{KW}, the equipopulation constraint is strictly enforced in practice, while other constraints on districting (such as geographic continuity of districts) are often relatively slack, and are thus neglected in much of the gerrymandering literature.} The designer's problem is thus precisely \eqref{PS}, in the state-independent sender case where $V(y,x)=V(y)$. Note that the designer's preferences are state-independent because she cares only about the probability of winning each district, and not directly about a district's composition.

In the gerrymandering context, a districting plan is pairwise if each district contains at most two voter types, and a pairwise districting plan is strictly single-dipped if more polarized districts are more favorable for the designer (i.e., if $\gamma(\mu)>\gamma(\mu')$ for all $\mu,\mu' \in \supp (\tau)$ such that $\supp(\mu)$ contains $x_1 < x_3$ and $\supp(\mu')$ contains $x_2 \in (x_1 ,x_3)$.) Since $V$ is state-independent, Theorem \ref{t:SDPD} implies that strictly single-dipped districting is optimal whenever $u_x$ is strictly log-supermodular. This result generalizes a main result of \citet{FH} (their Lemma 1), which shows that strictly single-dipped districting is optimal under an ``informative signal property'' that is equivalent to log-supermodularity of $u_x$.\footnote{\citet{FH} additionally assume that there is a finite number of districts and that $u$ satisfies another condition called ``central unimodality.''} As explained in \citet{KW}, the intuition for this result is that log-supermodularity of $u_x$ means that moderate voters ``swing more'' with the aggregate shock $y$ than more extreme voters, so that a marginal voter is less likely to be pivotal in a district consisting of moderates than in a district that is evenly divided between left-wing and right-wing extremists. The designer then optimally exploits this difference in pivot probabilities by assigning more favorable marginal voters to more polarized districts: i.e., by creating a single-dipped districting plan. 

\citet{KW} go on to apply the duality and complementary slackness developed in the current paper to derive further properties of optimal districting plans, and in particular ask whether it is optimal for the designer to segregate the strongest opposing voters (as in traditional ``pack-and-crack'' districting plans) or more moderate voters (as in an alternative plan proposed by \citealt{FH}, which resembles negative assortative disclosure).\footnote{As argued by \citet{CH}, this is a key question for assessing the likely consequences of restrictions on districting such as those instituted by the Voting Rights Act of 1965.} The reader is referred to that paper for these additional results.

\subsection{Specific Persuasion Models}

Our analysis covers most persuasion models with non-linear preferences considered to date, including \citeauthor{ZZ}'s \citeyear{ZZ} model of information disclosure in contests; \citeauthor{GS2018}'s \citeyear{GS2018} model of persuading a receiver with affiliated private information; and \citeauthor{GL}'s \citeyear{GL} model of optimal stress tests. The main results in the latter two papers show that single-peaked negative assortative disclosure is optimal in their models. In Appendix \ref{s:persuasion}, we describe how our analysis covers these prior models and in some cases provides additional results.

\section{Conclusion}
\label{s:conclusion}

This paper has developed a general model of assigning inputs to productive outputs, which we call \emph{optimal productive transport}. Our leading application is Bayesian persuasion, but the model also covers other applications including matching, option pricing, and partisan gerrymandering. In the persuasion context, our substantive results provide conditions for all optimal signals to be pairwise, for riskier or safer prospects to induce higher actions, and for full or negative assortative disclosure to be optimal. In some cases, we can characterize optimal signals as the solution to a pair of ordinary differential equations, or even solve them in closed form. Methodologically, we develop novel duality and complementary slackness theorems, which form the basis of all of our proofs.

We mention a few open issues. First, while the persuasion literature has made progress by allowing unrestricted disclosure policies, the pairwise signals that we highlight are not always realistic. (For example, in reality it is probably not feasible to design a stress test that pools only the weakest and strongest banks.) An alternative, complementary approach is to restrict the sender to partitioning the state space into intervals, as in \citet{Rayo} and \citet{OR}. An interesting observation is that, at least in the separable subcase of our model considered by \citeauthor{Rayo} and \citeauthor{OR}, our condition \eqref{e:nd} is equivalent to the condition that complete pooling is uniquely optimal among monotone partitions for all prior distributions. This suggests that, under our conditions for the optimality of single-dipped/-peaked disclosure, negative assortative disclosure might be the optimal unrestricted disclosure policy for all priors iff no disclosure is the optimal monotone policy for all priors. More generally, analyzing the relationship between the optimal pairwise signals we have characterized and simpler signals such as monotone partitions is an important direction for future research.

Second, in the informed receiver interpretation of our model mentioned in Section \ref{s:model}, our analysis pertains to disclosure mechanisms that do not first elicit the receiver's type, or \emph{public persuasion} in the language of \citet{KMZL}. Public persuasion turns out to be without loss in \citet{KMZL}, as well as in \citet{GS2018}. It would be interesting to investigate conditions for the optimality of public persuasion in our more general model, and in particular to see how they relate to our conditions for the optimality of full or negative assortative disclosure.

Third, while we have taken some steps toward fully characterizing optimal signals by deriving the differential equations \eqref{e:obed}--\eqref{e:q'} and solving them in a couple examples, much more remains to be done. Equations \eqref{e:obed}--\eqref{e:q'} are closely related to the optimality and Monge-Ampere equations in optimal transport (e.g., Section 1.7.6 in \citealt{santambrogio}). The rich mathematical literature on these equations may hold some insights for fully characterizing optimal signals in certain settings.

Finally, our model could be generalized to allow multidimensional states or actions. We suspect that our results on duality (Lemma \ref{l:dual}), complementary slackness (Theorem \ref{t:contact}), and pairwise signals (Theorem \ref{t:pairwise}) generalize up to some technicalities.\footnote{For example, our proof of Theorem \ref{t:contact} is facilitated by the existence of a bijection between actions $y$ and states $\chi (y)$ such that $u(y,\chi (y))=0$, cf.\ Assumption \ref{a:ord}.} Generalizing our other results would require a more general notion of single-dippedness. With a unidimensional action and a multidimensional state, one can still define a notion of single-dippedness as inducing higher actions at more extreme states; with multidimensional actions, the appropriate generalization is unclear.\footnote{Possibly relevant recent work on multidimensional martingale optimal transport includes \citet{GKL} and \citet{DeMarchTouzi}.} For results on multidimensional persuasion focusing on the linear case, see \citet{DK}.

\bibliographystyle{econometrica}
\bibliography{persuasionlit}

\renewcommand{\thesection}{A}

\section{Characterization of Strict Aggregate Single-Crossing}
\label{a:ad} 

We present two alternative conditions that are equivalent to strict aggregate single-crossing of $u$. Condition (2) is analogous to the ``signed-ratio monotonicity'' conditions for weak aggregate single-crossing in Theorem 1 of \citet{Quah2012} and Corollary 2 of \citet{CS}. We give a shorter proof based on the optimality of pairwise signals (see Appendix \ref{proof:ASC}). Condition (3) is novel. It corresponds to strict monotonicity of $u$ (i.e., $u_y(y,x)<0$), up to a normalizing factor $g(y)>0$. 
\begin{lemma}\label{l:ASC}
Let Assumption \ref{a:smooth} hold. The following statements are equivalent:
\begin{enumerate}
	\item Assumption \ref{a:qc} holds.
	\item For all $x$, $x'$, and $y$, we have
\begin{align}
u(y,x)=0 &\implies u_y(y,x)<0,\label{1}	\\
u(y,x)<0<u(y,x') &\implies u(y,x')u_y(y,x)-u(y,x)u_y(y,x')<0.\label{2}
\end{align}
	\item There exists a differentiable function $g(y)>0$ such that $\tilde u(y,x)=u(y,x)/g(y)$ satisfies $\tilde u_y(y,x)<0$ for all $(y,x)$.
\end{enumerate} 
\end{lemma}

\renewcommand{\thesection}{B}

\section{Uniqueness} \label{s:Brenier}

This appendix presents a notable technical result: under a regularity condition, strict single-dippedness implies that there is a unique optimal signal. Moreover, it also shows that conditions for uniqueness are much weaker when full disclosure is optimal. 

We say that a strictly single-dipped set $\Lambda$ is \emph{regular} if for each $y\in Y_\Lambda$, there exists $\varepsilon>0$ such that either (i) $\chi_1(\tilde y) =\chi_2(\tilde y)$ for all $\tilde y\in (y-\varepsilon,y)\cap Y_\Lambda$ or (ii) $\chi_1(\tilde y)< \chi_2(\tilde y)$ for all $\tilde y\in (y-\varepsilon,y)\cap Y_\Lambda$. This regularity condition rules out pathological cases where states switch infinitely many times from being disclosed to being paired. This condition is satisfied in every example in the literature that we know of.

\begin{theorem}\label{t:brenier}
If $X=[0,1]$, $\phi $ has a density, and $\Lambda$ is strictly single-dipped and regular, then there is a unique optimal signal.
\end{theorem}

In martingale optimal transport, the optimal solution is unique under the martingale Spence-Mirrlees condition (e.g., Proposition 3.5 in \citealt{BHT}), which coincides with our condition for the optimality of strict single-dippedness in the linear receiver case. The key implication of Theorem \ref{t:brenier} is that the optimal marginal distribution of actions is unique. There is no analogue of this result in martingale optimal transport, where this marginal distribution is fixed.

To see the intuition for Theorem \ref{t:brenier}, consider the case where $\phi$ is discrete and $\chi_2$ is strictly increasing. Let $\ol y=\max_{\mu\in \Lambda} \gamma (\mu)$ be the highest action that can be optimally induced.  Since $\Lambda$ is strictly single-dipped, there is a unique posterior $\ol \mu$ in $\Lambda$ inducing action $\ol y$, namely $\ol \mu=\delta_{\chi_2 (\ol y)}$ if $\chi_1(\ol y)=\chi_2(\ol y)$ and $\ol \mu=\ol \rho \delta_{\chi_1(\ol y)}+(1-\ol \rho)\delta_{\chi_2(\ol y)}$ where $\ol \rho \in (0,1)$ is uniquely determined by \eqref{e:FOCy} if $\chi_1(\ol y)<\chi_2(\ol y)$. Since $\chi_2$ is strictly increasing, for any optimal signal $\tau$, the state $\chi_2(\ol y)$ can only induce action $\ol y$ and thus $\tau (\ol \mu)=\phi (\chi_2(\ol y))$ if $\chi_1(\ol y)=\chi_2(\ol y)$ and $\tau (\ol \mu)=\phi (\chi_2(\ol y))/(1-\ol \rho) $ if $\chi_1(\ol y)<\chi_2(\ol y)$. Working our way through the support of $\phi$ from the highest state to the lowest in this fashion, we obtain the unique value of $\tau (\mu)$ for each $\mu \in \Lambda$. When $\phi$ has a density, the possibility that $\chi_2$ may be only weakly increasing does not threaten uniqueness of the optimal signal, because the set of states corresponding to flat regions of $\chi_2$ is at most countable and thus has $\phi$-measure $0$. Finally, our regularity condition ensures that the above argument extends easily from the discrete case to the continuous one.

In addition, when the prior has full support and the contact set is pairwise (e.g., the twist condition holds), full disclosure is uniquely optimal whenever it is optimal. To see the intuition, suppose full disclosure is optimal, and suppose there is another optimal signal that pools some states $x_1$ and $x_2$ to induce an action $y$. Then the signal that discloses all other states while pooling $x_1$ and $x_2$ to induce $y$ is also optimal. But then the signal that discloses all other states while pooling $x_1$, $x_2$, and the third state $\chi (y)\neq x_1, x_2$  to induce $y$ would also be optimal---but this signal is not pairwise, which is a contradiction.
\begin{remark}\label{t:FDu} Assume that $X=[0,1]$. If the contact set is pairwise and full disclosure is optimal, then it is uniquely optimal.
\end{remark}

\renewcommand{\thesection}{C}

\section{Examples of Negative Assortative Disclosure} \label{s:examples}

\begin{example}[Solving the Differential Equations] \label{e:RS} Consider the linear receiver case with $Y=X=[1/e,e]$, $f(x)=1/(2x)$, and $V(y,x)=y/x$. We claim that the unique optimal outcome matches each state $x_1\in [1/e, 1]$ with state $x_2=1/x_1$ with equal weights, so that the induced action is $y={x_1}/{2}+{1}/({2x_1})$. Thus, $\chi_1(y)=y-\sqrt{y^2-1}$, and $\chi_2(y)=y+\sqrt{y^2-1}$ for all $y\in [1,{e}/{2}+{1}/(2e)]$.

 Indeed, by Theorem \ref{t:SDPD}, the optimal outcome is strictly single-dipped, since $w(x)=1/x$ is strictly convex. By Corollary \ref{c:ndSDD}, \eqref{e:nd} holds, since $w'<0$. Hence, by Theorem \ref{t:NAD}, the optimal outcome is single-dipped negative assortative and satisfies \eqref{e:obed}--\eqref{e:boundary}. Now, for $x_2=1/x_1$ and $y={x_1}/{2}+{1}/({2x_1})$, \eqref{e:obed} holds because
\begin{gather*}
u(y,x_2)=\left( \frac{1}{2x_1}-\frac {x_1}{2}\right)=-\left(\frac{x_1}{2}-\frac{1}{2x_1}\right)=-u(y,x_1),\\
f(x_2)\frac {\df x_2}{\df y}=\frac{1}{{2}/{x_1}} \left(-\frac {1}{x_1^2}\frac{\df x_1}{\df y}\right)=-\frac {1}{2x_1}\frac{\df x_1}{\df y}=-f(x_1)\frac{\df x_1}{\df y},
\end{gather*}
\eqref{e:q'} holds because
\begin{gather*}
\frac{\df}{\df y} \left(w(x_1)\frac{1}{2}+w(x_2)\frac 12\right)=\frac{\df}{\df y}\left(\frac{1}{2x_1}+\frac {x_1}{2}\right)=\frac{\df }{\df y}y=1,\\
\frac{w(x_2)-w(x_1)}{x_2-x_1}=\frac{x_1-{1}/{x_1}}{{1}/{x_1}-x_1}=-1,
\end{gather*}
and \eqref{e:boundary} holds because $1/(1/e)=e$ and $1/1=1$.\footnote{We can also solve this example by directly applying Theorem \ref{t:contact}), because, for $q(y)=y$, the function
$V(y,x)+q(y)u(y,x)={y}/{x} +y(x -y)$ is maximized at $y=x/2+1/(2x)$ for all $x \in [1/e,e]$.}
\end{example}
  
\begin{example}[Quantile Persuasion] \label{e:quantile} Consider the quantile sub-subcase of the state-independent sender case, where $u(y,x)= \1 \{x\geq y\}-\kappa$ with $\kappa\in (0,1)$.
Let $\phi$ have a density on $[0,1]$. Assuming that the receiver breaks ties in favor of the sender, we have, for $x_1<x_2$,
\[
\gamma(\rho \delta_{x_1}+(1-\rho )\delta_{x_2})=
\begin{cases}
x_2, &\rho \leq 1-\kappa,\\
x_1, &\rho >1-\kappa.
\end{cases}
\]
Note that \eqref{e:nd} always holds for $\rho \in (0,1-\kappa)$. We claim that there exists an optimal single-dipped negative assortative signal where the induced distribution over actions $\alpha$ satisfies $\alpha([y,1])=\phi ([y,1])/\kappa$, and the posterior inducing any action $y\in [\ul y,1]$ is $(1-\kappa)\delta_{\chi_1(y)}+\kappa \delta_{\chi_2(y)}$, where $\chi_2(y)=y$, $\chi_1(y)$ solves $\kappa \phi ([0,\chi_1(y)])=(1-\kappa)\phi ([y,1])$, and $\ul y$ solves $\kappa \phi ([0,\ul y])=(1-\kappa)\phi ([\ul y,1])$.\footnote{See Appendix \ref{s:quantileproof} for the proof.} A notable feature of this signal is that, with the informed receiver interpretation, it would remain optimal even if the sender knew the receiver's type and could condition disclosure on it.
\end{example}

\begin{example}[A Stochastic Optimal Signal] \label{ex:segpair}
In the following example, for some priors negative assortative disclosure is optimal; and for other priors, the unique optimal signal randomizes conditional on the state, even though the prior is atomless.

\begin{figure}[t]
\centering
\begin{tikzpicture}[scale = 1]
	\small
	\begin{axis}		
		[axis x line = middle,
		axis y line = middle,
		xmin = -1.3, xmax = 3.3,
		ymin = -1.3, ymax = 1.3,
		xlabel=$x$,
		ylabel=$y$,
		xtick={-1.01, 0, 3.01},
		xticklabels={$-1$, $0$, $3$},
		ytick={-1.01, 0, 1.01},
		yticklabels={$-1$, , $1$},
		clip=false]

		\draw [dashed, black]			(-1, -1) -- (-1, 1);
		\draw [dashed, black]			(-1, -1) -- (0, -1);
		\draw [dashed, black]			(3, 0) -- (3,1);

		\draw [thick, solid, black]		(-1,-1) -- (0,0);
		\draw [thick, solid, black]		(-1,1) -- (0,0);
		\draw [thick, solid, black]		(0,0) -- (3,1);

		\draw [solid, red]				(-1, 1) -- (3,1);
		\draw [solid, red]				(-3/4, 3/4) -- (9/4, 3/4);
		\draw [solid, red]				(-2/4, 2/4) -- (6/4, 2/4);
		\draw [solid, red]				(-1/4, 1/4) -- (3/4, 1/4);
	\end{axis}
\end{tikzpicture}

	\caption{The Optimal Signal in Example \ref{ex:segpair}}
	\caption*{\emph{Notes:} The three black line segments depict the graphs of $\chi_1$ and $\chi_2$. The red line segments indicate pairs of states $\chi_1(y)$ and $\chi_2(y)$ that are optimal to pool to induce action $y\in (0,1]$. If the prior density satisfies $f(-y)>3f(3y)$ for all $y\in (0,1]$, then, for each state $x <0$, the optimal signal randomizes between disclosing $x$ (and inducing action $x$) and pooling $x$ with state $-3x$ (and inducing action $-x$).}
	\label{f:segpair}
\end{figure}

Consider the translation-invariant subcase of the state-independent sender case. Let $Y=X=[-1,3]$, let $\phi$ have a density $f$ with $f(-y)\geq 3f(3y)$ for all $y\in (0,1]$, let $u(y,x)=T(x-y$) with $T(0)=0$ and strictly log-concave $T'$, and let $V(y,x)=T(2y)$. With the informed receiver interpretation, this captures a case where, for example, $\kappa=1/2$, the distribution of $\varepsilon$ is $N(0,\sigma^2)$, and the distribution of $t$ is $N(0,(\sigma/2)^2)$.\footnote{By symmetry and strict log-concavity of $T'$, $V_{yy}(y)/V_y(y)=2T''(2y)/T'(2y) >(<)T''(0)/T'(0)=0$ for $0>(<)y,$ showing that \eqref{e:ndSDD} fails for $y<0$, and thus Theorem \ref{t:NAD} does not apply.}

 We claim that 
\[
\chi_1(y)=
\begin{cases}
y, &y\in [-1,0],\\
-y, &y\in (0,1],
\end{cases}
\quad \text{and}\quad 
\chi_2(y)=
\begin{cases}
y, &y\in [-1,0],\\
3y, &y\in (0,1],
\end{cases}
\]
so that the posterior inducing any action $y\in [-1,1]$ is $\delta_{\chi_1(y)}/2+\delta_{\chi_2(y)}/2$, and the distribution over actions $\alpha$ has density $a$ given by
\begin{align*}
a(y)&=
\begin{cases}
6f(3y), &y\in (0,1],\\
f(-y)-3f(3y), &y\in [-1,0). 
\end{cases}
\end{align*}
The unique optimal signal is single-dipped negative assortative iff $f(-y)=3f(3y)$ for all $y\in (0,1]$. In contrast, if $f(-y)>3f(3y)$ for all $y\in (0,1]$, then each state $x\in [-1,0)$ is mixed between inducing actions $y=x$ and $y=-x$.\footnote{See Appendix \ref{s:segpair} for the proof.} %
See Figure \ref{f:segpair}.

\end{example}

\renewcommand{\thesection}{D}

\section{Specific Persuasion Models} \label{s:persuasion}

This appendix shows how our analysis covers some well-known prior persuasion models, where single-dipped or single-peaked disclosure is optimal.\footnote{The applications in this section also illustrate some technical points. Section \ref{s:ZZ} illustrates how directly applying Theorem \ref{l:ssdd} can yield weaker sufficient conditions for the optimality of single-dipped disclosure than those in Theorem \ref{t:SDPD}. Sections \ref{s:GS} and \ref{s:GL} illustrate how our analysis extends when some of our assumptions are violated: in Section \ref{s:GS}, Assumption \ref{a:int} fails, so the receiver's optimal action may be at the boundary and thus violate the first-order condition; in Section \ref{s:GL}, Assumption \ref{a:ord} fails, as the sender only weakly prefers higher actions.}

\subsection{Contests}\label{s:ZZ}\mbox{} \citet{ZZ} study information disclosure in contests. In their model, two contestants, $A$ and $B$, compete for a prize by exerting efforts $z_A$ and $z_B$. The probability that contestant $i=A,B$ wins is $z_i/(z_A+z_B)$. Everyone knows contestant A's value $v_A=1$. Contestant B's value $v_B$ is known to contestant B and the designer. The sender designs a signal about $v_B$ to maximize expected total effort.

It is convenient to parameterize $x = 1/\sqrt{v_{B}}$ and $y=\sqrt{z_{A}}$. With this parameterization, \citeauthor{ZZ}'s Proposition 1 shows that, given a posterior $\mu$, contestant A exerts effort $z_A^\star=\gamma (\mu)^2$ determined by $\E_{\mu} \left [x-\left(1+x^2 \right)\gamma (\mu) \right] =0$, and contestant B (who knows $x$) exerts effort $z_B^\star(x)=\gamma (\mu)/x-\gamma (\mu)^2$, so the sender's expected utility is $z_A^\star+\E_\mu \left[z_B^\star(x)\right]=\E_\mu \left[\gamma (\mu)/x \right]$. We thus recover our model with $V(y,x)=y/x$ and $u(y,x)=x-(1+x^2)y$. 

\citeauthor{ZZ} give results on optimality of pairwise disclosure, full disclosure, and no disclosure. Our approach easily yields the following result, which additionally gives conditions for optimality of single-dipped/-peaked disclosure and negative assortative disclosure (which were not considered by \citeauthor{ZZ}).%

\begin{proposition}\label{p:ZZ} In \citeauthor{ZZ}'s contest model where the prior $\phi$ has a positive density on $X=[\ul x ,\ol x]$, where $0<\ul x<\ol x$, if $\ul x\geq 1$ then the unique optimal signal is full disclosure; and if  $\ol x \leq 1/\sqrt{3}$ ($1/\sqrt{3}\leq \ul x<\ul x\leq 1$), then the unique optimal signal is single-dipped (-peaked) negative assortative disclosure.
\end{proposition}

The proof of single-dippedness/-peakedness uses Theorem \ref{l:ssdd} with a perturbation that fixes both actions. In contrast, directly applying Theorem \ref{t:SDPD} would yield only the weaker result that single-peaked negative assortative disclosure is optimal if $1/\sqrt 2\leq \ul x<\ol x< 1$.\footnote{To see this, suppose $\ol x< 1$. Then $u_x(y,x)=1-2x y>0$ for $y\leq \ol x/(1+\ol x\,\!^2)=\max Y$. Moreover, $u_{yx}(y,x)/u_x(y,x)=-2x/(1-2x y)$ is always decreasing in $x$, while $V_{yx}(y_2,x)/u_x(y_1,x)=-1/(x^2-2x^3 y_1)$ is decreasing in $x $ iff $3\ul x \min Y= 3\ul x^2/(1+\ul x\,\!^2)\geq  1$, or equivalently $\ul x \geq 1/\sqrt 2$.}

\subsection{Affiliated Information}\label{s:GS}\mbox{} \citet{GS2018} consider a persuasion model with a privately informed receiver, where it is commonly known that the receiver wishes to accept a proposal iff $x$ exceeds a threshold $x_0$, and the receiver's type $t$ is his private signal of $x$. Letting $G(t|x)$ denote the distribution of $t$ conditional on $x$, with corresponding density $g(t|x)$, this setup maps to our model with $V(y,x)=G(y|x)$, $u(y,x)=(x-x_{0})g(y|x)$, and $g(t|x)$ strictly log-submodular in $(t,x)$.\footnote{The ordering convention here is that high $t$ is bad news about $x$. This ordering is opposite to \citeauthor{GS2018}'s, but follows our convention that the receiver accepts for types below a cutoff.}$^,$\footnote{\citeasnoun{IP20} study robust stress test design in a setting with multiple receivers with coordination motives. As they note, the single-receiver version of their model is a special case of \citeasnoun{GS2018}.} These preferences satisfy
Assumptions \ref{a:smooth} and \ref{a:qc} (see Lemma \ref{l:ASC}), but not Assumption \ref{a:int}, as $u(y,x)>0$ for all $y$ when $x>x_0$. Nonetheless, assuming that the receiver breaks ties in the sender's favor, we have $\gamma (\mu)=\max\{y:\int u(y,x)\df \mu\geq 0\}$.

Let us take for granted that Theorem \ref{l:ssdd} holds even though Assumption \ref{a:int} is violated (e.g., this is clearly true with a discrete prior). Applying Theorem \ref{l:ssdd} with a perturbation that fixes one action while increasing the other action and the sender's expected utility (for fixed actions), we obtain the following result, which reproduces \citeauthor{GS2018}'s main qualitative insight.\footnote{When the prior has positive density on $[0,1]$, \citeauthor{GS2018}'s Theorem 3.1 additionally implies that the optimal signal is single-peaked negative assortative.}

\begin{proposition}\label{c:spd1}
In \citeauthor{GS2018}'s model of persuading a privately informed receiver, every optimal signal is single-peaked.
\end{proposition}

\subsection{Stress Tests}\label{s:GL}\mbox{}
\citeasnoun{GL} consider a model of optimal stress tests. The sender is a bank regulator and the receiver is a perfectly competitive market. The bank has an asset that yields a random cash flow. The asset's quality is $x$, which is observed by the bank and the regulator but not the market, and is normalized to equal the asset's expected cash flow.\footnote{This is the model in Section 5 of their paper, where the bank observes $x$.} The regulator designs a test to reveal information about $x$. After observing the test result, the market offers a competitive price $y$ for the asset. Finally, the bank decides whether to keep the asset and receive the random cash flow, or sell it at price $y$. Letting $z$ denote the bank's final cash holding (equal to either the random cash flow or $y$), the bank's payoff equals $z+\1 \{z \geq x_{0}\}$, where $x_{0}$ is a constant. An interpretation is that the bank faces a run if its cash holding falls below $x_{0}$. The regulator designs the test to maximize expected social welfare, or equivalently to minimize the probability of a run.

\citeauthor{GL} show that a bank with a type-$x$ asset is willing to sell at a price $y$ iff $y$ exceeds a reservation price $\tilde{\sigma} (x)$ that satisfies $\tilde{\sigma} (x) >x $ if $x < x_0$, $\tilde{\sigma} (x) <x $ if $x > x_0$, and $\tilde{\sigma}' (x)\geq 0$. Intuitively, if $x < x_0$ then the bank demands a premium to forego the chance that a lucky cash flow shock pushes its holdings above $x_0$, while if $x > x_0$ then the bank desires insurance against bad cash flow shocks that push its holdings below $x_0$. However, the value of the regulator's problem is unaffected if the reservation price is re-defined as $\sigma (x) =x$ if $x \leq x_0$ and $\sigma (x) =\tilde{\sigma} (x)$ if $x > x_0$, because it is suboptimal for the regulator to induce a bank to sell at a price below $x_0$. It is more convenient to work with the normalized reservation price $\sigma (x)$. 

It is also convenient to restrict attention to tests that, for each $x$, either induce the bank to sell or fully disclose the bank's value: this is without loss because pooling two asset types that do not sell is weakly worse than disclosing these types. For such a test, the price induced by any posterior $\mu$ is $\gamma (\mu)=\mathbb{E}_{\mu}[x]$, so we are in the linear receiver case. We can capture the requirement that the bank always sells if $y\neq x$ by setting $V(y,x)=- \infty$ if $y<\sigma(x)$. Finally, letting $w(x)>0$ equal the social gain when a bank sells a type-$x$ asset at a price above $x_0$ (which equals the probability that a type-$x$ asset yields a cash flow below $x_0$), we obtain the linear receiver case of our model with
\[
V(y,x)=
\begin{cases}
w(x)\1\{y\geq x_0\}, &\text{if $y\geq \sigma(x)$,}\\
-\infty, &\text{otherwise.}
\end{cases}
\]
Note that $V$ violates Assumptions \ref{a:smooth} and \ref{a:ord}, as it is discontinuous and only weakly increasing in $y$. Nonetheless, if we assume a discrete prior (as do \citeauthor{GL}), we recover their main qualitative insight.

\begin{proposition}\label{p:GL} 
In \citeauthor{GL}'s stress test model with a discrete prior, there exists an optimal single-dipped signal.
\end{proposition}
To prove the proposition, we use a perturbation that fixes both actions. Since $V$ is only weakly increasing, this perturbation now only weakly increases the sender's expected utility. Nonetheless, when $X$ is finite, repeatedly applying such perturbations eventually yields a single-dipped signal. We also note that, as \citeauthor{GL} show, if $\mathbb{E}_{\phi}[x]<x_{0}$---so that no disclosure does not attain the sender's first-best outcome---then every optimal signal is single-dipped.\footnote{A related model by \citet{GT} studies optimal information disclosure to facilitate trade in an insurance market with adverse selection. Their model can be mapped to the linear receiver case with $V(y,x)=\nu(y)$ if $y\geq \sigma (x)$ and $V(y,x)=-\infty$ otherwise, where $\nu(y)$ is a strictly increasing, strictly concave function, and $\sigma$ is a continuous, strictly increasing function that satisfies $\sigma(x)<x$. Considering a similar perturbation as in \citeauthor{GL} shows that single-dipped negative assortative disclosure is optimal in their model. We also mention \citet{LW}, where a bank regulator discloses information about the design of a stress test to induce banks to make socially desirable investments. In this model, single-peaked disclosure is optimal.}

\renewcommand{\thesection}{E}

\section{Proofs}
\subsection{Proof of Lemma \ref{l:dual}}
The proof of Lemma \ref{l:dual} remains valid without Assumption \ref{a:ord} and when $X$ is an arbitrary compact metric space. 

By Corollary 2 in \citet{DK}, it suffices to show that $W$ is Lipschitz on $\Delta( X)$, endowed with the Kantorovich-Rubinstein distance 
\[
d_{KR}(\mu,\eta)=\sup \left\{\int_X p(x)\df (\mu-\eta)(x):\, p\text{ is $1$-Lipschitz on $X$} \right\},\quad \text{for all $\mu,\eta\in \Delta(X).$}
\]
Recall that the Kantorovich-Rubinstein distance metrizes the weak* topology on $\Delta(X)$ (e.g., Theorem 6.9 in \citealt{villani}).

Let $L_{V_y}$, $L_{V_x}$, and $L_{u_x}$ be the maximum values of  $|V_y|$,  $|V_x|$, and $|u_x|$ on $Y\times [0,1]$, which are well-defined because $V_y$, $V_x$, and $u_x$ are continuous,  by Assumption \ref{a:smooth}, and $Y\times [0,1]$ is compact. Then $V(y,x)$ is $L_{V_y}$-Lipschitz in $y$ for all $x$, $V(y,x)$ is $L_{V_x}$-Lipschitz in $x$ for all $y$, and $u(y,x)$ is $L_{u_x}$-Lipschitz in $x$ for all $y$. Moreover, let $l_{u_y}$ be the minimum value of $-\int_X u_y(\gamma(\mu),x)\df\mu(x)$ on $\Delta(X)$, which is well-defined by Assumption \ref{a:smooth} and is strictly positive (i.e., $l_{u_y}>0$), by Assumption \ref{a:qc}. Note that $\gamma$ is $L_{u_x}/l_{u_y}$-Lipschitz on $\Delta(X)$, because, by the implicit function theorem, the derivative of  $\gamma (\mu+\rho (\eta-\mu))$ with respect to $\rho$ at any $\rho\in [0,1]$ and $\mu,\eta \in \Delta(X)$ satisfies
\begin{align*}
\left|\frac{\df }{\df \rho}\gamma (\mu+\rho (\eta-\mu))\right| &= \left|\frac{\int u(\gamma (\mu+\rho(\eta-\mu)),x)\df (\eta -\mu)(x)}{-\int u_y(\gamma (\mu+\rho(\eta-\mu)),x)\df (\mu+\rho(\eta -\mu))(x)}\right|\\
&\leq \left|\frac{1}{l_{u_y}}\int u(\gamma (\mu+\rho(\eta-\mu)),x)\df (\eta -\mu)(x)\right|\\
&\leq \frac{L_{u_x}}{l_{u_y}}d_{KR}(\eta,\mu),
\end{align*}
where the last inequality holds by the definition of $d_{KR}$ and $L_{u_x}$-Lipschitz continuity of $u(y,x)$ in $x$ for all $y$. Now, for any $\mu,\eta\in \Delta(X)$, we have
\begin{align*}
|W(\eta)-W(\mu)|&=\left|\int (V(\gamma(\eta),x)-V(\gamma(\mu),x))\df \eta(x)+\int V(\gamma(\mu),x)\df (\eta-\mu)(x) \right|\\
&\leq \int \left|V(\gamma(\eta),x)-V(\gamma(\mu),x))\right|\df \eta(x)+\left|\int V(\gamma(\mu),x)\df (\eta-\mu)(x)) \right|\\
&\leq L_{V_y}\frac{L_{u_x}}{l_{u_y}}d_{KR}(\eta,\mu) +L_{V_x}d_{KR}(\eta,\mu),
\end{align*}
showing that $W$ is Lipschitz on $\Delta(X)$.

\subsection{Proof of Lemma \ref{l:Deq}}
The proof of Lemma \ref{l:Deq} remains valid without Assumption \ref{a:ord} and when $X$ is an arbitrary compact metric space. 

One direction is obvious. If $(p,q)$ is feasible for \eqref{DO}, then, for all $\mu\in \Delta (X)$, we have
\begin{gather*}
\int_X p(x)\df \mu(x) \geq \int_X (V(\gamma(\mu),x)+q(\gamma(\mu))u(\gamma(\mu),x))\df \mu(x)\\
=\int_X V(\gamma(\mu),x)\df \mu(x)=W(\mu),	
\end{gather*}
so $p$ is feasible for \eqref{DS}.

Suppose now that $p$ is feasible for \eqref{DS}. By \eqref{DS1} for $\mu=\delta_x$, with $x\in X$, we have 
\begin{equation}\label{e:p>=V}
p(x)\geq V(\gamma(\delta_x)),\quad \text{for all $x\in X$}.
\end{equation}
By \eqref{DS1} for
\[
\mu=\frac{u(y,x_2)}{u(y,x_2)-u(y,x_1)}\delta_{x_1}+\frac{-u(y,x_1)}{u(y,x_2)-u(y,x_1)}\delta_{x_2},
\]
with $x_1,x_2\in X$ and $y\in Y$ such that $u(y,x_1)<0<u(y,x_2)$, we have
\begin{equation}\label{e:p>=V+qu}
\frac{p(x_1)-V(y,x_1)}{u(y,x_1)}\geq \frac{p(x_2)-V(y,x_2)}{u(y,x_2)},\quad \text{if $u(y,x_1)<0<u(y,x_2)$}.
\end{equation}
By \eqref{e:p>=V}, it suffices to show that there exists a bounded, measurable $q$ such that \eqref{DO1} holds for all $y\neq \gamma(\delta_x)$, or equivalently $q(y)\in [\ul q(y),\ol q(y)]$, for all $y\in Y$, where $\ul q$ and $\ol q$ are defined by
\begin{align*}
\ul q(y) &=
\begin{cases}
\sup_{x_1\in X: u(y,x_1)<0}\frac{p(x_1)-V(y,x_1)}{u(y,x_1)}, &y>0,\\
-\infty, &y=0,
\end{cases}\\
\ol q(y) &=
\begin{cases}
\inf_{x_2\in X: u(y,x_2)>0}\frac{p(x_2)-V(y,x_2)}{u(y,x_2)}, &y<1,\\
+\infty, &y=1,
\end{cases}
\end{align*}
That is, it suffices to show that the correspondence $Q$ given by $Q(y)=[\ul q(y),\ol q(y)]$, for all $y\in Y$, admits a bounded, measurable selection $q$. We now show that there exists $C>0$ such that $Q(y)\cap [-C,C]$ is non-empty valued for all $y\in Y$, so $q$ given by $q(y)=\argmin_{r\in Q(y)}|r|$, for all $y\in Y$, is a required selection, by the measurable maximum theorem (Theorem 18.19 in \citealt{aliprantis2006}). By \eqref{e:p>=V+qu}, $Q(y)$ is non-empty for all $y\in (0,1)$, so it suffices to show that there exists $C>0$ such that $\ul q(y)\leq C$ and $\ol q(y)\geq -C$.

Define 
\begin{equation*}
\widetilde{q}(y,x)=
\begin{cases}
\frac{V_y(y,x)}{-u_{y}(y,x)}, & u(y,x)=0, \\
\frac{V(\gamma(\delta _{x}),x)-V(y,x)}{u(y,x)}, & 
u(y,x)\neq 0.
\end{cases}
\end{equation*}
Recall that Assumption \ref{a:qc} requires
that $u_{y}(y,x)<0$ when $u(y,x)=0$; so $\tilde{q}(y,x)$
is well-defined. It suffices to show that there exists $C>0$ such that $|\widetilde q(y,x)|\leq C$ for all $(y,x)\in Y\times X$, because then, by \eqref{e:p>=V}, we have
\begin{gather*}
\frac{p(x_2)-V(y,x_2)}{u(y, x_2)}\geq \frac{V(\gamma(\delta_{x_2}),x_2)-V(y,x_2)}{u(y,x_2)}=\widetilde q(y,x_2),\quad \text{if $u(y,x_2)>0$},\\
\frac{p(x_1)-V(y,x_1)}{u(y,x_1)}\leq \frac{V(\gamma(\delta_{x_1}),x_1)-V(y,x_1)}{u(y,x_1)}=\widetilde q(y,x_1),\quad \text{if $u(y,x_1)<0$},
\end{gather*}
so, by the definition of $\ol q$ and $\ul q$, we get that $\ol q(y)\geq -C$ and $\ul q(y)\leq C.$

Finally, to show that there exists $C>0$ such that $|\widetilde q(y,x)|\leq C$ for all $(y,x)\in Y\times X$, it suffices to show that $\widetilde q$ is continuous on the compact set $Y\times X$. By Berge's theorem, $\gamma(\delta _{x })$ is continuous in $x$, as it is a unique maximizer of a continuous function $U(y,x )=\int_{0}^y u(\tilde y,x)\df \tilde y$. Note that $\widetilde q$ is continuous at each $(y,x)$ such that $u(y,x)\neq 0$, because $V$, $u$, and $\gamma$ are continuous.
Next, consider $(y,x)$ such that $u(y,x)=0$, or equivalently $y=\gamma(\delta_x)$. For each $(y',x')\in Y\times X$, there exists $\hat y$ between $\gamma(\delta_{x'})$ and $y'$ such that
\[
[V(\gamma(\delta_{x'}),x')-V(y',x')]u_y(\hat y,x')=-V_y(\hat y,x')u(y',x'),
\]
by the mean value theorem applied to the function  
\[
[V(\gamma(\delta_{x'}),x')-V(\tilde y,x')]u(y',x')-[V(\gamma(\delta_{x'}),x')-V(y',x')]u(\tilde y,x'),
\]
where the argument $\tilde y$ is between $\gamma(\delta_{x'})$ and $y'$. Thus,
\[
\widetilde q(y',x')-\widetilde q(y,x)=\frac{V_y(\hat y,x')}{-u_y(\hat y,x')}-\frac{V_y(y,x)}{-u_y(y,x)}.
\]
If $(y',x')\rightarrow (y,x)$ then $(\hat y,x')\rightarrow (y,x)$, because $\gamma(\delta_x)$ is continuous in $x$. Hence, $\widetilde q(y',x')\rightarrow \widetilde q(y,x)$, because $V_y$ and $u_y$ are continuous. This shows that $\widetilde{q}$ is continuous on $Y\times X$.

\subsection{Proof of Theorem \ref{t:contact}}
The proof of Theorem \ref{t:contact} remains valid if Assumption \ref{a:ord} is replaced with the weaker requirement that $u(y,x)$ satisfies strict single-crossing in $x$: for all $y$ and $x<x'$, we have $u(y,x)\geq 0 \implies u(y,x')>0$.

By Lemma \ref{l:Deq}, there exists $q\in B(Y)$ such that $(p,q)$ is feasible for \eqref{DO}. Recall that, under strict single-crossing of $u(y,x)$ in $x$, for each action $y$, there is a unique state $\chi(y)$ such that $u(y,\chi(y))=0$.
 First, redefine $q(y)=-V_y(y,\chi(y))/u_y(y,\chi(y))$ for all $y\in Y$ such that $p(\chi(y))=V(y,\chi(y))$. We now show that $(p,q)$ is still feasible for \eqref{DO}. Fix any $y$ such that $p(\chi(y))=V(y,\chi(y))$ and any $x\in X$. For any $\varepsilon \in (0,1)$, define $y_{\varepsilon}\in Y$ as a unique solution to $(1-\varepsilon)u(y_\varepsilon ,\chi (y))+\varepsilon u(y_\varepsilon,x )=0$.
By the implicit function theorem,
\[
\lim_{\varepsilon \downarrow 0}\frac{y_\varepsilon -y}{\varepsilon} = \frac{u(y,x)}{-u_y(y,\chi (y))}.
\]
By \eqref{DO1}, we have
\begin{align*}
	V(y,\chi(y))\geq V(y_\varepsilon,\chi(y))+q(y_\varepsilon) u(y_\varepsilon,\chi(y)) \quad \text{and} \quad
	p(x)\geq V(y_\varepsilon,x)+q(y_\varepsilon) u(y_\varepsilon,x).
\end{align*}
Adding the first inequality multiplied by $1-\varepsilon$ and the second inequality multiplied by $\varepsilon$, and taking into account the definition of $y_\varepsilon$, we get
\[
p(x) \geq V(y,x) +\frac{(1-\varepsilon) [V(y_\varepsilon,\chi(y))-V(y,\chi(y))]+\varepsilon [V(y_\varepsilon,x)-V(y,x)]}{\varepsilon}.
\]
Taking the limit $\varepsilon \rightarrow 0$ gives
\[
p(x)\geq V(y,x)+\frac{V_y(y,\chi(y))}{-u_y(y,\chi(y))}u(y,x),
\]
showing that $(p,q)$ is still feasible with redefined $q$.

Thus, \eqref{e:q} holds, by construction, for all degenerate $\mu\in \Lambda$. Since, for non-degenerate $\mu\in \Lambda$, \eqref{e:FOC} integrated over $\mu$ yields \eqref{e:q}, it remains to show that \eqref{e:FOC} holds for each non-degenerate $\mu\in \Lambda$.

Fix a non-degenerate $\mu\in \Lambda$, so that
\begin{gather*}
\int_X (p(x)-V(\gamma(\mu),x)-q(\gamma(\mu))u(\gamma(\mu),x))\df \mu(x)=0.
\end{gather*}
Since the integrand is non-negative and continuous in $x$, it follows that
\begin{equation}\label{e:p=V+qu}
p(x)=V(\gamma(\mu),x)+q(\gamma(\mu))u(\gamma(\mu),x), \quad \text{for all $x\in \supp (\mu)$}.
\end{equation}
Since $\mu$ is non-degenerate, strict single-crossing of $u(y,x)$ in $x$ implies that there exist $x_1,x_2\in \supp (\mu)$ such that $x_1<\chi(\gamma(\mu))<x_2$. Thus, by \eqref{DO1}, for every $\tilde y \in Y$, we have
\[
p(x_1)=V(\gamma(\mu),x_1)+ q(\gamma(\mu))u(\gamma(\mu),x_1) \geq  V(\tilde y,x_1)+{q}(\tilde y)u(\tilde y,x_1).
\]
Therefore, for every $\tilde y > \gamma(\mu)$, we have
\[
\frac{{q}(\tilde y)-{q}(\gamma(\mu))}{\tilde y -\gamma(\mu)}\geq 
\frac{1}{-u(\tilde y,x_1)}\left[ \frac{V(\tilde y,x_1)-V(\gamma(\mu),x_1)}{\tilde y-\gamma(\mu)}+q(y )\frac{u(\tilde y,x_1)-u(\gamma(\mu),x_1)}{\tilde y-\gamma(\mu)}\right] .
\]
Since $V$ and $u$ have continuous partial derivatives in $y$, we have 
\[
\underline{q}_+'(\gamma(\mu)):=\liminf_{\tilde y
\downarrow \gamma(\mu) }\frac{{q}(\tilde y)-{q}(\gamma(\mu))}{\tilde y -\gamma(\mu)}\geq {C_1},
\]%
where 
\[
{C_1}=-\frac{1}{u(\gamma(\mu) ,x_1)}[V_y(\gamma(\mu) ,x_1)+q(\gamma(\mu))u_y(\gamma(\mu),x_1)].
\]
Applying a similar argument for $x=x_1$ and $\tilde y<\gamma(\mu) $, we get 
\[
\overline{q}_-'(\gamma(\mu) ):=\limsup_{\tilde{y}\uparrow \gamma(\mu)}\frac{q(\tilde y)-q(\gamma(\mu))}{\tilde y-\gamma(\mu) }\leq {C_1}.
\]
Similarly, considering $x=x_2$ with $\tilde y>\gamma(\mu) $
and $\tilde y<\gamma(\mu) $, we get 
\begin{align*}
\overline{q}_+'(\gamma(\mu))&: =\limsup_{\tilde y\downarrow \gamma(\mu)}\frac{q(\tilde y)-q(\gamma(\mu))}{\tilde y-\gamma(\mu) }\leq {C_2},\\
\underline{q}_-'(\gamma(\mu))&: =\liminf_{\tilde y\uparrow \gamma(\mu)}\frac{q(\tilde y)-q(\gamma(\mu))}{\tilde y- \gamma(\mu)}\geq {C_2},
\end{align*}
where 
\[
{C_2}=-\frac{1}{u(\gamma(\mu) ,x_2)}[V_y(\gamma(\mu) ,x_2)+q(\gamma(\mu))u_y(\gamma(\mu) ,x_2)].
\]
In sum, we have
\[
{C_1}\leq \underline{{q}}_+^{\prime }(\gamma(\mu) )\leq \overline{{q}}^{\prime }_+(\gamma(\mu))\leq {C_2}  \quad \text{and} \quad {C_2}\leq \underline{{q}}_-^{\prime }(\gamma(\mu))\leq \overline{{q}}_-^{\prime }(\gamma(\mu) )\leq {C_1}.
\]
We see that $ C_1 =C_2$ and all four Dini derivatives of $q$ at $\gamma(\mu)$ coincide, so ${q}$ has a derivative ${q}'(\gamma(\mu))$
at $\gamma(\mu)$ that satisfies $q^{\prime }(\gamma(\mu) )={C_1}={C_2}$. 

Since $x_1,x_2\in \supp (\mu)$ are arbitrary, \eqref{e:FOC} holds for all $x\in \supp (\mu)$ with $x\neq \chi(\gamma(\mu))$. For $x\in \supp (\mu)$ with $x=\chi(\gamma(\mu))$, \eqref{e:FOC} holds because, as shown above, we have that $q(\gamma(\mu))=-V_y(\gamma(\mu),\chi(\gamma(\mu)))/u_y(\gamma(\mu),\chi(\gamma(\mu)))$.

\subsection{Proof of Remark \ref{r:unique}} \label{a:r:unique}
The proof of Remark \ref{r:unique} remains valid if Assumption \ref{a:ord} is replaced with strict single-crossing of $u(y,x)$ in $x$.

Recall that to define $\Lambda$ we took an arbitrary solution $p$ to \eqref{DS}. Also, recall \eqref{e:p=V+qu} stating that
\[
p(x)=V(\gamma(\mu),x)+q(\gamma(\mu))u(\gamma(\mu),x),\quad \text{for all $\mu \in \Lambda$ and $x\in \supp (\mu)$}.
\]
Fix any solution $\tau$ to \eqref{PS}, so $\supp (\tau)\subset \Lambda$. Let $X^\star =\cup_{\mu\in \supp (\tau)} \supp (\mu)$. Then, by \eqref{PS1}, we have $\phi (X^\star)=1$, so the closure of $X^\star$ is $X$.

Next, take any $x\in X^\star$. If there is $\mu \in \supp (\tau)$ and $x\in \supp (\mu)$ such that $\gamma(\delta_x)=\gamma (\mu)$, then $p(x)=V(\gamma(\delta_x),x)$. Otherwise, there is $\mu\in \supp (\tau)$ and $x,x'\in \supp (\tau)$ such that either $x<\chi(\gamma(\mu))<x'$ or $x'<\chi(\gamma(\mu))<x$. Suppose that $x<\chi(\gamma(\mu))<x'$ (the other case is analogous and omitted). By Theorem \ref{t:contact}, we have
\begin{align*}
V_y(\gamma(\mu),x)+q(\gamma(\mu))u_y(\gamma(\mu),x)+q'(\gamma(\mu))u(\gamma(\mu),x) &=0,\\
V_y(\gamma(\mu),x')+q(\gamma(\mu))u_y(\gamma(\mu),x')+q'(\gamma(\mu))u(\gamma(\mu),x') &=0.	
\end{align*}
Adding the first equation multiplied by $u(\gamma(\mu),x')$ and the second multiplied by $-u(\gamma(\mu),x)$, we obtain
\[
q(\gamma(\mu))=-\frac{v(\gamma(\mu),x)u(\gamma(\mu),x')-v(\gamma(\mu),x')u(\gamma(\mu),x)}{u_y(\gamma(\mu),x)u(\gamma(\mu),x')-u_y(\gamma(\mu),x')u(\gamma(\mu),x)},
\]
which is well-defined because the denominator is strictly negative by Assumption \ref{a:qc}. Consequently, $p(x)=V(\gamma(\mu),x)+q(\gamma(\mu))u(\gamma(\mu),x)$. In sum, for each $x\in X^\star$, an arbitrary solution $p$ to \eqref{DS} is determined by a fixed solution $\tau$ to \eqref{PS}. Moreover, since $X$ is the closure of $X^\star$, there is a unique continuous extension of $p$ from $X^\star$ to $X$. This shows that there is a unique $p\in L(X)$ that solves \eqref{DS}.

\subsection{Proof of Proposition \ref{p:dekel}}
Define the weak order $\succsim$ on $\Delta(X)$ by $\mu\succsim \eta$ if $\tilde\gamma(\mu)\geq \tilde \gamma (\eta)$. Clearly, $\tilde \gamma$ satisfies Betweenness and Continuity iff $\succsim$ satisfies Axioms A1(a), A2', and A4 in \citet{dekel}. Thus, by his Proposition A.1, together with the characterization (**) in Proposition 1 and the argument in Section 3.C, $\tilde \gamma$ satisfies Betweenness and Continuity iff there exists a continuous function $\hat u$ from $[0,1]\times X$ to $[0,1]$ such that $\mu \succsim \eta \iff \hat \gamma(\mu)\geq \hat \gamma(\mu)$, where $\hat  \gamma(\mu)$, for any $\mu\in \Delta(X)$, is defined implicitly as the unique $\hat y\in [0,1]$ satisfying
\[
\int_X \hat u(\hat y,x)\df \mu (x) =(<)\hat y \iff y=(>)\hat y.
\]
In \citeauthor{dekel}'s construction, $\hat \gamma$ is continuous and thus there exists a continuous and strictly increasing function $\check \gamma $ such that $\hat \gamma (\mu) = \check \gamma (\tilde \gamma (\mu))$ for all $\mu\in \Delta(X)$. Then, $u$ defined by $u(y,x)=\hat u(\check \gamma(y),x)-\check \gamma(y)$ for all $(y,x)$ is as stated in the proposition. 

\subsection{Proof of Theorem \ref{p:S}}
We first prove the second part where the twist condition holds. The proof of this part remains valid if Assumption \ref{a:ord} is replaced with strict single-crossing of $u(y,x)$ in $x$.

Suppose by contradiction that there exists $\mu\in \Lambda$ with $x_1<x_2<x_3$ in $\supp (\mu)$. Then, by the definition of $\gamma(\mu)$ and strict single-crossing of $u(y,x)$ in $x$, we have $\min  \supp(\mu)<\chi(\gamma(\mu))<\max \supp (\mu)$. Thus, by redefining $x_1=\min  \supp(\mu)$ and $x_3=\max \supp (\mu)$ if necessary, we can assume that $x_1<\chi(\gamma(\mu))<x_3$, so \eqref{e:sing} holds. But this implies that the rows of the matrix $S$ are linearly independent, which contradicts the fact that \eqref{e:FOC} holds at $(\gamma(\mu),x_1)$, $(\gamma(\mu),x_2)$, and $(\gamma(\mu),x_3)$. Thus, $|\supp(\mu)|\leq 2$ for all $\mu \in \Lambda$, implying that every optimal signal is pairwise. 

We now turn to the first part. The proof of this part does not require Assumption \ref{a:ord}, and it remains valid when $X$ is an arbitrary compact metric space. 

For any $\mu \in \Delta (X )$, denote the set of distributions of posteriors with average posterior equal to $\mu$ by
\begin{equation*}
\Delta _{2}\left( \mu \right) =\left\{ \tau \in \Delta (\Delta (X)):\int_{\Delta (X )}\eta \df \tau \left( \eta \right) =\mu \right\}.
\end{equation*}

Let $\Delta _{2}^{Bin}(\mu )\subset \Delta _{2}(\mu)$ denote the set of such distributions where in addition the posterior is always supported on at most two states: 
\begin{equation*}
\Delta _{2}^{Bin} (\mu)=\left\{ \tau \in \Delta _{2}(\mu):\supp(\tau)\subset  \Delta_{1}^{Bin} \right\},
\end{equation*}
where
\begin{equation*}
\Delta _{1}^{Bin}=\left\{ \eta \in \Delta (X):\left\vert \supp(\eta)\right\vert \leq 2\right\}.
\end{equation*}

We wish to show that for each $\tau\in \Delta_2(\phi)$, there exists $\hat \tau \in \Delta_2^{Bin}(\phi) $ such that $\pi_{\hat \tau}=\pi_\tau$. 

We set the stage by defining some key objects and establishing their properties. Define $\Delta _{1}=\Delta \left( X \right) $ and $\Delta_{2}=\Delta \left( \Delta \left( X \right) \right) $. Since $X$ is compact, the sets $\Delta_1$ and $\Delta_2$ are also compact (in the weak* topology), by Prokhorov's Theorem (Theorem 15.11 in \citealt{aliprantis2006}). Moreover, $\Delta_{2}\left( \mu \right)$ is compact, since it is a closed subset of the compact set $\Delta_2$.

Define the correspondence $P:\Delta _{1}\rightrightarrows \Delta _{1}$ as%
\begin{equation*}
P\left( \mu \right) =\left\{ \eta \in \Delta _{1}:\int u\left( \gamma \left( \mu \right) ,x \right) \df \eta \left( x \right) =0 \right\}.
\end{equation*}%
For each $\mu \in \Delta _{1}$, $P\left( \mu \right) $ is a \emph{moment set}---a set of probability measures $\eta \in \Delta _{1}$ satisfying a given moment condition (e.g., \citealt{winkler}). By Assumption \ref{a:qc}, we have, for all $\mu ,\eta \in \Delta _{1}$,
\begin{equation}  \label{e:Prho*}
\eta \in P\left( \mu \right) \iff \gamma\left( \mu \right) =\gamma\left( \eta \right).
\end{equation}%
Clearly, $P(\mu)$ is nonempty (as $\mu\in P(\mu)$) and convex. Since $u(y,x)$ is continuous in $x$, $P(\mu)$ is a closed subset of $\Delta_1$, and hence is compact. Moreover, the correspondence $P$ has a closed graph. Indeed, consider two sequences $\mu_n\rightarrow \mu \in \Delta_1$ and $\eta_n\rightarrow \eta\in \Delta_1$ with $\mu_n\in \Delta_1$ and $\eta_n\in P(\mu_n),$ so that
\begin{equation*}
\int u\left(\gamma\left( \mu _{n}\right) ,x \right) \df \eta_{n}\left( x \right) =0.
\end{equation*}
Note that $\gamma(\mu)$ is a continuous function of $\mu$, by Berge's theorem (Theorem 17.31 in \citealt{aliprantis2006}). Since $u$ is also continuous, by Corollary 15.7 in \citet{aliprantis2006}, we have
\begin{equation*}
\int u\left( \gamma\left( \mu \right) ,x \right) \df \eta \left(x \right) =0,
\end{equation*}
proving that $\eta\in P(\mu) $, so $P$ has a closed graph. 

Define the correspondence $E:\Delta _{1}\rightrightarrows \Delta _{1}$ as%
\begin{equation*}
E\left( \mu \right) =P\left( \mu \right) \cap \Delta _{1}^{Bin}=\left\{ \eta\in P\left( \mu \right) :\left\vert \supp \eta \right\vert \leq 2\right\}.
\end{equation*}
Notice that for each $\mu\in \Delta_1$, the support of $\mu$ is well-defined, by Theorem 12.14 in \citet{aliprantis2006}. Moreover, from the proof of Theorem 15.8 in \citet{aliprantis2006}, it follows that $\Delta_1^{Bin}$ is a closed subset of $\Delta_1$, so both $\Delta_1^{Bin}$ and $E(\mu)$ are compact.

Define the correspondence $\Sigma :\Delta _{1}\rightrightarrows \Delta _{2}$ as
\begin{equation*}
\Sigma \left( \mu \right) =\left\{ \sigma \in \Delta \left( E\left( \mu \right) \right) :\mu =\int_{E\left( \mu \right) }\eta \df \sigma \left( \eta \right) \right\} .
\end{equation*}
Lemma \ref{lm_binary} shows that the correspondence $\Sigma$ admits a measurable selection. In turn, Lemma \ref{lm_binary} relies on Lemma \ref{l:Choquet}.
\begin{lemma}\label{l:Choquet}
Let Assumptions \ref{a:smooth} and \ref{a:qc} hold. For any $y\in Y$ and $\mu\in \Delta(X)$ such that $\int u(y,x)\df\mu(x)=0$, there exists $\sigma_\mu \in \Delta (\Delta(X))$ such that $\int \eta \df\sigma_\mu(\eta) =\mu$ and for each $\eta\in \supp (\sigma_\mu)$ we have $\int u(y,x) \df\eta(x)=0$ and $|\supp (\eta)|\leq 2$.
\end{lemma}
\begin{proof}
Follows immediately from the Choquet Theorem (Theorem 3.1 in \citealt{winkler}) and Richter-Rogosinsky's Theorem (Theorem 2.1 in  \citealt{winkler})	
\end{proof}

\begin{lemma}\label{lm_binary} There exists a measurable function $\mu \mapsto \sigma_{\mu }\in \Sigma \left( \mu \right) $. 
\end{lemma}

\begin{proof}
The correspondence $\Sigma $ is nonempty-valued, by Lemma \ref{l:Choquet}. Next, fix $\mu \in \Delta _{1}$, and consider a sequence $\sigma_{n}\rightarrow \Sigma\in \Delta_2 $ with $\sigma _{n}\in \Sigma \left(\mu \right)$. By the Portmanteau Theorem (Theorem 15.3 in \citealt{aliprantis2006}), we have
\begin{equation*}
\int_{E\left( \mu \right) }\eta \df \sigma _{n}\left( \eta \right) \rightarrow
\int_{E\left( \mu \right) }\eta \df \sigma \left( \eta \right) \quad \text{and} \quad 
\lim \sup_{n}\sigma _{n}\left( E\left( \mu \right) \right) \leq \sigma
\left( E\left( \mu \right) \right),
\end{equation*}%
where the last inequality holds because $E\left( \mu \right) $ is closed.
Thus, 
\begin{equation*}
\int_{E\left( \mu \right) }\eta \df \sigma \left( \eta \right) =\mu \quad \text{and} \quad 1=\lim \sup_{n}\sigma _{n}\left( E\left( \mu \right) \right) \leq \sigma \left( E\left( \mu \right) \right) \leq 1,
\end{equation*}
proving that $\sigma \in \Sigma \left( \mu \right) $. Thus, $\Sigma $ is closed-valued.

Next, consider two sequences $\mu_n\rightarrow \mu \in \Delta_1$ and $\sigma _n\rightarrow \sigma\in \Delta_2$ with $\mu_n\in \Delta_1$ and $\sigma_n\in \Sigma(\mu_n)$, so that 
\begin{gather*}
\mu _{n}=\int \eta \df \sigma _{n}\left( \eta \right), \quad \sigma _{n}\left( \Delta _{1}^{Bin}\right) =1, \quad \text{ and } \quad \sigma _{n}\left( P\left( \mu _{n}\right) \right) =1.
\end{gather*}
The Portmanteau Theorem implies that $\mu =\int \eta \df \sigma \left( \eta \right)$ and $\sigma \left( \Delta_{1}^{Bin}\right) =1$, since $\Delta_{1}^{Bin}$ is closed. Define $\overline P(\mu_n)$ as the closure of $\cup_{k=n}^\infty P(\mu_k)$. By construction, $P(\mu_k)\subset \overline P(\mu_k)\subset \overline P(\mu_n)$ for $k\geq n$, so the Portmanteau Theorem implies that $\sigma (\overline P(\mu_n))=1.$ Moreover, $\overline P (\mu_n)\downarrow \overline P \subset P(\mu)$, because $P$ has a closed graph. Hence, $\sigma (P(\mu))=1$, by the continuity of probability measures (Theorem 10.8 in \citealt{aliprantis2006}). That is, $\sigma \in \Sigma (\mu)$, showing that the correspondence $\Sigma$ has a closed graph.

Therefore, $\Sigma $ is measurable, by Theorem 18.20 in \citet{aliprantis2006}, as well as nonempty- and closed-valued. Hence, there exists a measurable function $\mu \mapsto \sigma _{\mu }\in \Sigma \left( \mu \right) $, by Theorem 18.13 in \citet{aliprantis2006}.
\end{proof}

Finally, taking a measurable selection, for each $\tau \in \Delta _{2}\left( \phi \right) $, define $\hat{\tau}\in \Delta _2 $ as
\begin{equation}  \label{e:tauhat}
\hat{\tau}\left( \widetilde \Delta_1\right) =\int_{\Delta _{1}}\sigma _{\mu}\left( \widetilde \Delta_1\right) \df \tau ( \mu )
\end{equation}
for every measurable set $\widetilde \Delta_1\subset \Delta _{1}$. By construction, $\hat{\tau}\in \Delta _{2}^{Bin}(\phi)$, since 
\begin{align*}
\hat{\tau}(\Delta_{1}^{Bin})=\int_{\Delta_{1}}\sigma_{\mu}(\Delta_{1}^{Bin})\df \tau(\mu)=1
\end{align*}
and 
\begin{align*}
\phi=\int_{\Delta_1} \mu \df \tau (\mu) = \int_{\Delta_1}\left(\int_{E(\mu)}\eta \df \sigma_\mu(\eta)\right) \df \tau(\mu)=\int_{\Delta_1} \eta \df \hat \tau (\eta),
\end{align*}
where the first equality holds by $\tau \in \Delta_2(\phi)$, the second by $\sigma_\mu\in \Sigma$, and the third by \eqref{e:tauhat}. Similarly, for each measurable $\widetilde Y\subset Y$ and $\widetilde X\subset X$, we have 
\begin{align*}
\pi_{\tau} (\widetilde Y,\widetilde X) &=\int_{\Delta_1}\1\{\gamma(\mu)\in \widetilde Y\}\mu(\widetilde X)\df \tau(\mu) \\
&=\int_{\Delta_1}\1\{\gamma(\mu)\in \widetilde Y\}\left(\int_{E(\mu)}\eta(\widetilde X)\df \sigma_\mu(\eta)\right)\df \tau(\mu) \\
&=\int_{\Delta_1}\left(\int_{E(\mu)}\1\{\gamma(\eta)\in \widetilde Y\}\eta(\widetilde X)\df \sigma_\mu(\eta)\right)\df \tau (\mu) \\
&=\int_{\Delta_1} \1\{\gamma(\eta)\in \widetilde Y\}\eta (\widetilde X)\df \hat \tau (\eta) \\
&=\pi_{\hat\tau} (\widetilde Y, \widetilde X),
\end{align*}
where the second equality holds by $\sigma_\mu\in \Sigma$, the third by %
\eqref{e:Prho*} and $E(\mu)\subset P(\mu)$, and the fourth by %
\eqref{e:tauhat}.

\subsection{Proof of Corollary \ref{c:nodisc}} The proof of Corollary \ref{c:nodisc} remains valid if Assumption \ref{a:ord} is replaced with strict single-crossing of $u(y,x)$ in $x$. Since $|X|\geq 3$, strict single-crossing of $u(y,x)$ in $x$ implies that there exist $x_1<x_2<x_3$ in $X$ such that $x_1<\chi(\gamma(\phi))<x_3$. 

Suppose that no disclosure is optimal.%
Then, by part (3) of Theorem~\ref{t:contact}, it follows that \eqref{e:FOC} holds for $\mu=\phi$ and all $x\in X$, so there exist constants $q(\gamma(\phi)),q'(\gamma(\phi))\in \R$ such that
\[
V_y(\gamma(\phi),x)=-q(\gamma(\phi))u_y(\gamma(\phi),x)-q'(\gamma(\phi))u(\gamma(\phi),x)\quad \text{for all $x\in X$}.
\]
Thus, $V_y(\gamma(\phi),\cdot)$ lies in a linear space $L$ spanned by $u_y(\gamma(\phi),\cdot)$ and $u(\gamma(\phi),\cdot)$, whose dimension is at most 2. But a generic  $V_y(\gamma(\phi),\cdot)$ actually lives in a linear space whose dimension is at least 3, since $|X|\geq 3$, and thus it does not belong to $L$, showing that generically no disclosure is suboptimal.

\subsection{Proof of Remark \ref{r:SDD}} Let $\Lambda$ be strictly single-dipped. Note that there do not exist distinct $\mu$ and $\eta$ in  $\Lambda$ such that $\gamma(\mu)=\gamma(\eta)$, as otherwise $\mu/2+\eta/2$, with $|\supp(\mu/2+\eta/2)|\geq 3$, would also be in $\Lambda$, contradicting that $\Lambda$ is pairwise. Thus, there exist unique functions $\chi_1$ and $\chi_2$ from $Y_\Lambda$ to $X$ such that $\supp(\mu)=\{\chi_1(\gamma(\mu)),\chi_1(\gamma(\mu))\}$ and $\chi_1(\gamma(\mu))= \chi(\gamma(\mu))= \chi_2(\gamma(\mu))$ or $\chi_1(\gamma(\mu))< \chi(\gamma(\mu))< \chi_2(\gamma(\mu))$ for all $\mu \in \Lambda$. Moreover, for all $y<y'$ in $Y_\Lambda$, we have $\chi_2(y)\leq \chi_2(y')$, as otherwise there would exist $\mu,\mu'\in \Lambda$ such that $\gamma (\mu)=y$, $\gamma (\mu')=y'$, and $\chi_1(y)\leq \chi (y)<\chi(y')\leq \chi_2(y')<\chi_2(y)$ contradicting that $\Lambda$ is single-dipped. Likewise, for all $y<y'$ in $Y_\Lambda$, we have $\chi_1(y')\notin (\chi_1(y),\chi_2(y))$, as otherwise there would exist $\mu,\mu'\in \Lambda$ such that $\gamma (\mu)=y$, $\gamma (\mu')=y'$, and $\chi_1(y)<\chi_1(y')<\chi_2(y)$ contradicting that $\Lambda$ is single-dipped.

\subsection{Proof of Theorem \ref{l:ssdd}} The proof of Theorem \ref{l:ssdd} remains valid if the condition $u_x (y,x)>0$ in Assumption \ref{a:ord} is replaced with strict single-crossing of $u(y,x)$ in $x$.

We will prove that $\Lambda$ is single-dipped, which implies that every optimal signal is single-dipped. We start with an appropriate version of the theorem of alternative.
\begin{lemma}\label{l:alter}
Exactly one of the following two alternatives holds.
\begin{enumerate}
	\item There exists $\a  >0$ such that $\a  R\leq 0$.
	\item There exists $\b \geq 0$ such that $R\b \geq 0$ and $R\b \neq 0$.
\end{enumerate}
\end{lemma}
\begin{proof}
Clearly, (1) and (2) cannot both hold, because premultiplying $R\b \geq 0$ with $R\b \neq0$ by $\a > 0$ yields $\a R\b >0$, whereas postmultiplying $\a R\leq 0$ by $\b \geq 0$ yields $\a R\b \leq 0.$

Now suppose that (1) does not hold. Then there does not exist $\a \geq 0$ such that
\begin{equation*}
\a 
\begin{pmatrix}
	R &-I 
\end{pmatrix}
\leq
\begin{pmatrix}
	0 &-e
\end{pmatrix}
\end{equation*}
where $I$ is an identity matrix and $e$ is a row vector of ones. Thus, by the theorem of alternative (e.g., Theorem 2.10 in \citealt{gale}), there exists $\b \geq 0$ and $\c \geq 0$ such that
\begin{equation*}
\begin{pmatrix}
	R\\
	-I
\end{pmatrix}
\begin{pmatrix}
	\b  &\c 
\end{pmatrix}
\geq 0\quad \text{and}\quad -e\c <0,
\end{equation*}
which in turn shows that (2) holds.
\end{proof}

We prove the theorem by contraposition. Suppose that $\Lambda$ is not single-dipped, so there exist $\mu_1,\mu_2\in \Lambda$ and $x_1<x_2<x_3$ such that $x_1,x_3\in \supp(\mu_1)$, $x_2\in \supp (\mu_2)$, and $\gamma(\mu_1)<\gamma(\mu_2)$.   By strict single-crossing of $u(y,x)$ in $x$, without loss, we can assume that $x_1< \chi(\gamma(\mu_1))< x_3$, by redefining $x_1=\min \supp (\mu_1)$ and $x_3=\max \supp (\mu_1)$ if necessary.

By \eqref{DO1} and Theorem \ref{t:contact}, we have
\begin{align*}
V(\gamma(\mu_1),x_1)+q(\gamma(\mu_1)) u(\gamma(\mu_1),x_1) &\geq V(\gamma(\mu_2),x_1) +q(\gamma(\mu_2)) u(\gamma(\mu_2),x_1),\\
V(\gamma(\mu_2),x_2)+q(\gamma(\mu_2)) u(\gamma(\mu_2),x_2) &\geq V(\gamma(\mu_1),x_2) +q(\gamma(\mu_1)) u(\gamma(\mu_1),x_2),\\
V(\gamma(\mu_1),x_3)+q(\gamma(\mu_1)) u(\gamma(\mu_1),x_3) &\geq V(\gamma(\mu_2),x_3) +q(\gamma(\mu_2)) u(\gamma(\mu_2),x_3).
\end{align*}
By \eqref{e:q}, we have, for $i=1,2$, 
\[
q(\gamma(\mu_i))=-\frac{\E_{\mu_{i}}[V_y(\gamma(\mu_i),x)]}{\E_{\mu_i} [u_y(\gamma(\mu_i),x)]}>0,
\]
where the inequality follows from Assumptions \ref{a:qc} and \ref{a:ord}. Thus, the vector $\a=(1,q(\gamma(\mu_1)),q(\gamma(\mu_2)))$ is strictly positive and satisfies $\a  R\leq 0$. By Lemma \ref{l:alter}, there does not exist a vector $\b \geq 0$ such that $R\b \geq 0$ and $R\b \neq 0$, as desired.

\subsection{Proof of Theorem \ref{t:SDPD}}\label{a:SDPD}
The proof uses the following five lemmas, whose proofs are deferred to Appendix \ref{Additional Proofs}. We start with the second part of the theorem, and we show that $\Lambda$ is strictly single-dipped (-peaked), which implies that every optimal signal is strictly single-dipped (-peaked).

\begin{lemma}\label{l:S><0}
If $u_{y x}(y,x)/u_x(y,x)$ and $V_{yx}(y,x)/u_x(y,x)$ are increasing (decreasing) in $x$ for all $y$, with at least one of them strictly increasing (decreasing), then $|S|>(<)0$ for all $y$ and $x_1<x_2<x_3$ such that $x_1<\chi(y)<x_3$.
\end{lemma}

\begin{lemma}\label{l:R>0}
If $u_{yx}(y,x)/u_{x}(y,x)$ and $V_{yx} (y_2,x)/u_x(y_1,x)$ are increasing (decreasing) in $x$ for all $y$ and $y_2\geq(\leq) y_1$, with at least one of them strictly increasing (decreasing), then $|R|>(<)0$ for all $x_1<x_2<x_3$ and all $y_2>(<)y_1$ such that $x_1< \chi(y_1)< x_3$.
\end{lemma}

\begin{lemma}\label{l:ySDD}
If $u_{yx}(y,x)/u_{x}(y,x)$ is increasing in $x$ for all $y$, then for all $x_1<x_2<x_3$ and all $y_2>y_1$ such that $x_1< \chi(y_1)< x_3$, we have
\begin{align*}
u(y_2,x_3)u(y_1,x_1)> u(y_2,x_1)u(y_1,x_3),\\
u(y_2,x_2)u(y_1,x_1)> u(y_2,x_1)u(y_1,x_2),\\
u(y_2,x_3)u(y_1,x_2)> u(y_2,x_2)u(y_1,x_3).
\end{align*}
\end{lemma}

\begin{lemma}\label{l:ySPD}
If $V_{yx} (y_2,x)/u_x(y_1,x)$ is decreasing in $x$ for all $y_2\leq y_1$, then for all $x_1<x_2<x_3$ and all $y_2<y_1$ such that $x_1< \chi(y_1)< x_3$, we have
\begin{align*}
\frac{u(y_1,x_1)}{V(y_1,x_1)-V(y_2,x_1)}<\frac{u(y_1,x_2)}{V(y_1,x_2)-V(y_2,x_2)}<\frac{u(y_1,x_3)}{V(y_1,x_3)-V(y_2,x_3)}.
\end{align*}
\end{lemma}

\begin{lemma}\label{l:stab}
Suppose that $V^n$ is a sequence of functions satisfying Assumption \ref{a:smooth} such that $V_y^n$ converges uniformly to $V_y$, and suppose that the contact sets $\Lambda^n$ under $V^n$ are single-dipped (-peaked). Then there exists a single-dipped (-peaked) optimal signal under $V$.
\end{lemma}

Now, the set $\Lambda$ is single-dipped (-peaked) by Theorem \ref{l:ssdd} with 
\[
\b =
\begin{pmatrix}
u(y_2,x_3)u(y_1,x_2)-u(y_2,x_2)u(y_1,x_3)\\
u(y_2,x_3)u(y_1,x_1)-u(y_2,x_1)u(y_1,x_3)\\
u(y_2,x_2)u(y_1,x_1)-u(y_2,x_1)u(y_1,x_2)
\end{pmatrix}
\quad
\left( \b =
\begin{pmatrix}
\frac{u(y_2,x_1)}{V(y_2,x_1)-V(y_1,x_1)} \\
\frac{u(y_2,x_2)}{V(y_2,x_2)-V(y_1,x_2)} \\
\frac{u(y_2,x_3)}{V(y_2,x_3)-V(y_1,x_3)}
\end{pmatrix}\right),
\]
as follows from Lemma \ref{l:R>0} and Lemma \ref{l:ySDD} (Lemma \ref{l:ySPD}). Moreover, $|\supp (\mu)|\leq 2$ for all $\mu\in \Lambda$ by Theorem \ref{t:pairwise} and Lemma \ref{l:S><0}, showing that $\Lambda$ is strictly single-dipped (-peaked).

Finally, we prove the first part of the theorem. Consider $V^n$ such that
\[
V_y^n(y,x)=V_y(y,x)+\int_0^x \frac{\tilde v(x)}{n}u_x (y,\tilde x) \df \tilde x,
\]
where $\tilde v (x)$ is a continuous, strictly positive, and strictly increasing (decreasing) function on $[0,1]$. Then $V_y^n(y,x)>0$ because $V_y(y,x)>0$ and $u_x(y,x)>0$ for all $(y,x)$, by Assumption \ref{a:ord}. Moreover, for all $y_2\geq (\leq)y_1$,
\[
\frac{V_{yx}^n(y_2,x)}{u_x(y_1,x)}=\frac{V_{yx}(y_2,x)}{u_x(y_1,x)} + \frac{\tilde v(x)}{n} \frac{u_x(y_2,x)}{u_x(y_1,x)}
\]
is strictly increasing (decreasing) in $x$, because  $\tilde v (x)$ is strictly positive and strictly increasing (decreasing) in $x$; ${V_{yx}(y_2,x)}/{u_x(y_1,x)}$ is increasing (decreasing) in $x$; and ${u_x(y_2,x)}/{u_x(y_1,x)}$ is increasing in $x$, since $u_{yx}(y,x)/u_{x}(y,x)$ is increasing (decreasing) in $x$. Thus, by Lemma \ref{l:stab}, there exists an optimal single-dipped (-peaked) signal.

\subsection{Proof of Theorem \ref{t:brenier}}
The proof of Theorem \ref{t:brenier} remains valid if Assumption \ref{a:ord} is replaced with strict single-crossing of $u(y,x)$ in $x$.

We start with the following lemma, which is also used in the proof of Theorem \ref{t:NAD}.
\begin{lemma}\label{l:nad}
If $X=[0,1]$ and $\Lambda$ is strictly single-dipped, then for each $y$ in $Y_\Lambda$ there exists $y'\leq y$ in $Y_\Lambda$ such that $\chi_1(y)\leq \chi_1(y')=\chi_2(y')\leq \chi_2(y)$.
\end{lemma}
\begin{proof}
We prove that a required $y'$ can be constructed as
\begin{gather*}
y'=\inf \{\tilde y\in Y_\Lambda:\chi_1(y)\leq \chi_1(\tilde y)\leq \chi_2(\tilde y)\leq \chi_2(y)\}.
\end{gather*}
By definition, $y'\leq y$. Moreover, $y'\in Y_\Lambda$, because $Y_\Lambda$ is compact. Suppose by contradiction that $\chi_1(y')<\chi_2(y')$. Let $X^\star = \cup_{\mu\in \Lambda} \supp (\mu)$. Since there exists an optimal signal $\tau$, which satisfies $\supp (\tau)\subset \Lambda$ and $\int \mu \df \tau =\phi$, we have $\phi (X^\star)=1$, so the closure of $X^\star $ is $X=[0,1]$. Thus, there exists $y''\in Y_\Lambda$ such that $\chi(y'')$ or $\chi_2(y'')$ is in  $(\chi_1(y'),\chi_2(y'))$. Since $\Lambda$ is strictly single-dipped, we have $y''<y'$ and $\chi_1(y')\leq \chi_1(y'')\leq \chi_2(y'')\leq \chi_2(y')$, contradicting the definition of $y'$.
\end{proof}
Next, we claim that if $y\in Y_\Lambda$ and $\varepsilon>0$ are such that $\chi_1(\tilde y)<\chi_2(\tilde y)$ for all $\tilde y\in (y-\varepsilon,y)\cap Y_\Lambda$, then $\chi_1(\tilde y_1)<\chi_2(\tilde y_2)$ for all $\tilde y_1,\tilde y_2\in (y-\varepsilon,y)\cap Y_\Lambda$. Suppose by contradiction that there exist $\tilde y_1,\tilde y_2\in (y-\varepsilon,y)\cap Y_\Lambda$ such that $\chi_1(\tilde y_1)\geq \chi_2(\tilde y_2)$. By Lemma \ref{l:nad}, there exists $\tilde y_1'\leq \tilde y_1$ in $Y_\Lambda$ such that
 $\chi_2(\tilde y_1')=\chi(\tilde y_1')=\chi_1(\tilde y_1')\geq \chi_1(\tilde y_1)\geq \chi_2(\tilde y_2)\geq \chi (\tilde y_2)$, so $\tilde y_1'\in (y-\varepsilon,y)\cap Y_\Lambda$ and $\chi_1(\tilde y_1')=\chi_2(\tilde y_1')$, yielding a contradiction.

Suppose now that $\phi $ has a density. Suppose by contradiction that there are two distinct optimal signals, $\tau$ and $\tau'$. Since $\Lambda$ is strictly single-dipped, for each $y\in Y_\Lambda$, there is a unique $\mu$ in $\Lambda$ such that $\gamma(\mu)=y$, namely $\mu=\rho_y \delta_{\chi_1(y)}+(1-\rho_y)\delta_{\chi_2(y)}$ where
\[
\rho_y=
\begin{cases}
\frac{u(y,\chi_2(y))}{u(y,\chi_2(y))-u(y,\chi_1(y))}, &\chi_1(y)<\chi_2(y),\\
0, &\chi_1(y)=\chi_2(y).
\end{cases}
\]
Thus, distinct signals $\tau$ and $\tau'$ must induce distinct distributions $\alpha$ and $\alpha'$ over actions $\gamma(\mu)$.  Let $\hat y=\sup \{y\in Y:\alpha([0,y])\neq\alpha'([0,y])\}\in Y_\Lambda$, where the inclusion follows from $\alpha\neq \alpha'$ and $\alpha(Y_\Lambda)=\alpha'(Y_\Lambda)=1$. By the regularity condition and the claim above, there exists $\varepsilon>0$ such that either (i) $\chi_1(\tilde y) =\chi_2(\tilde y)$ for all $\tilde y\in (\hat y-\varepsilon,\hat y)\cap Y_\Lambda$ or (ii) $\chi_1(\tilde y_1)< \chi_2(\tilde y_2)$ for all $\tilde y_1, \tilde y_2\in (\hat y-\varepsilon,\hat y)\cap Y_\Lambda$. We will now show that $\alpha([0,\tilde y])=\alpha'([0,\tilde y])$ for all $\tilde y\in (\hat y-\varepsilon,\hat y)$ contradicting the definition of $\hat y$. Since $\chi_2$ is increasing, states $x>\chi_2(\tilde y)$ can only induce actions $y>\tilde y$.  Thus, since $\gamma$ is bijective from $\Lambda$ to $Y_\Lambda$ and since $\alpha([0,y])=\alpha'([0,y])$ for all $y\geq \hat y$, in both cases (i) and (ii), we have, for all $\tilde y\in (\hat y-\varepsilon,\hat y)\cap Y_\Lambda$, 
\begin{align*}
\phi((\chi_2(\tilde y),1])-\phi([\chi_2(\tilde y),1])&\leq  \int_{[\tilde y,\hat y]} (1-\rho_y)\df \alpha(y) -\int_{[\tilde y,\hat y]} (1-\rho_y)\df \alpha'(y)\\ &\leq \phi([\chi_2(\tilde y),1])-\phi((\chi_2(\tilde y),1]).
\end{align*}
Moreover, since $\phi$ has a density, we have $\phi((\chi_2(\tilde y),1])=\phi([\chi_2(\tilde y),1])$, and hence
\[
\int_{[\tilde y,\hat y]} (1-\rho_y)\df \alpha(y) =\int_{[\tilde y,\hat y]} (1-\rho_y)\df \alpha'(y).
\]
Then, since $1-\rho_y>0$ for all $y\in Y_\Lambda$, and since $\supp (\alpha)\subset Y_\Lambda$ and $\supp (\alpha')\subset Y_\Lambda$, it follows that $\alpha([\tilde y,\hat y])=\alpha'([\tilde y,\hat y])$ for all $\tilde y\in (\hat y-\varepsilon,\hat y)$. Thus, since  $\alpha([0,y])=\alpha'([0,y])$ for all $y\geq \hat y$, it follows that $\alpha([0,\tilde y])=\alpha'([0,\tilde y])$ for all $\tilde y\in (\hat y-\varepsilon,\hat y)$.

\subsection{Proof of Remark \ref{t:FDu}} The proof of Remark \ref{t:FDu} remains valid if Assumption \ref{a:ord} is replaced with strict single-crossing of $u$ in $x$. 

Suppose by contradiction that $\Lambda$ contains $\mu=\rho \delta_{x_1}+(1-\rho)\delta_{x_2}$, with $x_1<x_2$ and $\rho \in (0,1)$. Denote $y=\gamma (\mu)$ and $x=\chi (y)$. By strict single-crossing of $u$ in $x$, we have $x_1<x<x_2$. Since $X=[0,1]$ and full disclosure is optimal, we have $\delta_{x}\in \Lambda$. Thus, 
\begin{gather*}
\rho p(x_1) +(1-\rho) p(x_2)=\rho V(y,x_1) +(1-\rho)V(y,x_2)\quad \text{and}\quad p(x) = V(y,x).
\end{gather*}
Adding the two equalities gives
\[
\tfrac \rho 2 p(x_1)+\tfrac 12p(x) +\tfrac {1-\rho}2  p(x_2)=\tfrac 12 \rho V(y,x_1) + \tfrac 1 2 V(y,x) +\tfrac 12(1-\rho)V(y,x_2),
\]
which shows that $\eta=\rho\delta_{x_1} /2 +\delta_x/2+(1-\rho)\delta_{x_2}/2$, contradicting that $\Lambda$ is pairwise.

\subsection{Proof of Theorem \ref{p:full}} The proof of Lemma \ref{l:dual} remains valid without Assumption \ref{a:ord} and when $X$ is an arbitrary compact metric space. The support of the full disclosure signal is the set of all degenerate posteriors on $X$. Thus, by Lemmas \ref{l:dual} and \ref{l:Deq}, full disclosure is optimal iff there exists $q\in B(Y)$ such that
\begin{gather*}
V(\gamma(\delta_x),x)\geq V(y,x) +q(y) u(y,x), \quad \text{for all $(y,x)\in Y\times X$},\\
\iff \frac{V(y,x_1)-V(\gamma(\delta_{x_1}),x_1)}{-u(y,x_1)}\leq q(y)\leq \frac{V(\gamma(\delta_{x_2}),x_2)-V(y,x_2)}{u(y,x_2)},
\end{gather*}
for all $y\in Y$ and $x_1,x_2\in X$ such that $u(y,x_1)<0<u(y,x_2)$.
As shown in the proof of Lemma \ref{l:Deq}, the left-hand side and right-hand side functions are bounded on $Y\times X$, so full disclosure is optimal iff, for all $y\in Y$ and $x_1,x_2\in X$ such that $u(y,x_1)<0<u(y,x_2)$, we have
\begin{gather*}
\frac{V(y,x_1)-V(\gamma(\delta_{x_1}),x_1)}{-u(y,x_1)}\leq \frac{V(\gamma(\delta_{x_2}),x_2)-V(y,x_2)}{u(y,x_2)},\\
\iff u(y,x_2)V(y,x_1)-u(y,x_1)V(y,x_2)\leq u(y,x_2)V(\gamma(\delta_{x_1}),x_1)-u(y,x_1)V(\gamma(\delta_{x_2}),x_2),\\
\iff \rho V(\gamma(\mu),x_1) +(1-\rho )V(\gamma(\mu),x_2)\leq \rho V(\gamma(\delta_{x_1})),x_1) +(1-\rho )V(\gamma(\delta_{x_2}),x_2),
\end{gather*}
where $\rho =u(y,x_2)/(u(y,x_2)-u(y,x_1))$, $\mu=\rho \delta_{x_1}+(1-\rho )\delta_{x_2}$, and $\gamma(\mu)=y$, by the definition of $\gamma (\mu)$. To complete the proof that full disclosure is optimal iff \eqref{e:full} holds for all $\mu$, note that for each $y$ and $x_1,x_2\in X$ such that $u(y,x_1)<0<u(y,x_2)$, we have $\rho =u(y,x_2)/(u(y,x_2)-u(y,x_1))\in (0,1)$; and conversely, for each $x_1<x_2$ and $\rho \in (0,1)$, there exists a unique $y\in (\gamma(\delta_{x_1}),\gamma(\delta_{x_2}))$ such that $\rho =u(y,x_2)/(u(y,x_2)-u(y,x_1))$.

Finally, assume that \eqref{e:full} holds with strict inequality for all $\mu$. Suppose by contradiction that full disclosure is not uniquely optimal. Then, by Lemmas \ref{l:dual} and \ref{l:Deq}, there exist $\eta\in \Lambda$ and distinct $x_1,x_2\in \supp (\eta)$. By the definition of $\gamma(\eta)$, without loss, we can assume that either $u(\gamma(\eta),x_1)= 0= u(\gamma(\eta),x_2)$ or $u(\gamma(\eta),x_1)< 0< u(\gamma(\eta),x_2)$. In the case $u(\gamma(\eta),x_1)= 0= u(\gamma(\eta),x_2)$, we have  $\gamma(\mu )=\gamma(\delta_{x_1})=\gamma(\delta_{x_1})$ for $\mu=\delta_{x_1}/2+\delta_{x_2}/2$, so
\[
\tfrac 12 V(\gamma(\mu ),x_1) +\tfrac 12 V(\gamma(\mu ),x_2)= \tfrac 12 V(\gamma(\delta_{x_1}),x_1) +\tfrac 12 V(\gamma(\delta_{x_2}),x_2),
\]
contradicting that \eqref{e:full} holds with strict inequality for $\mu $. In the case $u(\gamma(\eta),x_1)< 0< u(\gamma(\eta),x_2)$, we have $\gamma(\mu )=\gamma (\eta)$ for $\mu=\rho \delta_{x_1}+(1-\rho )\delta_{x_2}$ with $\rho =u(\gamma(\eta),x_2)/(u(\gamma(\eta),x_2)-u(\gamma(\eta),x_1)\in (0,1)$. Since $\eta\in \Lambda$ and $x_1,x_2\in \supp (\eta)$, we have
\begin{align*}
V(\gamma(\eta),x_1)+q(\gamma(\eta))u(\gamma(\eta),x_1)\geq V(\gamma(\delta_{x_1}),x_1),\\
V(\gamma(\eta),x_2)+q(\gamma(\eta))u(\gamma(\eta),x_2)\geq V(\gamma(\delta_{x_2}),x_1).	
\end{align*}
Adding the first inequality multiplied by $\rho$ and the second inequality multiplied by $1-\rho$ gives
\[
\rho V(\gamma(\mu ),x_1) +(1-\rho )V(\gamma(\mu ),x_2)\geq \rho V(\gamma(\delta_{x_1}),x_1) +(1-\rho )V(\gamma(\delta_{x_2}),x_2),
\]
contradicting that \eqref{e:full} holds with strict inequality for $\mu $.

\subsection{Proof of Corollary \ref{c:fd}'}
Condition \eqref{e:full} holds because 
\begin{align*}
 \rho & V(\rho x_1 +(1-\rho )x_2,x_1 )+(1-\rho )V(\rho x_1 +(1-\rho )x_2,x_2) \\
\leq & \rho (\rho V(x_1 ,x_1 )+(1-\rho )V(x_2,x_1))+(1-\rho )(\rho V(x_1 ,x_2)+(1-\rho )V(x_2,x_2)) \\
 \leq &\rho V(x_1 ,x_1 )+(1-\rho )V(x_2,x_2),
\end{align*}%
where the first inequality holds because $V(y ,x )$ is convex in $y$, and the second holds because $V(x_1 ,x_2)+V(x_2,x_1 )\leq V(x_1 ,x_1 )+V(x_2,x_2)$.

\subsection{Proof of Theorem \ref{t:NAD}}
The proof of Theorem \ref{t:NAD} remains valid if Assumption \ref{a:ord} is replaced with strict single-crossing of $u(y,x)$ in $x$. Since $X=[0,1]$, $\Lambda$ is strictly single-dipped, and for all $x_1<x_2$ there exists $\rho\in (0,1)$ such that \eqref{e:nd} holds, it follows that $ \chi_1(y_2)\leq \chi_1(y_1)$ for all $y_1<y_2$ in $Y_\Lambda$, and thus $\Lambda$ is negative assortative. Suppose by contradiction that there exist $y_1<y_2$ in $Y_\Lambda$ such that $\chi_1(y_1)<\chi_1(y_2)$. Then $\chi_2(y_1)\leq \chi_1(y_2)$, as otherwise there would exist $\mu_1,\mu_2\in \Lambda$ such that $\gamma(\mu_1)=y_1<y_2=\gamma(\mu_2)$, and $\chi_1(y_1)<\chi_1(y_2)<\chi_2(y_1)$ contradicting that $\Lambda$ is single-dipped. By Lemma \ref{l:nad}, there exist $y_1'\leq y_1$ and $y_2'\leq y_2$ in $Y_\Lambda$ such that $\chi_1(y_1)\leq \chi_1(y_1')=\chi_2(y_1')\leq \chi(y_1)\leq  \chi_2(y_1)\leq\chi_1(y_2)\leq \chi_1(y_2')=\chi_2(y_2')\leq \chi_2(y_2)$, and thus $y_1'\leq y_2'$. In fact, $y_1'<y_2'$, as otherwise we would have $\chi_1(y_1)\leq \chi(y_1)= \chi_2(y_1)=\chi_1(y_2)$, which implies $\chi_1(y_1)=\chi_2(y_1)=\chi_1(y_2)$, contradicting $\chi_1(y_1)<\chi_1(y_2)$. Thus, denoting $x_1=\chi(y_1')<\chi (y_2')=x_2$, we have $\delta_{x_1},\delta_{x_2}\in \Lambda$. For any $\mu=\rho  \delta_{x_1}+(1-\rho )\delta_{x_2}$ with $\rho \in (0,1)$, we have 
\begin{align*}
	p(x_1)&=V(\gamma(x_1),x_1)\geq V(\gamma(\mu),x_1)+q(\gamma(\mu))u(\gamma(\mu),x_1),\\
	p(x_2)&=V(\gamma(x_2),x_2)\geq V(\gamma(\mu),x_2)+q(\gamma(\mu))u(\gamma(\mu),x_2),
\end{align*}
by \eqref{DS1} and the definition of $\Lambda$. Adding the first inequality multiplied by $\rho $ and the second inequality multiplied by $1-\rho $, we obtain that \eqref{e:nd} fails for all $\rho \in (0,1)$, yielding a contradiction.

Now assuming that $\phi$ has a density $f$ and $\Lambda$ is negative assortative, we will show that the functions $\chi_1$ and $\chi_2$ are continuous and satisfy the differential equations \eqref{e:obed}--\eqref{e:q'} and the boundary condition \eqref{e:boundary}. Since the closure of $X^\star=\cup_{\mu\in \Lambda}\supp (\mu)$ is $X=[0,1]$, it follows that the closure of the union of the images of the functions $\chi_1$ and $\chi_2$ must also be equal to $[0,1]$. 
Since $\chi_1$ is decreasing and $\chi_2$ is increasing on the compact domain $Y_\Lambda$, and since $\chi_1(y)\leq \chi (y)\leq \chi_2(y)$ for all $y\in Y_\Lambda$, it follows that $\chi_1$ and $\chi_2$ are continuous functions such that $\chi_1(\ul y)=\chi (\ul y)=\chi_2(\ul y)$, $\chi_1(y)<\chi (y)<\chi_2(y)$ for all $y>\ul y$ in $Y_\Lambda$, $\chi_1(\ol y)=0$, $\chi_2(\ol y)=1$, and $(\chi_1(\ul z_i),\chi_2(\ul z_i)) =(\chi_1(\ol z_i), \chi_2(\ol z_i))$ for all $i$, where $\{(\ul z_i,\ol z_i)\}_i$ is an at most countable set of disjoint open intervals comprising the set $[\ul y,\ol y]\setminus Y_\Lambda$. Since $\phi $ has a density, the measure of the endpoints of these intervals is zero, and hence the set of optimal signals is unaffected if we  extend the domain of $\chi_1$ and $\chi_2$ to $[\ul y,\ol y]$ by setting $\chi_1(y)=\chi_1(\ul z_i)=\chi_1(\ol z_i)$ and $\chi_2(y)=\chi_2(\ul z_i)=\chi_2(\ol z_i)$ for all $y\in (\ul z_i,\ol z_i)$. In sum, without loss of generality, we can assume that $\chi_1$ and $\chi_2$ are continuous monotone functions defined on $[\ul y,\ol y]$ that satisfy \eqref{e:boundary} and $\chi_1(y)<\chi (y)<\chi_2(y)$ for all $y\in (\ul y,\ol y]$.

Since $\chi_1$ is continuously decreasing on $[\ul y, \ol y]$, $\chi_2$ is continuously increasing on $[\ul y, \ol y]$, and $\phi$ has a density, we can rewrite \eqref{PO2} for $\widetilde Y=[y,y']$, with $\ul y\leq y< y'\leq \ol y$, as
\[
\int_{y}^{y'} u(\tilde y,\chi_1(\tilde y))(-\df \phi ([0,\chi_1(\tilde y)]))+\int _{y}^{y'} u(\tilde y,\chi_2(\tilde y))\df \phi ([0,\chi_2(\tilde y)])=0.
\]
Taking the limit $y'\downarrow y$, we obtain  \eqref{e:obed} for all $y\in [\ul y,\ol y]$.

By Theorem \ref{t:contact}, for all $y> \ul y$ in $Y_\Lambda$,
\begin{gather*}
V_y(y,\chi_1(y))+ q(y)u_y(y,\chi_1(y))+q'(y)u(y,\chi_1(y))=0,\\
V_y(y,\chi_2(y))+ q(y)u_y(y,\chi_2(y))+q'(y)u(y,\chi_2(y))=0.
\end{gather*}
Solving for $q(y)$ and $q'(y)$, we get, for all $y> \ul y$ in $Y_\Lambda$,
\begin{gather*}
q(y)= \frac{V_y(y,\chi_1(y))u(y,\chi_2(y))-V_y(y,\chi_2(y))u(y,\chi_1(y))}{u(y,\chi_1(y))u_y(y,\chi_2(y))-u(y,\chi_2(y))u_y(y,\chi_1(y))},\\
q'(y)=\frac{V_y(y,\chi_1(y))u_y(y,\chi_2(y))-V_y(y,\chi_2(y))u_y(y,\chi_1(y))}{u_y(y,\chi_1(y))u(y,\chi_2(y))-u_y(y,\chi_2(y))u(y,\chi_1(y))},
\end{gather*}
where the denominators in the expressions for $q(y)$ and $q'(y)$ are not equal to $0$, by Assumption \ref{a:qc}.
Noting that $q'$ is the derivative of $q$ gives \eqref{e:q'} for all $y> \ul y$ in $Y_\Lambda$.

\subsection{Proof of Corollary \ref{c:ndSDD}}
Noting that $\rho u(\gamma (\mu),x_1)+(1-\rho )u(\gamma(\mu),x_2)=0$ and denoting $y=\gamma (\mu)$, we infer that \eqref{e:nd} fails if there exist $x_1<x_2$ such that for all $y\in(\gamma (\delta_{x_1}),\gamma (\delta_{x_2}))$, we have
\[
u(y,x_2)(V(y,x_1)-V(\gamma(\delta_{x_1}),x_1))-u(y,x_1)(V(y,x_2)-V(\gamma(\delta_{x_2}),x_2))\leq 0.
\]
By Taylor's theorem and some algebra, we get
\begin{gather*}
u(y,x_2)(V(y,x_1)-V(\gamma(\delta_{x_1}),x_1))-u(y,x_1)(V(y,x_2)-V(\gamma(\delta_{x_2}),x_2))\\
=\frac 12u_y(y,\chi(y))\left( V_{yy}(y,\chi (y)) - \frac{V_y(y,\chi (y)) u_{yy}(y,\chi (y))}{u_y(y,\chi (y))} \right.\\
\left.-2\frac{V_{yx}(y,\chi (y)) u_y(y,\chi (y)) - V_y(y,\chi (y)) u_{yx}(y,\chi (y))}{u_x(y,\chi (y))}\right)\\
\cdot (y-\gamma(\delta_{x_1}))(\gamma(\delta_{x_2})-y)(\gamma(\delta_{x_2})-\gamma(\delta_{x_1})) \\
+o((y-\gamma(\delta_{x_1}))(\gamma(\delta_{x_2})-y)(\gamma(\delta_{x_2})-\gamma(\delta_{x_1}))).
\end{gather*}
Hence, if \eqref{e:ndSDD} fails at some $y$, then there exist $x_2>x_1$ with $\gamma(\delta_{x_2})-y>0$ and $y-\gamma(\delta_{x_1})>0$ small enough such that \eqref{e:nd} fails for all $\rho \in (0,1)$.

Note that ${\df\chi (y)}/{\df y}=-{u_y(y,\chi(y))}/{u_x(y,\chi(y))}$, by the implicit function theorem applied to $u(y,\chi (y))=0$.
Thus, the derivative of $q(y)=-{V_y(y,\chi(y))}/{u_y(y,\chi(y))}$ is given by
\[
q'(y)=-\frac{V_{yy}(y,\chi(y))}{u_y(y,\chi(y))}+\frac{V_{yx}(y,\chi(y))}{u_x(y,\chi(y))}+\frac{V_y(y,\chi(y)) u_{yy}(y,\chi(y))}{(u_y(y,\chi(y)))^2}-\frac{V_y(y,\chi(y)) u_{yx}(y,\chi(y))}{u_y(y,\chi(y))u_x(y,\chi(y))}.
\]

Conversely, suppose that \eqref{e:ndSDD}, together with all other assumptions of the corollary, holds. Then, for $y>\gamma (\delta_x)$, we have
\begin{align*}
	&V(y,x)-\frac{V_y(y,\chi(y))}{u_y(y,\chi(y))}u(y,x)-V(\gamma(\delta_x),x)\\
	=&(V(\tilde y,x) +q(\tilde y)u(\tilde y,x))|^{y}_{\gamma(\delta_x)}\\
	=&\int^{y}_{\gamma(\delta_x)} [V_y(\tilde y,x)+q(\tilde y)u_y(\tilde y,x)+q'(\tilde y)u(\tilde y,x)]\df \tilde y\\
	\geq &\int^{y}_{\gamma(\delta_x)} \left [V_y(\tilde y,x)-\frac{V_y(\tilde y,\chi(\tilde y))}{u_y(\tilde y,\chi(\tilde y))} u_y(\tilde y,x)\right ]\df \tilde y\\
	&+\int^{y}_{\gamma(\delta_x)} \left [\frac{V_y(\tilde y,\chi (\tilde y)) u_{yx}(\tilde y,\chi (\tilde y))}{u_y(\tilde y,\chi (\tilde y))u_x(\tilde y,\chi (\tilde y))} - \frac{V_{yx} (\tilde y,\chi (\tilde y))}{u_x (\tilde y,\chi (\tilde y))}\right ]u(\tilde y,x)\df \tilde y\\
	=&\int^{y}_{\gamma(\delta_x)}\int ^{\chi (\tilde y)}_{x} \left [\frac{V_y(\tilde y,\chi(\tilde y))}{u_y(\tilde y,\chi(\tilde y))} u_{yx}(\tilde y,\tilde x)-V_{yx} (\tilde y,\tilde x)\right ]\df \tilde x \df \tilde y\\
	&+\int^{y}_{\gamma(\delta_x)}\int ^{\chi (\tilde y)}_{x} \left [ \frac{V_{yx} (\tilde y,\chi (\tilde y))}{u_x (\tilde y,\chi (\tilde y))} - \frac{V_y(\tilde y,\chi (\tilde y)) u_{yx}(\tilde y,\chi (\tilde y))}{u_x(\tilde y,\chi (\tilde y))} \right ]u_x(\tilde y,\tilde x)\df \tilde x \df \tilde y\\
	= &\int^{y}_{\gamma(\delta_x)}\int ^{\chi (\tilde y)}_{x} \left [\frac{V_{yx} (\tilde y,\chi (\tilde y))}{u_x (\tilde y,\chi (\tilde y))} -\frac{V_{yx} (\tilde y,\tilde x)}{u_x(\tilde y,\tilde x)}\right ]u_x(\tilde y,\tilde x)\df \tilde x \df \tilde y\\
	&+\int^{y}_{\gamma(\delta_x)}\int ^{\chi (\tilde y)}_{x} \frac{V_y(\tilde y,\chi(\tilde y))}{-u_y(\tilde y,\chi(\tilde y))} \left [\frac{u_{yx}(\tilde y,\chi(\tilde y))}{u_x (\tilde y,\chi(\tilde y))}-\frac{u_{yx}(\tilde y,\tilde x)}{u_x (\tilde y,\tilde x)} \right]u_x(\tilde y,\tilde x)\df \tilde x \df \tilde y>0,
\end{align*}
where the first and last equalities are by rearrangement, the second and third equalities are by the fundamental theorem of calculus, the first inequality is by \eqref{e:ndSDD} and substitution of $q(\tilde y)$ and $q'(\tilde y)$, and the last inequality is by our assumptions imposed in the corollary.

By Taylor's theorem, we have, for $x_1<x_2$ and $y\in(\gamma(\delta_{x_1}),\gamma(\delta_{x_2}))$,
\begin{gather*}
u(y,x_2)(V(y,x_1)-V(\gamma(\delta_{x_1}),x_1))-u(y,x_1)(V(y,x_2)-V(\gamma(\delta_{x_2}),x_2))\\
=\left[V(\gamma(\delta_{x_2}),x_1)-\frac{V_y(\gamma(\delta_{x_2}),x_2)}{u_y(\gamma(\delta_{x_2}),x_2)}u(\gamma(x_2),x_1)-V(\gamma(\delta_{x_1}),x_1)\right]\\
\cdot(-u_y(\gamma (\delta_{x_2}),x_2))(\gamma(\delta_{x_2})-y)+o(\gamma(\delta_{x_2})-y).
\end{gather*}
Hence \eqref{e:nd} holds for sufficiently small $\rho >0$.

\newpage

\begin{center}
\Large{Online Appendix}
\end{center}

\renewcommand{\thesection}{F}

\section{Additional Proofs} \label{Additional Proofs}

\subsection{Proof of Lemma \ref{l:ASC}} \label{proof:ASC}

$(1)\implies (2).$ It is easy to see that Assumption \ref{a:qc} for $\mu=\delta_x$ such that $u(y,x)=0$ yields \eqref{1}. Similarly, Assumption \ref{a:qc} for $\mu=\rho \delta_x+(1-\rho )\delta_x$ such that $u(y,x)<0<u(y,x')$ and $\rho u(y,x)+(1-\rho )u(y,x')=0$  yields \eqref{2}.

$(2)\implies (1).$ By Lemma \ref{l:Choquet}, for any $y\in Y$ and $\mu\in \Delta(X)$ such that $\int u(y,x)\df\mu=0$, there exists $\sigma_\mu \in \Delta (\Delta(X))$ such that $\int \eta \df\sigma_\mu =\mu$, and for each $\eta\in \supp (\sigma_\mu)$ there exist $x,x'\in X$ and $\rho \in [0,1]$ such that $\eta = \rho \delta_x +(1-\rho )\delta_{x'}$ and
\begin{equation}\label{Eu=0}
\rho u(y,x)+(1-\rho )u(y,x')=0.
\end{equation}
It suffices to show that
\begin{equation}\label{3}
\rho u_y(y,x)+(1-\rho )u_y(y,x')<0.
\end{equation}
There are two cases to consider. First, if $\rho u(y,x)=0$, then \eqref{3} follows from \eqref{1} and \eqref{Eu=0}. Second, if $\rho u(y,x)\neq 0$, then \eqref{3} follows from \eqref{2} and \eqref{Eu=0}.

$(3)\implies (1).$ 
Notice that
\begin{equation*}
\int u(y,x)\df \mu =0 \iff \int \tilde u(y,x)\df \mu=0.	
\end{equation*}
Hence, if $\tilde u_y(y,x)<0$ for all $(y,x)$ and $\int u(y,x)\df \mu=0$, then 
\begin{equation*}
\int u_y(y,x)\df \mu = g(y)\int \tilde u_y(y,x)\df \mu+  g'(y) \int \tilde u(y,x)\df \mu=g(y)\int \tilde u_y(y,x)\df \mu<0,
\end{equation*}
yielding Assumption \ref{a:qc}.

$(1)\implies (3).$ We rely on the following lemma.
\begin{lemma}\label{l:ASC*}
If Assumptions \ref{a:smooth} and \ref{a:qc} hold, then there exists a continuous function $h(y)$ such that
\begin{equation}\label{e:ASC*}
u_y(y,x)+h (y) u(y,x) <0, \quad \text{for all } (y,x)\in Y\times X.
\end{equation}
\end{lemma}
Given this lemma, the required $g$ is given by
\[
g(y)=e^{-\int_0^y h (\tilde y)\df \tilde y},
\]
as follows from
\[
\tilde u_y(y,x)= \frac{\partial }{\partial y}\left(\frac{u(y,x)}{e^{-\int_0^y h (\tilde y)\df \tilde y}}\right) =\frac{u_y(y,x)+h(y)u(y,x)}{e^{-\int_0^y h (\tilde y)\df \tilde y}}<0.
\]

\begin{proof}[Proof of Lemma \ref{l:ASC*}] Fix $y\in Y$. Let $M_+(X)$ be the set of positive Borel measures on $X$. Define the set $C\subset \R^3$ as follows 
\[C=\left\{\left(\int u(y,x)\df \mu,\int u_y(y,x)\df \mu -z,\int \df \mu\right)\ \big|\ \mu\in M_+(X), \ z\geq 0 \right\}.\]
Clearly, $C$ is a convex cone. 

Moreover, $C$ is closed, because $u(y,x)$ and $u_y(y,x)$ are continuous in $x$. To see this, let sequences $\mu_n\in M_{+}(X)$ and $z_n\in \R_+^n$ be such that
\[
\int u(y,x) \df \mu_n\rightarrow c_1,\ \int u_y(y,x)\df \mu_n -z_n\rightarrow c_2,\ \int \df \mu_n\rightarrow c_3
\] 
for some $(c_1,c_2,c_3)\in \R^3$. It follows from $\int \df \mu_n\rightarrow c_3$ that all $\mu_n$ belong to a compact subset of positive measures whose total variation is bounded by $\sup_n \int \df \mu_n$, and hence, up to extraction of a subsequence, $\mu_n\rightarrow \mu\in M_+(X)$, with $\int \df \mu =c_3$. 
Since $u(y,x)$ and $u_y(y,x)$ are continuous in $x$, we get $\int u(y,x)\df \mu_n\rightarrow \int u(y,x)\df \mu=c_1$ and $\int u_y(y,x)\df \mu_n\rightarrow \int u_y(y,x)\df \mu$. Hence, $z_n\rightarrow \int u_y(y,x)\df \mu-c_2=z\geq 0$. In sum,
\[
\int u(y,x) \df \mu=c_1,\ \int u_y(y,x)\df \mu -z= c_2,\ \int \df \mu= c_3,
\]
showing that $C$ is closed.

Next, notice that Assumption \ref{a:qc} implies that $(0,0,1)\notin C.$ Thus, by the separation theorem (e.g., Corollary 5.84 in \citealt{aliprantis2006}), there exists $\beta \in \R^3$ such that, for all $\mu\in M_+(X)$ and $z\geq 0$,
\begin{align*}
0\b_1+0\b_2+1\b_3<0&\leq \left(\int u(y,x)\df \mu\right)\b_1+\left(\int u_y(y,x)\df \mu -z\right)\b_2+\left(\int \df \mu\right)\b_3,
\end{align*}
or equivalently
\begin{equation}\label{Cy}
\begin{aligned}
u(y,x) \b_1 + u_y(y,x) \b_2 +\b_3 &\geq 0, \quad \text{for all}\ x\in X,\\
-\b_2 &\geq 0,\\
\b_3 &<0.
\end{aligned}
\end{equation}
We now show that there exists a scalar $h(y) \in \R$ satisfying
\begin{equation}\label{e:gamma}
u_y(y,x)+h(y) u(y,x)<0, \quad \text{for all }  x\in X.	
\end{equation}
There are two cases. First, if $\b_2<0$ then $h(y) = \b_1/\b_2\in \R$ satisfies \eqref{e:gamma}. Second, if $\b_2=0$ then \eqref{Cy} implies that
\[
u(y,x)\b_1\geq -\b_3>0, \quad \text{for all } x\in X.
\]
Thus, we have either (i) $u(y,x)>0$ for all $x\in X$, so, taking into account continuity of $u(y,x)$ and $u_y(y,x)$ in $x$,  
\[h(y)=\min_{x\in X}\left\{-\frac{u_y(y,x)}{u(y,x)}\right\}-1\in \R\] 
satisfies \eqref{e:gamma}; or (ii) $u(y,x)<0$ for all $x\in X$, so  
\[h(y) =\max_{x\in X}\left\{-\frac{u_y(y,x)}{u(y,x)}\right\}+1\in \R\] 
satisfies \eqref{e:gamma}. 

It remains to show that if for all $y\in Y$ there exists $h(y)\in \R$ satisfying \eqref{e:gamma}, then there exists a continuous function $\tilde h:Y\rightarrow \R$ satisfying \eqref{e:gamma}. Define a correspondence $\varphi:Y\rightrightarrows \R$,
\[
\varphi (y)=\{r\in \R: u_y(y,x)+ru(y,x)<0, \quad  \text{for all } x\in X\}.
\]
Note that $\varphi$ is nonempty valued by assumption, and is clearly convex valued. In addition, $\varphi$ has open lower sections, because for each $r\in \R$ the set 
\[\{y\in Y:u_y(y,x)+r u(y,x)<0,\quad \text{for all }x\in X\}\]
is open, since $u_y$ and $u$ are continuous on the compact set $Y\times X$. Thus, by Browder's Selection Theorem (Theorem 17.63 in \citealt{aliprantis2006}), $\varphi$ admits a continuous selection $\tilde h$, which by construction satisfies \eqref{e:gamma}.
\end{proof}

\subsection{Proof of Lemma \ref{l:S><0}}
We consider the case where $u_{y x}/u_x$ and $V_{yx}/u_x$ are increasing in $x$; the case where $u_{y x}/u_x$ and $V_{yx}/u_x$ are decreasing in $x$ is analogous and thus omitted. 

Fix $x_1<x_2<x_3$ and $y$ such that $u(y,x_1)<0<u(y,x_3)$. The inequality $|S|>0$ follows from the following displayed equations: 
\[
u(y,x_3)-u(y,x_1)=\int_{x_1}^{x_3} u_x (y,x)\df x>0,
\]
where the inequality holds by Assumption \ref{a:ord};
\[
\begin{vmatrix}
	u(y,x_1) &u(y,x_3)\\
	u_{y}(y,x_1) &u_{y}(y,x_3)
\end{vmatrix}=-u(y,x_3)u_{y}(y,x_1) + u(y,x_1)u_{y}(y,x_3)>0,
\]
where the inequality holds by part (2) of Lemma \ref{l:ASC};
\[
\begin{vmatrix}
	V_y(y,x_1) & V_y(y,x_3)\\
	u(y,x_1) & u(y,x_3)
\end{vmatrix} 
=u(y,x_3)V_y(y,x_1) -u(y,x_1)V_y(y,x_3)>0,
\]
where the inequality holds by Assumption \ref{a:ord}; 
\begin{gather*}
-\begin{vmatrix}
V_y(y,x_2)-V_y(y,x_1) &V_y(y,x_3)-V_y(y,x_2)\\
u(y,x_2)-u(y,x_1) &u(y,x_3)-u(y,x_2)
\end{vmatrix}\\
=(V_y(y,x_3)-V_y(y,x_2))(u(y,x_2)-u(y,x_1))-(V_y(y,x_2)-V_y(y,x_1))(u(y,x_3)-u(y,x_2))\\
=\int_{x_2}^{x_3}\int_{x_1}^{x_2} (V_{yx}(y,\tilde x)u_x(y, x) -V_{yx}(y,x)u_x(y,\tilde x))\df x \df \tilde x \geq (>) 0,
\end{gather*}
where the inequality holds by Assumption \ref{a:ord} and (strict) monotonicity of $V_{yx} /u_x$ in $x$;
\begin{gather*}
\begin{vmatrix}
u(y,x_2)-u(y,x_1) &u(y,x_3)-u(y,x_2)\\
u_y(y,x_2)-u_y(y,x_1) &u_y(y,x_3)-u_y(y,x_2)
\end{vmatrix}\\
=(u(y,x_2)-u(y,x_1))(u_y(y,x_3)-u_y(y,x_2))-(u(y,x_3)-u(y,x_2))(u_y(y,x_2)-u_y(y,x_1))\\
=\int_{x_2}^{x_3}\int_{x_1}^{x_2} (u_x(y,x) u_{yx}(y,\tilde x) -u_x(y,\tilde x)u_{yx}(y,x)) \df x \df \tilde x \geq(>)0,
\end{gather*}
where the inequality holds by Assumption \ref{a:ord} and (strict) monotonicity of $u_{yx}/u_{x}$ in $x$;
\begin{align*}
&\frac{
\begin{vmatrix}
V_y(y,x_1) &V_y(y,x_2) &V_y(y,x_3)\\
u(y,x_1) &u(y,x_2) &u(y,x_3) \\
u_y(y,x_1) &u_y(y,x_2) &u_y(y,x_3) 
\end{vmatrix}
}
{
\begin{vmatrix}
	u(y,x_1) &u(y,x_3)\\
	u_{y}(y,x_1) &u_{y}(y,x_3)
\end{vmatrix}
}
(u(y,x_3)-u(y,x_1))\\
=-&\begin{vmatrix}
V_y(y,x_2)-V_y(y,x_1) &V_y(y,x_3)-V_y(y,x_2)\\
u(y,x_2)-u(y,x_1) &u(y,x_3)-u(y,x_2)
\end{vmatrix}\\
+&
\frac{
\begin{vmatrix}
	V_y(y,x_1) & V_y(y,x_3)\\
	u(y,x_1) & u(y,x_3)
\end{vmatrix} 
}
{
\begin{vmatrix}
	u(y,x_1) &u(y,x_3)\\
	u_y(y,x_1) &u_y(y,x_3)
\end{vmatrix}
}
\begin{vmatrix}
u(y,x_2)-u(y,x_1) &u(y,x_3)-u(y,x_2)\\
u_y(y,x_2)-u_y(y,x_1) &u_y(y,x_3)-u_y(y,x_2)
\end{vmatrix},
\end{align*}
where the equality holds by rearrangement.

\subsection{Proof of Lemma \ref{l:R>0}}
We consider the case where $u_{yx}/u_{x}$ and $V_{yx}/u_x$ are increasing in $x$; the case where $u_{y x}/u_x$ and $V_{yx}/u_x$ are decreasing in $x$ is analogous and thus omitted. 

Fix $x_1<x_2<x_3$ and $y_2>y_1$ such that $u(y_1,x_1)<0<u(y_1,x_2)$. The inequality $|R|>0$ follows from the following displayed equations: 
\[
u(y_1,x_3)-u(y_1,x_1)=\int_{x_1}^{x_3} u_x (y_1,x)\df x>0,
\]
where the inequality holds by Assumption \ref{a:ord};
\begin{gather*}
\begin{vmatrix}
	u(y_1,x_1) &u(y_1,x_3)\\
	u(y_2,x_1) &u(y_2,x_3)
\end{vmatrix}\\
=-u(y_1,x_3)u(y_2,x_1) + u(y_1,x_1)u(y_2,x_3)\\
=-g(y_1)\tilde u(y_1,x_3)g(y_2)\tilde u(y_2,x_1) + g(y_1)\tilde u(y_1,x_1)g(y_2)\tilde u(y_2,x_3)\\
=g(y_1)g(y_2)[-\tilde u(y_1,x_3)(\tilde u(y_2,x_1)-\tilde u(y_1,x_1))+\tilde u(y_1,x_1)(\tilde u(y_2,x_3)-\tilde u(y_1,x_3))]\\
=g(y_1)g(y_2)\int_{y_1}^{y_2}[-\tilde u(y_1,x_3) \tilde u_y(y,x_1)+\tilde u(y_1,x_1)\tilde u_y (y,x_3) ]\df y>0,
\end{gather*}
where  the inequality and the second equality hold by parts (2) and (3) of Lemma \ref{l:ASC};
\begin{gather*}
\begin{vmatrix}
	V(y_2,x_1)-V(y_1,x_1) & V(y_2,x_3)-V(y_1,x_3)\\
	u(y_1,x_1) & u(y_1,x_3)
\end{vmatrix} \\
=u(y_1,x_3)\int_{y_1}^{y_2} V_y(y,x_1)\df y -u(y_1,x_1)\int_{y_1}^{y_2} V_y(y,x_3)\df y>0,	
\end{gather*}
where the inequality holds by Assumption \ref{a:ord};
\begin{gather*}
-\begin{vmatrix}
V(y_2,x_2)-V(y_1,x_2)-V(y_2,x_1)+V(y_1,x_1) &V(y_2,x_3)-V(y_1,x_3)-V(y_2,x_2)+V(y_1,x_2)\\
u(y_1,x_2)-u(y_1,x_1) &u(y_1,x_3)-u(y_1,x_2)
\end{vmatrix}\\
=(V(y_2,x_3)-V(y_1,x_3)-V(y_2,x_2)+V(y_1,x_2))(u(y_1,x_2)-u(y_1,x_1))\\
-(V(y_2,x_2)-V(y_1,x_2)-V(y_2,x_1)+V(y_1,x_1))(u(y_1,x_3)-u(y_1,x_2))\\
=\int_{y_1}^{y_2}\int_{x_2}^{x_3}\int_{x_1}^{x_2} (V_{yx}(y,\tilde x)u_x(y_1, x) -V_{yx}(y,x)u_x(y_1,\tilde x))\df x \df \tilde x \df y \geq (>) 0,
\end{gather*}
where the inequality holds by Assumption \ref{a:ord} and (strict) monotonicity of $V_{yx} /u_x$ in $x$;
\begin{gather*}
\begin{vmatrix}
u(y_1,x_2)-u(y_1,x_1) &u(y_1,x_3)-u(y_1,x_2)\\
u(y_2,x_2)-u(y_2,x_1) &u(y_2,x_3)-u(y_2,x_2)
\end{vmatrix}\\
=(u(y_1,x_2)-u(y_1,x_1))(u(y_2,x_3)-u(y_2,x_2))-(u(y_1,x_3)-u(y_1,x_2))(u(y_2,x_2)-u(y_2,x_1))\\
= \int_{x_2}^{x_3}\int_{x_1}^{x_2}  (u_x(y_1,x)u_x(y_2,\tilde x)-u_x(y_1,\tilde x)u_x(y_2,x) ) \df x \df \tilde x  \geq(>)0,
\end{gather*}
where the inequality holds by Assumption \ref{a:ord} and (strict) monotonicity of $u_{yx}/u_{x}$ in $x$, which imply that, for $y_2>y_1$ and $\tilde x>x$, we have
\begin{align*}
\ln \frac{u_x(y_1,x)u_x(y_2,\tilde x)}{u_x(y_1,\tilde x)u_x(y_2,x)} 
&= \int_{y_1}^{y_2}\frac{\partial }{\partial y} [\ln u_x (y,\tilde x)-\ln u_x (y,x)]\df y 
=\int_{y_1}^{y_2} \left[
\frac{u_{yx}(y,\tilde x)}{u_x(y,\tilde x)}-\frac{u_{yx}(y,x)}{u_x(y,x)} \right]\df y\geq (>)0;
\end{align*}
\begin{align*}
&\frac{
\begin{vmatrix}
V(y_2,x_1)-V(y_1,x_1) &-(V(y_2,x_2)-V(y_1,x_2)) &V(y_2,x_3)-V(y_1,x_3)\\
-u(y_1,x_1) &u(y_1,x_2) &-u(y_1,x_3) \\
u(y_2,x_1) &-u(y_2,x_2) &u(y_2,x_3)
\end{vmatrix}
}
{
\begin{vmatrix}
	u(y_1,x_1) &u(y_1,x_3)\\
	u(y_2,x_1) &u(y_2,x_3)
\end{vmatrix}
}
(u(y_1,x_3)-u(y_1,x_1))\\
=-&\begin{vmatrix}
V(y_2,x_2)-V(y_1,x_2)-V(y_2,x_1)+V(y_1,x_1) &V(y_2,x_3)-V(y_1,x_3)-V(y_2,x_2)+V(y_1,x_2)\\
u(y_1,x_2)-u(y_1,x_1) &u(y_1,x_3)-u(y_1,x_2)
\end{vmatrix}\\
+&
\frac{
\begin{vmatrix}
	V(y_2,x_1)-V(y_1,x_1) & V(y_2,x_3)-V(y_1,x_3)\\
	u(y_1,x_1) & u(y_1,x_3)
\end{vmatrix}
}
{
\begin{vmatrix}
	u(y_1,x_1) &u(y_1,x_3)\\
	u(y_2,x_1) &u(y_2,x_3)
\end{vmatrix}
}
\begin{vmatrix}
u(y_1,x_2)-u(y_1,x_1) &u(y_1,x_3)-u(y_1,x_2)\\
u(y_2,x_2)-u(y_2,x_1) &u(y_2,x_3)-u(y_2,x_2)
\end{vmatrix},
\end{align*}
where the equality holds by rearrangement.

\subsection{Proof of Lemma \ref{l:ySDD}}
Fix $x_1<x_2<x_3$ and $y_2>y_1$ such that $u(y_1,x_1)<0<u(y_1,x_3)$. The first claimed inequality follows as in the proof of Lemma \ref{l:R>0}, by Assumption \ref{a:qc} and $u(y_1,x_1)<0<u(y_1,x_3)$. We thus focus on the second and third inequalities.

As in the proof of Lemma \ref{l:R>0}, Assumption \ref{a:ord} and monotonicity of $u_{yx}/u_{x}$ in $x$ yield
\begin{gather*}
u(y_1,x_3)>u(y_1,x_2)>u(y_1,x_1),\\
\frac{u(y_2,x_3)-u(y_2,x_2)}{u(y_1,x_3)-u(y_1,x_2)}\geq\frac{u(y_2,x_2)-u(y_2,x_1)}{u(y_1,x_2)-u(y_1,x_1)}.
\end{gather*}
There are three cases to consider.

(1) $u(y_1,x_2)=0$. In this case, $u(y_2,x_2)<0$, by Assumption \ref{a:qc}. Thus,
\begin{align*}
u(y_2,x_2)u(y_1,x_1)&>0=u(y_2,x_1)u(y_1,x_2),\\
u(y_2,x_3)u(y_1,x_2)&=0> u(y_2,x_2)u(y_1,x_3).
\end{align*}

(2) $u(y_1,x_2)>0$. In this case, as follows from the proof of Lemma \ref{l:R>0},
\[
u(y_2,x_2)u(y_1,x_1)>u(y_2,x_1)u(y_1,x_2),
\]
by Assumption \ref{a:qc} and $u(y_1,x_1)<0<u(y_1,x_2)$. Thus,
\begin{gather*}
\frac{u(y_2,x_3)-u(y_2,x_2)}{u(y_1,x_3)-u(y_1,x_2)}\geq\frac{u(y_2,x_2)-u(y_2,x_1)}{u(y_1,x_2)-u(y_1,x_1)}>\frac{u(y_2,x_2)}{u(y_1,x_2)}\\
\implies u(y_2,x_3)u(y_1,x_2)> u(y_2,x_2)u(y_1,x_3).
\end{gather*}

(3) $u(y_1,x_2)<0$. In this case, as follows from the proof of Lemma \ref{l:R>0},
\[
u(y_2,x_3)u(y_1,x_2)> u(y_2,x_2)u(y_1,x_3),
\]
by Assumption \ref{a:qc} and $u(y_1,x_2)<0<u(y_1,x_3)$. Thus,
\begin{gather*}
\frac{u(y_2,x_2)}{u(y_1,x_2)}>\frac{u(y_2,x_3)-u(y_2,x_2)}{u(y_1,x_3)-u(y_1,x_2)}\geq\frac{u(y_2,x_2)-u(y_2,x_1)}{u(y_1,x_2)-u(y_1,x_1)}\\
\implies u(y_2,x_2)u(y_1,x_1)>u(y_2,x_1)u(y_1,x_2).\qedhere
\end{gather*}

\subsection{Proof of Lemma \ref{l:ySPD}}
Fix $x_1<x_2<x_3$ and $y_2<y_1$ such that $x_1< \chi(y_1)< x_3$. As in the proof of Lemma \ref{l:R>0}, Assumption \ref{a:ord} and monotonicity of $V_{yx} /u_x$ in $x$ yield
\begin{gather}
V(y_1,x_j)-V(y_2,x_j)>0\quad \text{for $j=1,2,3$},\label{e:dV>0}\\
u(y_1,x_3)>u(y_1,x_2)>u(y_1,x_1),\label{e:du>0}\\
\begin{gathered}
\frac{V(y_1,x_3)-V(y_2,x_3)-V(y_1,x_2)+V(y_2,x_2)}{u(y_1,x_3)-u(y_1,x_2)}\\ \leq\frac{V(y_1,x_2)-V(y_2,x_2)-V(y_1,x_1)+V(y_2,x_1)}{u(y_1,x_2)-u(y_1,x_1)}.\label{d:dV/du}
\end{gathered}
\end{gather}

There are two cases to consider.

(1) $u(y_1,x_2)\geq 0$. In this case,  we have
\[
\frac{u(y_1,x_1)}{V(y_1,x_1)-V(y_2,x_1)}<\frac{u(y_1,x_2)}{V(y_1,x_2)-V(y_2,x_2)},
\]
by \eqref{e:dV>0} and $u(y_1,x_1)<0\leq u(y_1,x_2)$, and
\[
\frac{u(y_1,x_2)}{V(y_1,x_2)-V(y_2,x_2)}<\frac{u(y_1,x_3)}{V(y_1,x_3)-V(y_2,x_3)},
\]
by
\begin{align*}
u(y_1,x_2)(V(y_1,x_3)-V(y_2,x_3))
\leq &u(y_1,x_2) \frac{u(y_1,x_3)-u(y_1,x_1)}{u(y_1,x_2)-u(y_1,x_1)}(V(y_1,x_2)-V(y_2,x_2))  \\
<&u(y_1,x_3)(V(y_1,x_2)-V(y_2,x_2)),
\end{align*}
where the first inequality holds by \eqref{d:dV/du}, $V(y_1,x_1)>V(y_2,x_1)$, $u(y_1,x_3)>u(y_1,x_2)$, and $u(y_1,x_2)\geq 0$, and the second inequality holds by $V(y_1,x_2)>V(y_2,x_2)$, $u(y_1,x_3)>u(y_1,x_2)$, and $u(y_1,x_1)<0$.

(2) $u(y_1,x_2)\leq 0$. In this case, we have  
\[
\frac{u(y_1,x_2)}{V(y_1,x_2)-V(y_2,x_2)}<\frac{u(y_1,x_3)}{V(y_1,x_3)-V(y_2,x_3)},
\]
by \eqref{e:dV>0} and $u(y_1,x_2)\leq 0< u(y_1,x_3)$, and
\[
\frac{u(y_1,x_1)}{V(y_1,x_1)-V(y_2,x_1)}<\frac{u(y_1,x_2)}{V(y_1,x_2)-V(y_2,x_2)},
\]
by
\begin{align*}
 -u(y_1,x_2)(V(y_1,x_1)-V(y_2,x_1))
\leq & -u(y_1,x_2)\frac{u(y_1,x_3)-u(y_1,x_1)}{u(y_1,x_3)-u(y_1,x_2)}(V(y_1,x_2)-V(y_2,x_2)) \\
<&-u(y_1,x_1)(V(y_1,x_2)-V(y_2,x_2)),
\end{align*}
where the first inequality holds by \eqref{d:dV/du}, $V(y_1,x_3)>V(y_2,x_3)$, $u(y_1,x_3)>u(y_1,x_2)$, and $u(y_1,x_2)\leq 0$, and the second inequality holds by $V(y_1,x_2)>V(y_2,x_2)$, $u(y_1,x_2)>u(y_1,x_1)$, and $u(y_1,x_3)>0$.

\subsection{Proof of Lemma \ref{l:stab}}
We give the proof for the single-dipped case. The proof remains valid if Assumption \ref{a:ord} is replaced with strict single-crossing of $u(y,x)$ in $x$. Let  $\tau^n$ be any optimal signal under $V^n$, so that $\supp (\tau^n)\subset \Lambda^n$. Since the set of compact subsets of a compact set is compact (in the Hausdorff topology), taking a subsequence if necessary, $\Lambda^n$ converges to some compact set $\ol\Lambda\subset \Delta(X)$. Since the set of signals is compact (in the weak* topology), taking a subsequence if necessary, $\tau^n$ converges weakly to some signal $ \tau$. Finally, since $\Lambda^n\rightarrow \ol\Lambda$, $\tau^n\rightarrow \tau$, and $\supp( \tau^n)\subset \Lambda^n$, it follows that $\supp (\tau) \subset \ol \Lambda$, by  Box 1.13 in \citet{santambrogio}. 

We claim that $\tau$ is optimal under $V$. Since $V_y^n$ converges uniformly to $V_y$, for each $\delta>0$ there exists $n_\delta\in \N$ such that, for all $n\geq n_\delta$, we have $|V_y^n(y,x)-V_y(y,x)|\leq \delta$ for all $(y,x)$. Since $\tau^n$ is optimal under $V^n$, for each signal $\tilde\tau$ we have
\begin{align*}
\int_{\Delta(X)}\int_X \int_0^y V_y (\tilde y,x)\df \tilde y\df \mu(x)\df \tau ^n(\mu) &\geq \int_{\Delta(X)}\int_X  \int_0^y V_y^n (\tilde y,x)\df \tilde y\df \mu(x)\df \tau ^n(\mu)-\delta\\
&\geq \int_{\Delta(X)}\int_X \int_0^y V_y^n (\tilde y,x)\df \tilde y\df \mu(x)\df \tilde \tau(\mu)	-\delta\\
&\geq \int_{\Delta(X)}\int_X \int_0^y V_y(\tilde y,x)\df \tilde y\df \mu(x)\df \tilde \tau(\mu)-2\delta.
\end{align*}
Passing to the limit as $\delta\rightarrow 0$ and $n\rightarrow \infty$ establishes the optimality of $\tau$  under $V$.	

Suppose by contradiction that $\ol \Lambda$ is not single-dipped. Then there exist $\mu_1,\mu_2\in \ol \Lambda$ and $x_1<x_2<x_3$ such that $x_1,x_3\in \supp(\mu_1)$, $x_2\in \supp (\mu_2)$, and $\gamma(\mu_1)<\gamma(\mu_2)$. Since $\Lambda^n\rightarrow \ol \Lambda$, there exist $\mu_1^n,\mu_2^n\in \Lambda^n$ such that $\mu_1^n\rightarrow \mu_1$ and $\mu_2^n\rightarrow \mu_2$. Since $\gamma(\mu)$ is continuous in $\mu$ and since the support correspondence is lower hemicontinuous, by Theorem 17.14 in \citet{aliprantis2006}, it follows that there exists $n$, $\mu_1^n,\mu_2^n\in \Lambda^n$, and $x_1^n<x_2^n<x_3^n$ such that $x_1^n,x_3^n\in \supp(\mu_1^n)$, $x_2^n\in \supp (\mu_2^n)$, and $\gamma(\mu_1^n)<\gamma(\mu_2^n)$, contradicting that $\Lambda^n$ is single-dipped.

\subsection{Proof for Example \ref{e:quantile}} \label{s:quantileproof}
First, notice that the outcome $\pi$ that corresponds to the proposed signal is implementable. \eqref{PO1} holds because, for all $y\in [\ul y,1]$, the marginal distribution over actions satisfies
\[
\alpha_{\pi}([a,1])=\phi([0,\chi_1(y)])+\phi ([a,1]),
\]
and the posterior conditional on $y$ is
\begin{gather*}
\pi_y=\frac{\df\phi ([0,\chi_1(y)])}{\df\phi ([0,\chi_1(y)]+\df\phi ([a,1])}\delta_{\chi_1(y)}+\frac{\df\phi ([a,1])}{\df\phi ([0,\chi_1(y)]+\df\phi ([a,1])}\delta_{\chi_1(y)},
\end{gather*}
as follows from $\kappa \phi ([0,\chi_1(y)])=(1-\kappa)\phi ([a,1])$, which implies that $\kappa \df\phi ([0,\chi_1(y)])=(1-\kappa)\df\phi ([a,1])$ and that $\chi_1$ is a continuous, strictly decreasing function. \eqref{PO2} holds because, for all $y\in [\ul y,1]$, 
\[
\E_{\pi_y}[u(y,x)]=\E_{\pi_y}[\1 \{x\geq y\}-\kappa]=\pi_y([y,1])-\kappa=0.
\]
Consider now any other implementable outcome $\tilde\pi$. By \eqref{PO2}, there exists $\tilde\pi_y $ with $\tilde\pi_y([y,1])\geq \kappa$, as otherwise $\E_{\tilde \pi_y}[u(y,x)]<0$. Thus, by \eqref{PO1}, $\alpha_{\tilde \pi} ([y,1])\leq \phi([y,1])/\kappa$, as follows from
\[
\phi([y,1])=\int_Y \tilde \pi_{\tilde y}([y,1])\df \alpha_{\tilde \pi} (\tilde y)\geq\int_{y}^1 \tilde \pi_{\tilde y}([y,1])\df \alpha_{\tilde \pi} (\tilde y)  \geq \kappa \alpha_{\tilde \pi} ([y,1]).
\]
Since $\alpha_\pi([y,1])=\phi([y,1])/\kappa$, it follows that $\alpha_{\pi}$ first-order stochastically dominates $\alpha_{\tilde \pi}$, and thus, for an increasing $V$, 
\[
\int_{Y\times X} V(y)\df \pi(y,x)=\int_{Y} V(y)\df \alpha_{\pi}(y)\geq \int_{Y} V(y)\df \alpha_{\tilde \pi}(y)=\int_{Y\times X} V(y)\df \tilde \pi(y,x),
\]
showing that $\pi$ is optimal.

\subsection{Proof for Example \ref{ex:segpair}} \label{s:segpair}
We will show that $\Lambda=\{\delta_{\chi_1(y)}/2+\delta_{\chi_2(y)}/2:y\in [-1,1]\}$. Then, by Theorem \ref{t:brenier}, there is a unique optimal signal. Consider a signal $\tau$ that induces the distribution over actions $\alpha$ and the only posterior inducing each action $y\in \supp(\alpha)$ is $\mu = \delta_{\chi_1(y)}/2+\delta_{\chi_2(y)}/2$. By construction, $\supp (\tau)\subset \Lambda$. Moreover,  $\int \mu \df \tau=\phi$, because, for each $y\in [0,1]$,
\begin{gather*}
\phi ([\chi_2(y),3])=\phi([3y,3])= \tfrac{1}{2}\alpha ([y,1]),\\
\phi ([-1,\chi_1(-y)])=\phi ([-1,-y])=\alpha ([-1,-y])+\tfrac{1}{2}\alpha ([y,1]).
\end{gather*}
Hence $\tau$ is optimal. Finally, the following lemma shows that $\Lambda$ is as stated.
\begin{lemma}\label{l:gerry}
Functions
\[
p(x)=
\begin{cases}
T(2x), &x\in [-1,0),\\
3 T(\frac 23 x), &x\in [0,3],
\end{cases}
\quad\text{and}\quad 
q(y)=
\begin{cases}
\frac{2T'(2y)}{T'(0)}, &y\in [-1, 0),\\
2, &y\in [0,3].
\end{cases}
\]
satisfy \eqref{DO1} with equality if $y\in [-1,1]$ and $x\in \{\chi_1(y),\chi_2(y)\}$, and strict inequality otherwise.
\end{lemma}
\begin{proof}[Proof of Lemma \ref{l:gerry}]
Since $T$ is symmetric about 0 (i.e., $T(x-y)=-T(y-x)$) and $T'$ is strictly log-concave, it follows that $T'(0)>T'(z)$ for all $z\neq 0$ and $T(z)$ is strictly concave for $z\geq 0$. Hence, if $z_1'\leq z_1\leq z_2\leq z_2'$, $(z_1',z_2')\neq (z_1,z_2)$, and $\rho'z_1'+(1-\rho')z_2'=\rho z_1+(1-\rho)z_2$, for some $z_1,z_2,z_1',z_2'\geq 0$ and $\rho,\rho'\in (0,1)$, then $\rho'T(z_1')+(1-\rho)T(z_2')<\rho T(z_1)+(1-\rho)T(z_2)$, by Jensen's inequality.

We split the analysis into six cases.

(1) For $y\in [0,3]$ and $x\in [y,3]$, \eqref{DO1} simplifies to
\[
3T(\tfrac 23 x) \geq T(2y) + 2T(x-y),
\]
which holds with equality for $x=3y=\chi_2(a)$ and strict inequality for $x\neq 3y$.

(2) For $y\in [0,3]$ and $x\in (0,y)$, \eqref{DO1} simplifies to
\[
3T(\tfrac 23 x)+2T(y-x) \geq T(2y) + 4T(0),
\]
which always holds with strict inequality.

(3) For $y\in [0,3]$ and $x \in [-1,0]$, \eqref{DO1} simplifies to
\[
2T(y-x) \geq  T(2y)+T(-2x),
\]
which holds with equality for $x=-y=\chi_1(y)$ and strict inequality for $x\neq -y$.

(4) For $y\in [-1,0)$ and $x\in [0,3]$, \eqref{DO1} simplifies to
\[
3T(\tfrac 23 x)+ T(-2y) \geq q (y)T(x-y)+2T(0),
\]
which always holds with strict inequality because $q(y)<2$ and $T(x-y)>0$.

(5) For $y\in [-1,0)$ and $x\in (y,0)$, \eqref{DO1} simplifies to
\[
T(-2y)  \geq T(-2x)+ q (y)T(x-y),
\]
which is equivalent to 
\[
\frac{T(-2y)-T(-2x)}{T'(-2y) (-2y+2x)}\geq \frac{T(x-y)-T(0)}{T'(0)(x-y)},
\]
which always holds with strict inequality because $T(z)$ is strictly concave for $z\geq 0$, and thus the left-hand side is strictly greater than 1 whereas the right-hand side is strictly less than 1.

(6) For $y\in [-1,0)$ and $x\in [-1,y]$, \eqref{DO1} simplifies to
\[
T(-2y) +q (y)T(y-x) \geq T(-2x),
\]
which holds with equality for $x=y=\chi_1(y)$. For $x<y$, the inequality is equivalent to
\[
\frac{T(y-x)-T(0)}{T'(0)(y-x)}\geq \frac{T(-2x)-T(-2y)}{T'(-2y)(-2x+2y)},
\]
which always holds with strict inequality because
\begin{align*}
\frac{T(-2x)-T(-2y)}{T'(-2y)(2y-2x)}&=\frac{1}{2y-2x}\int_{0}^{2(y-x)} \frac{T'(z-2y)}{T'(-2y)}\df z\\
&<\frac{1}{2y-2x}\int_{0}^{2(y-x)} \frac{T'(z)}{T'(0)}\df z\\
& =\frac{T(2y-2x)-T(0)}{T'(0)(2y-2x)}\\
&<\frac{T(y-x)-T(0)}{T'(0)(y-x)},
\end{align*}
where the first inequality holds because $T'$ is strictly log-concave, and the second inequality holds because $T(z)$ is strictly concave for $z\geq 0$.
\end{proof}

\subsection{Proof of Proposition \ref{p:ZZ}}
Recall that most results remain valid if the condition $u_x (y,x)>0$ in Assumption \ref{a:ord} is replaced with strict single-crossing of $u(y,x)$ in $x$.
Clearly, $\gamma(\mu)={\E_\mu [x]}/({1+\E_\mu[x^2]})$. To ensure that Assumption \ref{a:int} holds, we normalize $Y=[\min_{x\in [\ul x,\ol x]} \gamma(\delta_x),\max_{x\in [\ul x,\ol x]} \gamma(\delta_x)]$. Assumptions \ref{a:smooth} and \ref{a:qc} obviously hold. Moreover, since $\gamma(\delta_x)$ is strictly increasing on $[0,1]$ and strictly decreasing on $[1,+\infty)$, it follows that $u(\gamma(\delta_x),x')>0$ if $x<x'\leq 1$ and if $1\leq x'<x$. Thus, if $\ol x\leq 1$, then $u(y,x)$ satisfies strict single-crossing in $x$, whereas, if $\ul x \geq 1$, $u(y,x)$ also satisfies strict single-crossing in $x$ once the state is redefined as $-x$. So Theorems \ref{t:pairwise}, \ref{l:ssdd}, \ref{p:full}, and \ref{t:NAD} apply.

Lemma \ref{l:ZZ1} replicates Lemma 1 and Proposition 3 in \citet{ZZ}.
\begin{lemma}\label{l:ZZ1}
If $x_1<x_2$  and $x_1x_2> (<)1$, then  $\rho V(\gamma(\delta_{x_1}),x_1)+(1-\rho )V(\gamma(\delta_{x_2}),x_2)>(<)\rho V(\gamma(\mu),x_1)+(1-\rho )V(\gamma(\mu),x_2)$ for all $\rho\in (0,1)$.
\end{lemma}
\begin{proof} 
For $\mu=\rho \delta_{x_1}+(1-\rho )\delta_{x_2}$, $\gamma(\mu)={(\rho x_1 +(1-\rho )x_2)}/{(1+\rho x_1^2+(1-\rho )x_2^2)}$. 
Thus, if $x_1<x_2$ and $x_1x_2>(<)1$, we have
\begin{gather*}
\frac{\df}{\df \rho }\gamma(\mu) =  \frac{(x_2-x_1)(x_1x_2-1)}{(1+\rho x_1^2+(1-\rho )x_2^2)^2}>(<)0,\\
\frac{\df^2}{\df \rho^2 }\gamma(\mu)  =  \frac{(x_2-x_1)(x_1x_2-1)(x_2^2-x_1^2)}{(1+\rho x_1^2+(1-\rho )x_2^2)^3}>(<)0.
\end{gather*}
Define $\varphi (\rho)=\gamma(\mu)\left({\rho }/{x_1}+{(1-\rho )}/{x_2}\right) $.
Thus, if $x_1<x_2$ and $x_1x_2>(<)1$, we have
\[
\varphi''(\rho )={\left(\frac{\rho }{x_1}+\frac{1-\rho }{x_2}\right)}{\frac{\df^2}{\df \rho^2 }\gamma(\mu)} +2{\left(\frac{1}{x_1}-\frac{1}{x_2}\right)}\frac{\df}{\df \rho }\gamma(\mu)>(<)0,
\]
so $\varphi$ is strictly convex (concave), and $\rho\varphi(1)+(1-\rho) \varphi(0)>(<)\varphi (\rho)$.
\end{proof}
If $\ul x\geq 1$, then $x_1x_2>1$ for all $\ul x_1\leq x_1<x_2$, so full disclosure is uniquely optimal by Theorem \ref{p:full} and Lemma \ref{l:ZZ1}. Assume henceforth that $\ul x\leq 1$. 

After some algebra,  we get, for all $y$ and $x_1<x_2<x_3$,
\[
|S|=\frac{(x_3-x_2)(x_3-x_1)(x_2-x_1)(1-x_2x_3-x_1x_3-x_1x_2)}{x_1x_2x_3}
\]
If $\ol x\leq 1/\sqrt 3$ ($ \ul x\geq 1/\sqrt 3$), then $|S|>(<)0$ for all $x_1<x_2<x_3\leq \ol x$  ($\ul x\leq x_1<x_2<x_3$), so $\Lambda$ is pairwise by Theorem \ref{t:pairwise}. Proposition 4 in \citet{ZZ} derives a version of this result for a finite set $X$.

Moreover, if $\ol x \leq 1/\sqrt 3$ ($\ul x \geq 1/\sqrt 3$), then $\Lambda$ is single-dipped (-peaked), as follows from Theorem \ref{l:ssdd} with
\begin{gather*}
\beta=
\begin{pmatrix}
u(y_2,x_3)u(y_1,x_2)-u(y_2,x_2)u(y_1,x_3) \\
u(y_2,x_3)u(y_1,x_1)-u(y_2,x_1)u(y_1,x_3) \\
u(y_2,x_2)u(y_1,x_1)-u(y_2,x_1)u(y_1,x_2)
\end{pmatrix}\\
\begin{pmatrix}
\beta=-
\begin{pmatrix}
u(y_2,x_3)u(y_1,x_2)-u(y_2,x_2)u(y_1,x_3) \\
u(y_2,x_3)u(y_1,x_1)-u(y_2,x_1)u(y_1,x_3) \\
u(y_2,x_2)u(y_1,x_1)-u(y_2,x_1)u(y_1,x_2)
\end{pmatrix}
\end{pmatrix},	
\end{gather*}
because, for $y<y'$ and $x<x'$ with $xx'<1$, we have
\[
u(y',x')u(y,x)-u(y',x)u(y,x')=(y'-y)(x'-x)(1-xx')>0,
\]
and 
\[
R\b=
\begin{pmatrix}
(y_2-y_1)^2|S|\\
0\\
0
\end{pmatrix}
\gneq 
\begin{pmatrix}
0\\
0\\
0
\end{pmatrix}
\begin{pmatrix}
R\b =
\begin{pmatrix}
-(y_2-y_1)^2|S|\\
0\\
0
\end{pmatrix}
\gneq 
\begin{pmatrix}
0\\
0\\
0
\end{pmatrix}
\end{pmatrix}.
\]
Thus $\Lambda$ is strictly single-dipped (-peaked) if $\ol x \leq 1/\sqrt 3$ ($\ul x \geq 1/\sqrt 3$). Finally, since, by Lemma \ref{l:ZZ1}, \eqref{e:nd} holds for all $\rho \in (0,1)$, Theorem \ref{t:NAD} yields that, if $\ol x \leq 1/\sqrt 3$ ($\ul x \geq 1/\sqrt 3$), then the optimal signal is unique and single-dipped (-peaked) negative assortative.

\subsection{Proof of Proposition \ref{c:spd1}} Suppose by contradiction that an optimal outcome assigns positive probability to a strictly single-dipped triple $(y_1,x_1)$, $(y_2,x_2)$, $(y_1,x_3)$, with $x_1 <x_2<x_3$, $y_2 <y_1$, and $x_1\leq x _{0}\leq x_3$. Consider a perturbation that reallocates mass $\beta_1\varepsilon$ on $x_1$ and mass $\beta_3\varepsilon$ on $x_3$ from $y_1$ to $y_2$, while reallocating mass $\beta_2\varepsilon$ on $x_2$ from $y_2$ to $y_1$ where $\varepsilon>0$ is small enough and $\beta=(\b_1,\b_2,\b_3)$ is given by
\[\b=
\begin{cases}
\left(0,\frac{1}{(x_2-x_0)g(y_2|x_2)},\frac{1}{(x_2-x_0)g(y_2|x_3)}\right), &x_2>x_0,\\
\left(0,1,0\right), &x_2=x_0, \\
\left(\frac{1}{(x_0-x_1)g(y_1|x_1)},\frac{1}{(x_0-x_2)g(y_1|x_2)},0\right), &x_2<x_0,
\end{cases}
\]
where $x_1 <x_2<x_3$, $y_2 <y_1$, and $x_1\leq x _{0}\leq x_3$.
 We focus on the case $x _{0}<x_2$, as the other cases are analogous. The above perturbation increases action $y_1$, because, by strict log-submodularity of $g$, 
\begin{equation*}
u(y_1,x_2)y_2-u(y_1,x_3)y_3= \frac{g(y_1|x _2)}{g(y_2 |x_2)}-\frac{g(y_1|x _3)}{g(y_2 |x _3)} >
0.
\end{equation*}%
Moreover, the same perturbation also increases the sender's expected utility for fixed $y_1 ,y_2$. This follows because
\begin{align*}
& (V(y_1,x_2)-V(y_2,x_2))y_2-(V(y_1,x_3)-V(y_2,x_3))y_3 \\
=&  \left( \frac{G(y_1|x_2)-G(y_2|x_2)}{(x_2-x _{0})g(y_2 |x_2)}-\frac{G(y_1|x_3)-G(y_2|x_3)}{(x _3-x _{0})g(y_2 |x_3)}\right)  \\
>& \frac{1}{(x_2-x _{0})}\left( \frac{G(y_1|x_2)-G(y_2 |x_2)}{g(y_2 |x_2)}-\frac{G(y_1|x_3)-G(y_2|x_3)}{g(y_2 |x_3)}\right)  \\
=& \frac{1 }{(x _2-x _{0})}\int_{y_2}^{y_1}\left( \frac{g(t|x_2)}{g(y_2|x_2)}-\frac{g(t|x_3)}{g(y_2 |x_3)}\right) \df t\geq 0,
\end{align*}
where the first inequality is by $x _{0}<x_2<x_3$ and the second inequality is by log-submodularity of $g$. Thus, this perturbation is strictly profitable for the sender, so every optimal outcome is single-peaked.

\subsection{Proof of Proposition \ref{p:GL}}
As shown by \citet{KG}, there exists an optimal outcome with a finite support. Suppose the support contains a strictly single-peaked triple $(y_1,x_1)$, $(y_2,x_2)$, $(y_1,x_3)$, with $x_1<x_2<x_3$, $y_1<a_2$, and $x_1<a_1<x_3$. Notice that $V(y_1,x_3)\neq-\infty$ (so $y_1\geq \sigma(x_3))$, as otherwise the sender's expected utility would be $-\infty$, which cannot be optimal. Taking into account that $\sigma(x)=x$ for $x\leq x_0$ gives $y_1 >x_0$. Thus, the first row in $R$ is zero. Consider a perturbation that shifts weights $\b_1=(x_3-x_2)\varepsilon$ and $\b_3=(x_2-x_1)\varepsilon$ on $x_1$ and $x_3$ from $y_1$ to $y_2$ and shifts weight $\b_2=(x_3-x_1)\varepsilon$ from $y_2$ to $y_1$, where $\varepsilon$ takes the maximum value such that $\b_1\leq \pi(\{(y_1,x_1\})$, $\b_2\leq \pi(\{(y_2,x_2\})$, $\b_3\leq \pi(\{(y_1,x_3\})$, so that a strictly single-peaked triple is removed. This perturbation holds fixed $y_1$ and $y_2$ and thus does not change the sender's expected utility, since the first row in $R$ is zero. Repeating such perturbations until all strictly single-peaked triples are removed (a finite number of times since $\supp (\pi)$ is finite) yields a single-dipped outcome that is weakly preferred by the sender.

\end{document}